\documentclass[12pt, a4j]{article}

\usepackage{amsmath, amssymb}
\usepackage{mathtools}
\usepackage{type1cm}
\usepackage{accents}
\usepackage{latexsym}
\usepackage{graphicx} 
\usepackage{algorithm}
\usepackage{algpseudocode}
\usepackage{booktabs, tabularx}
\usepackage{amsthm}
\usepackage{color}
\usepackage[english]{babel}
\usepackage[a4paper, margin=2.5cm, centering]{geometry}
\usepackage{enumerate}
\usepackage{authblk}
\usepackage{url}
\usepackage{adjustbox}
\usepackage{multirow}

\usepackage[pagewise]{lineno}

\algrenewcommand\algorithmicindent{2mm}%

\makeatletter
\algnewcommand{\LineComment}[1]{\Statex \hskip\ALG@thistlm /*  
			      \parbox[t]{\dimexpr\textwidth-2\leftmargin-\labelsep-\labelwidth-2\rightmargin}{ #1 */ \strut}}
\makeatother

\newlength{\continueindent}
\usepackage{etoolbox}
\makeatletter
\newcommand*{\ALG@customparshape}{\parshape 2 \leftmargin \linewidth \dimexpr\ALG@tlm+\continueindent\relax \dimexpr\linewidth+\leftmargin-\ALG@tlm-\continueindent\relax}
\apptocmd{\ALG@beginblock}{\ALG@customparshape}{}{\errmessage{failed to patch}}
\makeatother

\newtheorem{lem}{Lemma}
\newtheorem{theorem}{Theorem}

\newtheorem{observation}[theorem]{Observation}

\def\riso{\underset{r}\approx}
\def\notiso{\not\approx}

\makeatletter
\def\mojiparline#1{
    \newcounter{mpl}
    \setcounter{mpl}{#1}
    \@tempdima=\linewidth
    \advance\@tempdima by-\value{mpl}zw
    \addtocounter{mpl}{-1}
    \divide\@tempdima by \value{mpl}
    \advance\kanjiskip by\@tempdima
    \advance\parindent by\@tempdima
}
\makeatother

\providecommand{\ch}{\mathrm{ch}}
\providecommand{\rank}{\mathrm{rank}}
\providecommand{\lca}{\mathrm{lca}}
\providecommand{\col}{\mathrm{col}}
\providecommand{\val}{\mathrm{val}}
\providecommand{\deg}{\mathrm{deg}}
\providecommand{\mul}{\mathrm{mul}}
\providecommand{\code}{\mathrm{code}}
\providecommand{\pot}{\mathrm{pot}}
\providecommand{\anc}{\mathrm{anc}}

\providecommand{\mm}{\mathrm{M}}
\providecommand{\rev}{\mathrm{rev}}
\providecommand{\res}{\mathrm{res}}
\providecommand{\rootG}[1]{\mathrm{r}(#1, G)}

\providecommand{\pathset}{\mathcal{P}}
\providecommand{\fvset}{\mathbf{f}}
\providecommand{\freq}{\mathrm{frq}}
\providecommand{\Ch}{\mathrm{ch}}

\providecommand{\parent}{\mathrm{p}}

\providecommand{\depth}{\mathrm{d}}

\newcommand{\procMain}{1}
\newcommand{\procTreeSignature}{2}
\newcommand{\procTreeRanking}{3}
\newcommand{\procChildCheck}{4}
\newcommand{\procPotEdges}{5}
\newcommand{\procRestrictedRoot}{6}
\newcommand{\procRestrictedCentroid}{7}

\urldef{\mailsamp}\path|{
		yamashita, 
		r.masui, 
		x.zhou, 
		chenxi, 
		shurbevski, 
		nag
	}@amp.i.kyoto-u.ac.jp|
\urldef{\mailsbio}\path|takutsu@kuicr.kyoto-u.ac.jp|

\begin{document}

 \title{Enumerating Chemical Graphs with Two Disjoint Cycles
	    Satisfying Given Path Frequency Specifications}
	    
 \author[1]{Kyousuke Yamashita}
 \author[1]{Ryuji Masui}
 \author[1]{Xiang Zhou}
 \author[1]{Chenxi Wang}
 \author[1]{Aleksandar Shurbevski}
 \author[1]{Hiroshi Nagamochi}
 \author[2]{Tatsuya Akutsu}
 \affil[1]{Department of Applied Mathematics and Physics, 	  Graduate School of Informatics, 	  Kyoto University, \mailsamp}
 \affil[2]{Bioinformatics Center,             Institute for Chemical Research,             Kyoto University,       \mailsbio}\renewcommand\Authands{ and }

\date{\today}
 \maketitle

\begin{center}
{\bf Abstract}\\
\end{center}
\begin{sloppypar}
\noindent Enumerating chemical graphs satisfying given constraints is 
a fundamental problem in mathematical and computational chemistry,
and plays an essential part in a recently proposed framework
for the inverse QSAR/QSPR.
In this paper, 
constraints are given by feature vectors each of which consists
of the frequencies of paths in a given set of paths.
We consider the problem of enumerating chemical graphs 
that satisfy the path frequency constraints, 
which are given by a pair of feature vectors specifying upper and lower
bounds of the frequency of each path.
We design a branch-and-bound algorithm for enumerating chemical graphs
of bi-block 2-augmented structure, that is
graphs that contain two edge-disjoint cycles.
We present some computational experiments with an
implementation of our proposed algorithm.
\end{sloppypar}

\section{Introduction}\label{sec:introduction}
Among various combinatorial problems in chemistry, 
the enumeration of chemical structures is a fundamental problem and
has a long history going as far back as the work of  Cayley~\cite{cayley1875}.
Applications of the enumeration of chemical structures include 
structure determination using mass-spectrum and/or NMR-spectrum~\cite{BF78-dendral, FS96-chemics}, 
virtual exploration of the chemical universe~\cite{FR07, MS07}, 
reconstruction of molecular structures from their signatures~\cite{FCV03, HDK93}, 
and classification of chemical compounds~\cite{DKWK05}. 
It also plays a key role in 
inverse QSAR/QSPR (quantitative structure-activity
relationship/quantitative 
structure-property relationship)~\cite{FCV03, HDK93, KHLF93, Skvortsova93}, 
which seeks for new chemical compounds having specified activities/properties.
A closely related problem, called the \textit{pre-image problem}, 
has been studied in the field of machine learning
\cite{BZT04, WSB04}, in which
some properties of a desired object are computed as a feature vector
in a feature space, 
and then the feature vector is mapped back to the input space, 
where feature vectors correspond to descriptors in QSAR/QSPR studies and
the object which is mapped back is called a {\em pre-image}. 
 Recently, Azam \textit{et al.}~\cite{ACZSNA20, IWSNA20}, as well as
 Chiewvanichakorn {\it et al.}~\cite{CWZSNA20} used
 graph enumeration
 in a newly developed framework for the inverse QSAR/QSPR
  together with the idea of Akutsu and Nagamochi~\cite{AN19}
 for getting an input vector that achieves a given output value 
 of an artificial neural network
 to search for chemical compounds that are likely to achieve 
 a given value for some chemical property.

Widely used ways of defining feature vectors for graph objects 
are based on the frequency of labeled paths~\cite{KTI03, MUAPV05}, 
or small fragments~\cite{BFSS03, DKWK05}. 
Akutsu \textit{et al.}~\cite{akutsu12} have formulated the graph 
pre-image problem as the problem of inferring graphs 
from the frequency of paths of labeled vertices, 
and proved that the problem is NP-hard even for planar graphs with bounded degree~\cite{akutsu12}. 
They also developed a dynamic programming algorithm for trees with
bounded degree, which was extended by Fazekas \textit{et al.}~\cite{FIOST13} 
for graphs of bounded degree and treewidth~2.
Although these algorithms do not work in time polynomial in the input
size (i.e., the size of a given feature vector), 
Nagamochi~\cite{nagamochi09} has proved that a graph determined by 
the frequency of paths with length 1 can be found in polynomial time, 
if such a graph exists. 

From a practical viewpoint, 
useful enumeration tools such as MOLGEN~\cite{MOLGEN5}, 
OMG~\cite{OMG}, and similar  have been
developed.
However, 
they are not necessarily very efficient if large structures are to be enumerated
because many of them treat general graph structures. 
In particular, the number of
molecules (i.e., chemical graphs) with up to 30 atoms (vertices)
{\tt C}, {\tt N}, {\tt O}, and {\tt S}, 
may exceed~$10^{60}$~\cite{BMG96}.
It might be possible to develop much faster algorithms 
if we restrict the class of target chemical structures
and employ recent techniques for enumeration of graph structures. 
Fujiwara~\textit{et al.}~\cite{Fujiwara08} studied the enumeration 
of tree-like chemical graphs that satisfy
a given feature vector which specifies the frequency of all paths of up to a prescribed length $K$
in a chemical compound to be constructed.
They proposed a branch-and-bound algorithm 
which consists of a branching procedure based on 
the tree enumeration algorithm of Nakano and Uno~\cite{NU05_colored, NU03_rooted} 
and bounding operations based on properties of path frequency and atom-atom bonds.
To reduce the size of the search space, 
Ishida {\em et al}.~\cite{IZNA08} introduced a new bounding operation, 
called \emph{detachment-cut}, 
based on the result of  Nagamochi~\cite{nagamochi09}. 
However, in many instances, constraining the frequency of prescribed paths 
by a single feature vector admits no solution.
To seek solutions effectively with relaxed constraints, 
Shimizu {\em et al}.~\cite{SNA11} 
introduced the problem of enumerating 
tree-like hydrogen-suppressed chemical graphs that  
satisfy one of a given set of feature vectors, 
which is specified as the range between a pair of an upper and a lower feature vector. 
They proposed a branch-and-bound algorithm for the problem, 
and afterward 
Suzuki {\em et al}.~\cite{Suzuki14} proposed
a more efficient and effective algorithm, called a 2-Phase algorithm. 
Furthermore, Suzuki~{\em et al.}~\cite{Suzuki14} extended the class of structures to 
{\em monocyclic graphs} (i.e., graphs containing exactly one cycle).
They proposed an algorithm for enumerating monocyclic graphs
by adding a simple edge to a pair of nonadjacent vertices of a multi-tree 
obtained by the 2-Phase algorithm. 
The above work on efficient enumeration of restricted classes of chemical graphs
has been identified by
Vogt and Bajorath~\cite{VB12} as a new trend in 
the field of chemoinformatics.

As shown by Nakano and Uno~\cite{NU03_rooted, NU05_colored}, 
the class of trees admits a nice labeling scheme based on a depth-first-search, 
called ``left-heavy trees,"
which enables us to generate trees in a constant time per tree
without executing any explicit 
comparison for isomorphism with
previously generated labeled graphs. 
Development of algorithms for  enumerating chemical graphs with a ``non-tree structure"
is a challenging task if we still wish to attain as high computational
efficiency as has been achieved for the enumeration of tree-like chemical graphs, 
since no such effective labeling scheme is known for general graphs. 
 
A multi-graph is a graph that can have multiple edges between the same pair of vertices.
In this work, two multiple edges with the same endvertices are not considered as a cycle. 
Rather, these multi-edges can be thought of as multiple 
bonds in chemical compounds. 
Let us define  a {\em $k$-augmented tree} to be  a connected multi-graph such that the number of pairs of
adjacent vertices minus the number of vertices is $k-1$. 
That is, a $k$-augmented tree is a graph obtained by adding some edges to a pair of nonadjacent vertices in a $(k-1)$-augmented tree. 
Hence a multi-tree is a 0-augmented tree. 

In the $97, 092, 888$  chemical compounds in
the  PubChem database, 
the ratio of the number of chemical compounds of a $k$-augmented tree structure 
to that of  all registered   chemical compounds
is around $2.9\%$, $13.3\%$, $28.2\%$, $24.2\%$  and $16.0\%$   
for $k=0, 1, 2, 3$ and $4$, respectively.

As a next step toward efficient enumeration of non-tree chemical graphs, 
we consider the problem of enumerating 
chemical graphs with a 2-augmented tree structure.
We classify 2-augmented trees into two types of structure;
one type of structure has  two edge-disjoint cycles and
the other has three cycles each pair of which share edges. 
We call the former structure a {\em bi-block  2-augmented tree}.
See Figure~\ref{fig-anchor} for an example of a bi-block 2-augmented tree  structure.

\begin{figure}[htbp]
\begin{center}
\includegraphics[width=8cm]{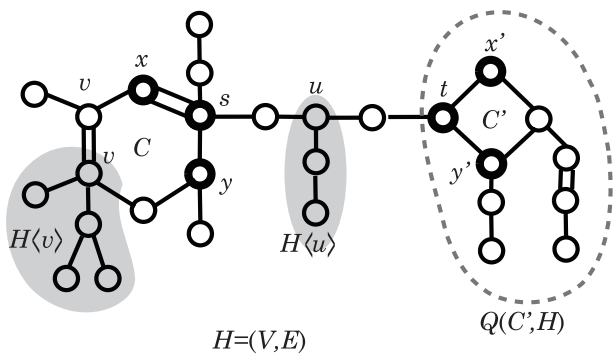}
\caption{An example of a bi-block 2-augmented tree $H$ with cycles $C$ and $C'$, 
where the vertices $s$ and $t$ are anchor vertices of $C$ and $C'$, respectively, 
and the pairs $\{s, x\}$ and $\{s, y\}$ are anchor pairs of $C$, 
and the pairs $\{t, x'\}$ and $\{t, y'\}$ are anchor pairs of $C'$.  
The subgraph enclosed by a dotted line represents $Q(C', H)$, and the subgraphs
 shaded in gray represent $H\langle u \rangle$ and $H\langle v \rangle$, respectively. } 
\label{fig-anchor}
\end{center}
\end{figure}

In this paper we design a branch-and-bound algorithm for enumerating 
bi-block 2-augmented trees
that satisfy a specification of path frequencies of graphs to be enumerated. 
Although no effective labeling scheme
akin to the one for trees proposed by 
Nakano and Uno~\cite{NU03_rooted, NU05_colored}
 is known for bi-block 2-augmented trees, 
we can have a monocyclic chemical graph by removing the edges between 
two adjacent vertices in one of the cycles of a bi-block 2-augmented tree. 
Thus, we can obtain bi-block 2-augmented trees by 
adding edges between a pair of nonadjacent vertices of a monocyclic chemical graph. 
For a bi-block 2-augmented tree $H$ and a pair $\{x, y\}$ of adjacent vertices in $H$, 
an algorithm for enumerating monocyclic chemical graphs, 
e.g., the one which 
Suzuki~{\em et al}.~\cite{Suzuki14} proposed, can be used without any major modification, 
to enumerate 
such monocyclic chemical graphs obtained by removing all edges between $x$ and $y$ 
from the bi-block 2-augmented tree $H$.
Thus the main task is to efficiently test canonicity (isomorphic uniqueness) of  
generated labeled bi-block 2-augmented trees.
To design such a procedure, 
we define  the  ``parent''  of a bi-block 2-augmented tree $H$ 
to be a monocyclic chemical graph $G_H$ so that the same bi-block 2-augmented tree $H$ will not be
generated from different monocyclic chemical graphs $G$ and $G'$.
As a result, we can combine a new procedure
with the algorithm for enumerating monocyclic chemical graphs to obtain
an algorithm for enumerating hydrogen-suppressed 
chemical graphs with bi-block 2-augmented tree
structure which satisfy a given specification on path frequencies.

From our experimental results where path frequency specifications
were obtained from chemical compounds from the PubChem database, 
we observe that our algorithm offers a clear advantage when
frequencies of non-trivial paths are specified as input, 
over other chemical graph enumeration tools such as
MOLGEN~5.0~\cite{MOLGEN5}, which is  regarded to be a state-of-the art 
software package for chemical graph enumeration.

\section{Preliminaries}
\label{sec:chemical-graphs}
This section reviews some basic definitions on graphs and introduces chemical graphs used in the paper.  

\subsection{Multigraphs}
Let $\mathbb{Z_+}$ denote the set of positive integers.
For two integers $a$ and $b$ with $a \leq b$, the set of all integers 
$i$ with $a \leq i \leq b$ is denoted by $[a, b]$. 

A graph is defined to be an ordered pair $(V, E)$ of a finite set $V$ 
of vertices and a finite set $E$ of edges, where an edge is an unordered pair of distinct vertices. 
Thus, we assume that no self-loops exist in graphs. 
An edge with two endvertices $u$ and $v$ is denoted by~$uv$. 

A graph such that  each vertex and each edge have unique indexes is called a {\em labeled graph}. 
We treat graphs as labeled when we distinguish vertices, edges or some other structures in a graph. 
A graph is called a {\em multi-graph} when there can be more than
one edge between the same endvertices, and a graph is called a {\em simple graph} 
if no multiple edges are allowed. 
In this paper, graphs are treated to be multigraphs unless stated otherwise. 

We denote the vertex set and the edge set of a graph $G$ by $V(G)$ and $E(G)$, respectively. 
Let $G$ be a graph. 
The multiplicity, i.e., the number of edges between
two vertices $u, v \in V(G)$ is denoted by $\mathrm{mul}_G(u, v)$. 
If the graph $G$ is clear from the context, 
then we denote $\mathrm{mul}_G(u, v)$ by $\mathrm{mul}(u, v)$. 
An edge $uv$ in  $G$ is called {\em simple} if $\mathrm{mul}(u, v)=1$. 
If there is no edge between $u$ and $v$, then we denote $\mathrm{mul}(u, v) = 0$. 
For a vertex $v$ in  $G$, we call the number of edges incident to $v$ the {\em degree}
 of $v$ and denote it by $\mathrm{deg}(v)$. 
A pair $(V', E')$ of  subsets $V'\subseteq V(G)$ 
and $E' \subseteq E(G)$  is called a {\em subgraph} of $G$
if  $(V', E')$ forms a graph; i.e., the two endvertices of each edge in $E'$ are contained in $V'$. 
We write $G' \subseteq G$ when $G'$ is a subgraph of $G$. 
We say that a subset $X\subseteq V(G)$ {\em induces} a subgraph $G'$ if
$V(G')=X$ and $E(G')$ contains every edge between two vertices in $X$, 
and the subgraph induced from $G$ by $X$ is denoted by $G[X]$. 
Let $\{u, v\} \subseteq V(G)$ be a pair of  vertices. 
Let $G - uv$ denote
 the  graph $G'$  obtained by removing $\mathrm{mul}_G(u, v)$
 simple edges between $u$ and $v$ from $G$.
Conversely, let $G + q \cdot uv$ denote the  graph $G'$ 
obtained by adding $q\geq 0$ simple edges between $u$ and $v$ to $G$, 
i.e., $\mathrm{mul}_{G}'(u, v)=\mathrm{mul}_G(u, v)+q$.
In particular, if $q=1$, then we denote $G+1\cdot uv$ by $G+uv$. 
For a subset $E' \subseteq E(G)$, let $G - E'$ denote the graph $(V(G), E(G) \setminus E')$.

For a nonnegative integer $k$, a  graph $P$   
which consists of $k+1$ distinct vertices $v_0, v_1, \ldots, v_k $ and
$k$ edges $v_{i-1}v_{i}$, $i \in [1, k]$, is called a {\em path} of length $k$, 
also a $v_0, v_k$-path, and is denoted by $P=v_0 v_1\cdots v_k$. 
A graph $C$ which consists of a path $v_0v_1 \cdots v_k$ of length $k \geq 2$ and an edge 
between $x_k$ and $x_0$ is called a {\em cycle}, denoted by $C=v_0 v_1 \cdots v_k v_0$. 
Thus, in this paper, a graph with two vertices $u$ and $v$ 
and two edges between them is not treated
as a cycle.
A connected multi-graph (resp., simple graph) that does not contain any cycle
as its subgraph is called a {\em multi-tree} (resp., {\em simple tree}). 
 
Every tree $T$ has either a vertex $v$ or
an adjacent vertex pair $\{v, v'\}$ 
removal of which  leaves no connected component 
with more than $\lfloor|V(T)|/2 \rfloor$ vertices~\cite{J69}. 
Such a vertex or an adjacent vertex pair is called a {\em centroid}, 
where  a centroid $v$ is called a {\em unicentroid} and
a centroid $\{v, v'\}$  is called a {\em bicentroid}.

A {\em rooted} tree is a tree in which either a vertex or an edge is designated as a root, 
where a vertex $v$ is called an {\em ancestor} of another vertex $u$ if $v$ 
appears along the path between $u$ and the root, 
and an ancestor $v$ of a vertex $u$ is called the {\em parent} if $v$ and $u$ are adjacent. 
In a rooted multi-tree $T$, the parent vertex of a non-root vertex $v$ is denoted by 
$\parent(v)$ and 
the {\em depth}  $\depth(v)$ of a vertex $v$ is defined to be 
the length of the path in $T$ between $v$ and the root. 
If a vertex $v \in V(T)$ is the root or an endvertex of the root edge, 
then it has no parent vertex and $\depth(v)=0$. 
For a vertex $v \in V(T)$, we denote by $D_v$ 
the set of descendants of $v$, where $v\not\in D_v$, 
by $\ch(v)$ the set of children of $v$, 
and  by  $T_v$   the subtree $T[\{v\}\cup D_v]$ of $T$ induced from $T$ by $\{v\}\cup D_v$. 
For an edge $uv \in E(T)$ such that $u=\parent(v)$, 
we denote by $T_{uv}$  the subtree $T[\{u, v\}\cup D_v]$ of $T$.
That is, $T_{uv}$ consists of the subtree $T_v$ and the vertex $u=\parent(v)$
 joined by $\mathrm{mul}_T(u, v)$ edges between $u$ and $v$.

\subsection{$k$-Augmented Trees}
In Section~\ref{sec:introduction}, we defined a $k$-augmented tree 
to be a connected multi-graph such that the number 
of pairs of adjacent vertices minus the number of vertices is $k-1$. 
Note that we can obtain a $k$-augmented tree from a multi-tree $T$ by adding edges 
between $k$ different pairs of nonadjacent vertices in~$T$. 
A monocyclic chemical graph is a connected graph which contains exactly one cycle. 
We introduce terminology on the structure 
of monocyclic chemical graphs and bi-block 2-augmented trees.
\medskip

\noindent {\bf Monocyclic chemical graphs~} 
Let $G$ be a monocyclic chemical graph with a unique cycle $C$. 
For a vertex $v \in V(C)$, we define the {\em pendent tree} 
$G \langle v \rangle$ rooted at $v$   to be 
the subgraph induced by the set of vertices reachable from $v$
in the graph $G - E(C)$.

We treat $G\langle v \rangle$ as a tree rooted at $v$. 
We extend the notation such that for vertices $v \in V(C)$ and 
$u \in V(G \langle v \rangle)$, $u \neq v$, 
we denote by $G \langle u \rangle$ the tree induced in 
$G\langle v \rangle$ by the set of vertices reachable
from $u$ in the graph $G\langle v \rangle - u\parent(u)$, and take it to be rooted at~$u$.
For a vertex $v \in V(G)$, let $\mathrm{r}(v, G)$ denote the vertex 
$v^*$ in $C$ such that $v \in V(G \langle v^* \rangle)$.  
\medskip

\noindent {\bf Bi-block 2-augmented trees~} 
Let $H$ be a bi-block 2-augmented tree with two edge-disjoint cycles $C$ and $C'$. 
We define subgraphs $P(C, C')$ and $Q(C, H)$ as follows.
We define $P(C, C')$ to be the path from $C$ to $C'$ that includes exactly 
one vertex from $C$ and $C'$ as endpoints, 
 respectively, where  $P(C, C')$ consists of a single vertex when the two cycles share a vertex. 
We call each endpoint of this path, or the vertex that is contained by both cycles an {\em anchor vertex}.  
We denote the anchor vertex in $C$ of $H$ by $\mathrm{anc}(C, H)$. 
We call a pair of an anchor vertex $v=\mathrm{anc}(C, H)$ 
and a neighbor of $v$ in $C$ an {\em anchor pair} in $C$. 
Let $\mathrm{mma}(H)$ denote the minimum multiplicity over all anchor pairs. 
We define $Q(C, H)$ to be the subgraph induced from
$H$ by the set of vertices reachable from $\mathrm{anc}(C, H)$
in the graph $H - (E(P(C, C')) \cup E(C'))$.

Figure~\ref{fig-anchor} illustrates a bi-block 2-augmented tree $H$
with two edge-disjoint cycles $C$ and $C'$. 
In the figure, we see that $s = \mathrm{anc}(C, H)$, $t = \mathrm{anc}(C', H)$, $\{s, x\}$ and $\{s, y\}$
 are the anchor pairs of $C$ and $\{t, x'\}$ and $\{t, y'\}$ are the anchor pairs of $C'$, 
 and $\mathrm{mma}(H)=1$.
The subgraph enclosed by a dotted line in this figure represents $Q(C', H)$.

In the context of bi-block 2-augmented trees, 
for each vertex $v \in V(P(C, C')) \cup V(C) \cup V(C')$ we define
the pendent tree $H \langle v \rangle$ to be the tree rooted at $v$
and induced by $v$ and all vertices reachable
from $v$ 
in the graph $H - (E(P(C, C')) \cup E(C) \cup E(C'))$.
Again, we extend the notation, such that for a vertex $v \in  V(P(C, C')) \cup V(C) \cup V(C')$ 
and $u \in V(H \langle v \rangle)$, $u \neq v$, we denote 
by $H \langle u \rangle$ the tree rooted at $u$ induced in $H \langle v \rangle$ by the 
set of vertices reachable from $u$ in $H \langle v \rangle - u \parent(u)$.
Two examples of pendent trees are illustrated in Figure~\ref{fig-anchor}, shaded in gray.

\subsection{Chemical Graphs}

To represent chemical compounds as multigraphs, 
we define  a set $\Sigma$  of colors on vertices, each of which corresponds to 
a chemical element  such as oxygen~{\tt O}, nitrogen~{\tt N}, or carbon~{\tt C}, 
denoting the color of a vertex $v$  by $\mathrm{col}(v)$, 
and the valence of a chemical element by an integer function val({\em c}) $\in \mathbb{Z_+}$
of  each color {\em c} $\in \Sigma$.
The size of a bond between two adjacent elements is indicated 
by the multiplicity between the two corresponding vertices.
A multi-graph $G$ is said to be $\Sigma$-{\em colored}
if each vertex $v \in V(G)$ is assigned a color $\mathrm{col}(v) \in \Sigma$. 
Then chemical compounds can be viewed as $\Sigma$-colored, 
connected multi-graphs without self-loops, where vertices and colors 
represent atoms and elements, respectively. 
In the remaining part of this paper, we assume that graphs are
$\Sigma$-colored multi-graphs unless stated otherwise. 

Let $G$ be a $\Sigma$-colored multi-graph. 
In this paper, we treat {\em hydrogen-suppressed chemical graphs} where for some vertex $v$, 
$\mathrm{deg}(v)$ may be smaller than the valence $\mathrm{val}(\mathrm{col}(v))$ of the color $\mathrm{col}(v)$ of $v$.
For a vertex $v \in V(G)$, we define the {\em residual degree} $\res(v)$ 
to be $\mathrm{val}(\mathrm{col}(v))-\mathrm{deg}(v)$.   
In a hydrogen-suppressed chemical graph, for a vertex $v$, the residual degree $\res(v)$ represents the number of hydrogen atoms that are adjacent to $v$.

\subsection{Isomorphism of Chemical Graphs}
\label{sec:isomorphism}
In enumerating chemical graphs, we must avoid duplication of equivalent graphs.
For example, two chemical graphs $G$ and $G'$ may have the same graph structure, 
and imply the same chemical compound even if they are different as labeled graphs.
This case is formalized by the notion of \emph{isomorphism} as follows.
Let $G = (V, E)$ and $G' = (V', E')$ be two chemical graphs.
The following bijection $\psi$ from $V(G)$ to $V(G')$ is called an {\em isomorphism} from $G$ to $G'$:
\begin{itemize}
\item[(i)] for each vertex $x \in V(G)$, it holds that ${\rm col}(x) = {\rm col}(\psi(x))$; and
\item[(ii)] for each pair $\{x, y\} \subseteq V(G)$, 
it holds that ${\rm mul}_{G}(x, y) = {\rm mul}_{G}'(\psi(x), \psi(y))$. 
\end{itemize}

If there exists an isomorphism from $G$ to $G'$, then we say that $G$ and $G'$ are {\em isomorphic}.
We write $G \approx G'$ if $G$ and $G'$ are isomorphic, and write $G \not\approx G'$ otherwise.
For two sets $\mathcal{G}$ and  $\mathcal{G}'$ of chemical graphs, 
we say that  $\mathcal{G}'$  
{\em represents} $\mathcal{G}$ if \\
~-~for each  chemical graph  $G\in \mathcal{G}$, there is a chemical graph $G'\in \mathcal{G}'$ such that 
$G\approx G'$ and vice versa; and \\
~-~for any two chemical graphs $G'_1, G'_2\in \mathcal{G}'$, it holds that $G'_1\not\approx G'_2$. 

An {\em automorphism} of a chemical graph $G$ is defined to be 
an isomorphism $\psi$ from $V(G)$ to $V(G)$ itself. 

We also introduce isomorphism between rooted chemical graphs.
Let $G$ be a chemical graph rooted at a vertex $v_{r}$, 
and $G'$ be a chemical graph rooted at a vertex $v'_{r}$.
If there exists an isomorphism $\psi$ from $G$ to $G'$ 
such that $\psi(v_{r}) = v'_{r}$, then we call $\psi$ a {\em rooted isomorphism}.
If there exists a rooted isomorphism from $G$ to $G'$, we say that $G$ and~$G'$ 
are {\em rooted isomorphic} and denote this by $G \underset{r}{\approx} G'$.

\subsection{Feature Vectors}
\label{sec:feature_vectors}

Our feature vectors represent occurrences of path structures in graphs. 
For a set $\Sigma$ of colors where each color 
$c \in \Sigma$ is attributed a valence $\val(c)$, 
and integers $d \geq 1$ and $K \geq 0$, 
let $c_0, c_1, \ldots, c_K \in \Sigma$ be $K+1$ 
colors and $m_1, m_2, \ldots, m_K \in [1, d]$ be $K$ integers, 
where possibly $c_i=c_j$ or $m_i=m_j$ for some $i, j$.  
Then the alternating sequence  $t = (c_0, m_1, c_1, \ldots , m_K, c_K)$
is called a {\em colored sequence} of length $|t| = K$.
Let $\Sigma^{K, d}$ denote the set of all
colored sequences of length $K$ and maximum multiplicity at most~$d$,
and denote the union of $\Sigma^{0, d}$, $\Sigma^{1, d}$, 
$\ldots$, $\Sigma^{K, d}$ by $\Sigma^{\leq K, d}$.
For a colored sequence $t=(c_0, m_1, c_1, \ldots , m_K, c_K)\in \Sigma^{K, d}$, 
we define the reverse sequence 
$\mathrm{rev}(t)\in \Sigma^{K, d}$ of $t$ to be 
$\mathrm{rev}(t) = (c_K, m_K, c_{K-1}, \ldots , $ $m_1, c_0)$.

Given a $\Sigma$-colored path $P=v_0v_1\cdots v_K$ of length $K$, we define its 
colored sequence $ \gamma(P) \in \Sigma^{K, d}$ to be  
\begin{equation}
\gamma(P)=(\mathrm{col}(v_0), \mathrm{mul}(v_0, v_1), \mathrm{col}(v_1), \ldots, 
\mathrm{mul}(v_{K-1}, v_K), \mathrm{col}(v_K)). \nonumber
\end{equation}
Let $G$ be a $\Sigma$-colored labeled multigraph. 
For a colored sequence $t \in \Sigma^{\leq K, d}$, 
the {\em frequency} $\mathrm{frq}(t, G)$ of $t$ in $G$
is defined to be the number of vertex-rooted subgraphs $G' \subseteq G$ such that $G'$ is a path 
 with $\gamma(G') = t$. 
We define the {\em feature vector} $\mathbf{f}(G)$ of {\em level} $K$ of $G$ 
to be the vector such that $\mathbf{f}(G)[t] = \mathrm{frq}(t, G)$ 
for each colored sequence $t \in \Sigma^{\leq K, d}$. 

Figure~\ref{fig-featurevector} gives an example of 
a hydrogen-suppressed chemical graph $G$
and its feature vector $\mathbf{f}(G)$, where 
$G$ is a $\Sigma$-colored multigraph with $\Sigma=\{${\tt O}, {\tt N}, {\tt C}$\}$, 
$d=2$, and $K=1$.

\begin{figure}[htbp]
\begin{center}
\includegraphics[width=13cm]{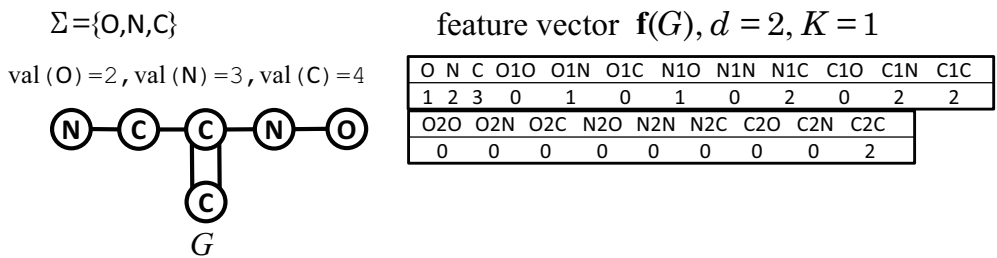}
\caption{An example of a hydrogen-suppressed chemical graph $G$ 
and its feature vector~$\mathbf{f}(G)$ 
with maximum multiplicity $d=2$ and maximum length $K=1$. 
The feature vector consists of the number of 
occurrences of each atom type and each bond type. 
The field {\tt C}2{\tt C} denotes the double bond between two carbon atoms. 
Each path is counted when traversing it in both directions, and the entries 
$\mathrm{frq}(t, G)$ for $t=${\tt C}1{\tt C} and $t=${\tt C}2{\tt C} are 2, 
 even though there is a single such bond in the chemical compound implied by~$G$. 
}
\label{fig-featurevector}
\end{center}
\end{figure}

Given a color set $\Sigma$ and integers $d$ and $K$, 
the set of $|\Sigma^{\leq K, d}|$-dimensional vectors whose entries are
nonnegative integers is  
called a  {\em feature vector space} and is
denoted by ${\bf f}(\Sigma, K, d)$.
Equivalently, each vector $\mathbf{g} \in \fvset(\Sigma, K, d)$ 
is a mapping $\mathbf{g} : \Sigma^{\leq K, d} \to \mathbb{Z}_{+}$.
For two vectors $\mathbf{g}, \mathbf{g}' \in \fvset(\Sigma, K, d)$, 
we write $\mathbf{g} \leq \mathbf{g}'$ if for each entry $t \in \Sigma^{\leq K, d}$, 
it holds that $\mathbf{g}[t] \leq \mathbf{g}'[t]$. 
For two given vectors ${\bf g}, {\bf g}' \in {\bf f}(\Sigma, K, d)$, 
a chemical graph $G$ is called {\em feasible}
if ${\bf g} \leq {\bf f}(G) \leq {\bf g}'$ and $\mathrm{res}(v) \geq 0$
holds for all vertices~$v \in V(G)$.
Let $\mathcal{G}({\bf g}, {\bf g}')$
denote the set of all chemical graphs feasible to the pair $({\bf g}, {\bf g}')$.

\section{The Problem of Enumerating Bi-block 2-Augmented Trees}
\label{sec:problem_2-cycle}

Let $\mathcal{G}_1$ and $\mathcal{G}_2$ respectively
denote the sets of labeled
$\Sigma$-colored monocyclic and bi-block 2-augmented trees.
Following, for two given vectors ${\bf g}_{\ell}$ and ${\bf g}_{u}$, 
let $\mathcal{G}_1(\mathbf{g}_{\ell}, \mathbf{g}_u)$
and  $\mathcal{G}_2(\mathbf{g}_{\ell}, \mathbf{g}_u)$
denote the set of graphs in $\mathcal{G}_1$ and $\mathcal{G}_2$, respectively, 
that are feasible to the pair $(\mathbf{g}_{\ell}, \mathbf{g}_u)$.

The problem we deal with in this paper
is for given vectors $\mathbf{g}_{\ell}$ and $\mathbf{g}_u$
to construct a set $\mathcal{G}'_2$
that represents the set $\mathcal{G}_2(\mathbf{g}_{\ell}, \mathbf{g}_u)$
of feasible bi-block 2-augmented trees.

We approach this task based on the assumption that we are given a
correct collection $\mathcal{G}'_1$ of monocyclic chemical graphs
(e.g., obtained by the algorithm due to Suzuki~{\em et al.}~\cite{Suzuki14}), 
in the sense that
for each bi-block 2-augmented tree $H \in \mathcal{G}_2(\mathbf{g}_{\ell}, \mathbf{g}_u)$
a graph isomorphic to $H$ can be obtained by
adding a multiple edge between a pair of non-adjacent vertices
in some of the graphs in the collection~$\mathcal{G}'_1$.

We here consider how to characterize 
a set of monocyclic chemical graphs that is sufficient 
to generate a set of bi-block 2-augmented trees that represents 
$\mathcal{G}_2(\mathbf{g}_{\ell}, \mathbf{g}_u)$
in the above manner. 
Analogously with the idea due to Suzuki~{\em et al.}~\cite{Suzuki14}, 
we modify $\mathbf{g}_{\ell}$ into a lower bound $\mathbf{g}^{\dagger}_{\ell}$ 
so that a graph isomorphic to 
each of the bi-block 2-augmented trees in $\mathcal{G}_2(\mathbf{g}_{\ell}, \mathbf{g}_u)$ 
can be constructed from
the {\rm 1}-augmented trees in $\mathcal{G}_1(\mathbf{g}^{\dagger}_{\ell}, \mathbf{g}_{\ell})$.
For a given lower bound $\mathbf{g}_{\ell} \in \fvset(\Sigma, K, d)$ on feature vectors, 
let $\mathbf{g}^{\dagger}_{\ell} \in \fvset(\Sigma, K, d)$ be the vector defined as follows:

  - for each $t \in \Sigma^{0, d}$, let $\mathbf{g}^{\dagger}_{\ell}[t] = \mathbf{g}_{\ell}[t]$;  

 - for each $t = (c, m, c') \in \Sigma^{1, d}$ with $c, c' \in \Sigma$ and $m \in [1, d]$, let 
\begin{equation}
\mathbf{g}^{\dagger}_{\ell}[t]= \left\{ 
\begin{array}{ll}
{\rm max}\{ 0, \mathbf{g}_{\ell}[t]-1\} \ {\ } \mathrm{if} \ c \neq c'  \nonumber \\
{\rm max}\{ 0, \mathbf{g}_{\ell}[t]-2\} \ {\ } \mathrm{if} \ c = c;
\end{array}
\right. 
\end{equation}

  - for each $t \in \Sigma^{\leq K, d} \setminus \Sigma^{\leq 1, d}$, let $\mathbf{g}^{\dagger}_{\ell}[t] = 0$.

\begin{lem}
\label{lem:existence} 
For a given set $\Sigma$ of colors and integers $K \geq 0$ and $d \geq 1$, 
let the vectors $ \mathbf{g}_{\ell}, \mathbf{g}_u \in \fvset(\Sigma, K, d)$
satisfy   $\mathbf{g}_{\ell} \leq \mathbf{g}_u$ and
 $\mathbf{g}_{\ell}[t] = \mathbf{g}_u[t]$ for each $t \in \Sigma^{0, d}$, 
and  let $H \in \mathcal{G}_{{\rm 2}}(\mathbf{g}_{\ell}, \mathbf{g}_u)$. 
Then, for each anchor pair $\{x, y\}$ in $H$, it holds that 
$H - xy \in \mathcal{G}_1(\mathbf{g}^{\dagger}_{\ell}, \mathbf{g}_u)$.  
\end{lem}

\begin{proof}
Let $xy \in E(H)$ be an edge such that $\{x, y\}$ is an anchor pair in $H$,
and let $G$ be the graph $H - xy$.
We show that $G \in \mathcal{G}_1({\bf g}^{\dagger}_{\ell}, {\bf g}_{u})$ 
i.e., $G$ is a $1$-augmented tree with ${\bf g}^{\dagger}_{\ell} \leq {\bf f}(G) \leq {\bf g}_{u}$.
Obviously, $G = H - xy$ is a monocyclic chemical graph, 
since removing from $H$ any edge between an anchor pair leaves exactly one cycle in $G$.
For $H\in \mathcal{G}_2({\bf g}_{\ell}, {\bf g}_{u})$, 
we have ${\bf g}_{\ell}\leq \fvset(H)\leq {\bf g}_{u}$. 
Since the frequency of any colored sequence does not increase, 
we have $\fvset(G)\leq \fvset(H)\leq {\bf g}_{u}$.
Hence, the upper bound ${\bf g}_{u}$ remains valid for $G$.
Next we consider the lower bound of $G=H - xy$.
For any colored sequence $t \in \Sigma^{0, d}$, 
when removing an edge between $x$ and $y$, 
we have 
${\bf g}_{\ell}[t] = \freq(t, G) = \freq(t, H) = {\bf g}_{u}[t]$ 
since the frequency of a colored sequence with length 0 does not change.
For a colored sequence $t = (c, m, c') \in \Sigma^{1, d}$ 
with $\lvert t \rvert \leq L$, when removing an edge $xy$, 
if $\col(x)=c$, $\col(y)=c'$ or $\col(x)=c'$, $\col(y)=c$ and $\mul_H(x, y)=m$, then we have:
\begin{equation*}
	\fvset(G)[t] = 
	\left\{
		\begin{array}{ll}
			\max\{0 , \fvset(H)[t] - 1\}, & \text{if}~c \neq c' \\
			\max\{0 , \fvset(H)[t] - 2\}, & \text{if}~c = c' 
		\end{array}
	\right.
\end{equation*}
otherwise we have $\fvset(G)[t]=\fvset(H)[t]$.

For any colored sequence $t=(c, m, c') \in \Sigma^{1, d}$, 
since ${\bf g}_{\ell}\leq \fvset(H)$, 
by removing edge $xy$ we get:
\begin{equation*}
\label{eq:decrease_in_fv}
	\fvset(G)[t] \geq 
	\left\{
		\begin{array}{ll}
			\max\{0 , \mathbf{g}_{L}[t] - 1\}, & \text{if}~c \neq c' \\
			\max\{0 , \mathbf{g}_{L}[t] - 2\}, & \text{if}~c = c'. 
		\end{array}
	\right.
\end{equation*}

For any colored sequence 
$t\in\Sigma^{\leq K, d} \setminus \Sigma^{\leq 1, d}$, 
it is clear that $\fvset(G)[t]\geq \mathbf{g}^{\dagger}_{\ell}[t]=0$.
Hence, we have $\mathbf{g}^{\dagger}_{\ell} \leq \fvset(G) \leq \mathbf{g}_{u}$,
as required.
%
f
\end{proof}

Lemma~\ref{lem:existence} demonstrates that for each
 bi-block 2-augmented tree 
$H \in \mathcal{G}_2(\mathbf{g}_{\ell}, \mathbf{g}_u)$, 
there is at least one monocyclic chemical graph 
$G \in \mathcal{G}_1(\mathbf{g}^{\dagger}_{\ell}, \mathbf{g}_u)$ 
and a pair $\{x, y\}$ of nonadjacent vertices 
that delivers $H$ as $G+q \cdot xy$
for an integer $q \leq \min \{\res(y), \res(x)\}$. 
Thus our first task is to generate a set $\mathcal{G}'_1$ that represents 
$\mathcal{G}_1(\mathbf{g}^{\dagger}_{\ell}, \mathbf{g}_u)$ 
by some rules  for generating monocyclic chemical graphs such as Suzuki~{\em et al}.~\cite{Suzuki14} proposed. 
Then we enumerate all bi-block 2-augmented trees in $\mathcal{G}_2(\mathbf{g}_{\ell}, \mathbf{g}_u)$
 by adding a multiple edge between a pair of nonadjacent vertices in each given monocyclic chemical graph, 
where the vertex pair becomes an anchor pair in the newly created bi-block 2-augmented tree. 
The problem is formalized as follows. 

\bigskip

\noindent 
{\bf Enumerating bi-block 2-augmented trees from a given set of monocyclic chemical graphs}\\
{\bf Input}: A color set $\Sigma$, integers $K \geq 0$ and $d \geq 1$, 
two vectors $\mathbf{g}_{\ell}, \mathbf{g}_u \in \fvset(\Sigma, K, d)$ such that
 $\mathbf{g}_{\ell} \leq \mathbf{g}_u$ and 
 $\mathbf{g}_{\ell}[t] = \mathbf{g}_u[t]$ for each  $t \in \Sigma^{0, d}$, 
and a set $\mathcal{G}'_1$ that represents
the set $\mathcal{G}_1(\mathbf{g}^{\dagger}_{\ell}, \mathbf{g}_u)$. \\
{\bf Output}: A set $\mathcal{G}'_2$ that represents the set
			  $\mathcal{G}_2(\mathbf{g}_{\ell}, \mathbf{g}_u)$.\par
\medskip

\subsection{Ideas Behind the Enumeration Process}
\label{sec:ideas}
We start with the assumption that
we are given a set $\mathcal{G}'_1 = \{G_1, G_2, \ldots, G_p\}$
that represents the set 
$\mathcal{G}_1(\mathbf{g}^{\dagger}_{\ell}, \mathbf{g}_u)$.
From this set we construct a set $\mathcal{G}'_2$
that represents the set  $\mathcal{G}_2(\mathbf{g}_{\ell}, \mathbf{g}_u)$ by
iterating over the set $\mathcal{G}'_1$ and for each graph $G_i$,
constructing a bi-block 2-augmented graph $H = G_i + q \cdot xy$
by adding a multiple edge between pairs $\{x, y \} \subseteq V(G_i)$ 
of non-adjacent vertices,
and each integer $q \in [1, \min \{d, \res(x), \res(y)\}]$,
discarding such graphs that are not feasible for the lower and
upper bounds ${\bf g}_{\ell}$ and ${\bf g}_u$ on feature vectors.
During this process, we must make sure that
a graph $H'$ isomorphic to $H$ has not already been
constructed as $G_j + q' \cdot x'y'$ during the enumeration process.
We call the duplication arising when 
$H = G_i + q\cdot xy$ and $H' = G_j + q'\cdot x'y'$ are isomorphic
with $i \neq j $ \emph{inter-duplication}, 
and that when $H = G_i + q \cdot xy \approx H' = G_i + q' \cdot x'y'$
\emph{intra-duplication}.
The main challenge now is how to avoid such duplications 
without storing all generated graphs and 
explicitly comparing the new one with each of them.

In many enumerating algorithms, in order to efficiently cope with duplication, 
the concept of a {\em family tree} has been 
widely employed~\cite{NU03_rooted, NU05_colored}. 
To define a family tree for graphs, we need to define a {\em parent-child relationship}
between graph structures, so that the parent structure of a given multi-graph $G$
is uniquely determined from the  topological structure of $G$.
As it will be observed in Section~\ref{sec:parent-child-relationship},   each child structure
  is obtained by adding a vertex or an edge to its parent structure. 
Let $\pi$ be a mapping giving a parent-child relationship among graphs, 
where a graph $G = \pi(H)$ is called the {\em parent} of $H$ 
and a graph $H$ with $\pi(H) = G$ is called a {\em child} of~$G$. 
Any definition for a parent-child relationship is valid as long as the parent $\pi(H)$
is uniquely determined as an unlabeled multi-graph based only on the topological
information of $H$. 

\begin{lem}
For two graphs $H$ and $H'$, if $\pi(H) \notiso \pi(H')$, then $H \notiso H'$. 
\end{lem}
\begin{proof}
Suppose that $H \approx H'$. 
Since the parents of $H$ and $H'$ are 
uniquely determined based on the topological structures of $H$ and $H'$, 
we have $\pi(H) \approx \pi(H')$, 
which contradicts the assumption that $\pi(H) \notiso \pi(H')$. 
\end{proof}

We define a parent-child relationship $\pi$
between monocyclic chemical graphs and bi-block 2-augmented trees in Section~\ref{sec:parent-child-relationship}.
This leaves the problem of dealing with intra-duplications, with which we
deal in Section~\ref{sec:avoid-symmetrical-duplication}.
To this aim, for a monocyclic graph $G$ 
we define a \emph{proper} set $F$ of edges
to be a set of non-adjacent vertex pairs $\{x, y\} \subseteq V(G)$ such that \\
- for any $\{x, y\} \in F$, $G + xy$ is a bi-block 2-augmented tree, \\ 
- for any $H$ with $\pi(H) = G$ there is a pair $\{x, y\} \in F$ and an integer $q$ such that
  $H$ is isomorphic to $G + q \cdot xy$, and \\
- for any two distinct pairs $\{x, y\}, \{x', y'\} \in F$ and integers $q$ and $q'$
  it holds that $G + q \cdot xy \notiso G + q' \cdot x' y'$. \\
We investigate how to construct a proper set of edges in 
Secs.~\ref{sec:avoid-symmetrical-duplication} and~\ref{sec:potential_edges}.

To conclude this section, we give a general procedure 
that given a monocyclic graph generates all its mutually non-isomorphic children,
as an outline of our 
enumeration algorithm.

\bigskip
{\bf Procedure~{\procMain}}
\begin{algorithmic}[1]
\Require A monocyclic graph $G$ with maximum edge multiplicity at most~$d$.
\Ensure A set of graphs representing the 
			set of children bi-block 2-augmented trees 
		  of $G$ with multiplicity at most~$d$.
\State{Calculate a set $F$ of potential edges}; 
/* By Lemma~\ref{lem:proper_set} in Sec.~\ref{sec:potential_edges} */
  \For{{\bf each} pair $\{x, y\}$ of vertices in $F$}
    \For{{\bf each} $q \in [1, \min\{ d, \res(x), \res(y)\} ]$}
	    \If {$G + q \cdot xy$ is a child of $G$}
	    /* Checked by Procedure~{\procChildCheck} in Sec.~\ref{sec:parent-child-relationship} */
		    \State{{\bf output} $G + q \cdot  xy$}
	    \EndIf
    \EndFor
\EndFor.
\end{algorithmic}

\section{Signature}
\label{sec:signature}

We store multi-graphs as labeled ones and two graphs
with different labelings may be isomorphic to each other. 
Hence, testing for isomorphism of two multi-graphs reduces to finding
such labelings for these graphs so that the two labeled 
graphs completely match each other including the color of vertices, 
adjacency and edge multiplicity between each pair of vertices. 
For a class $\mathcal{G}$ of multi-graphs, if we have a way of 
choosing a labeling of each multi-graph $G \in \mathcal{G}$ 
which is unique up to the graph's automorphism, 
then we can test the isomorphism of two graphs directly by comparing their labels. 
Such a labeling for $G$ is called a {\em canonical form} of $G$. 
Once such a canonical form is obtained, we can easily 
encode each multi-graph $G \in \mathcal{G}$ into a code $\sigma(G)$, 
called the {\em signature} of $G$, such that two multi-graphs
$G, G' \in \mathcal{G}$ are isomorphic if and only if $\sigma(G)=\sigma(G')$.

\subsection{Lexicographical Order}
We fix a total order of colors in $\Sigma$ arbitrarily, e.g., ${\tt O} < {\tt N}< {\tt C}$. 
We define a {\em lexicographical order} among sequences with elements 
in $\Sigma \cup \mathbb{Z_+}$ as follows. 
A sequence $A=(a_1, a_2, \ldots , a_p)$ is lexicographically 
smaller than a sequence $B=(b_1, b_2, \ldots , b_q)$
 if there is an index $k \in [1, \mathrm{min}\{p, q\}]$ such that \\
~~ (1) $a_i=b_i$ for each $i \in [1, k]$; and \\
~~ (2) $k=p<q$ and $k < \mathrm{min}\{p, q\}$, or $a_{k+1} < b_{k+1}$. \\
In this case we denote $A \prec B$. 
If $p=q$ and $a_i = b_i$ for each $i \in [1, p]$, then we denote $A=B$. 
Let $A \preceq B$ mean that $A \prec B$ or $A = B$.

We will represent a lexicographically ascending (resp., descending) 
order on a collection $S = (s_1, s_2, \ldots, s_k)$ 
of $k$ sequences by a permutation
$\pi : [1, k] \to [1, k]$ such that for $1 \leq i < j \leq k$
it holds that $s_{\pi(i)} \preceq s_{\pi(j)}$ 
(resp., $s_{\pi(j)} \preceq s_{\pi(i)}$).

For a collection $S$ of sequences, let us denote by $||S|| = \sum_{s \in S} |s|$
the total length of the sequences
in the collection~$S$.
A known algorithm due to Aho {\em et al.}~\cite{AHU74}
can be used to lexicographically sort a collection $S$ of sequences
over an alphabet of size $n$
in $O(||S|| + n)$ computation time.

\subsection{Canonical Form and Signature of Trees}
We review canonical forms of rooted trees~\cite{NU03_rooted, NU05_colored}.

\subsubsection{Rooted Trees and Ordered  Trees}
\label{sub:ordered-trees}

An {\em ordered tree} is a rooted tree, given together with  a total order among the children of each vertex, 
where by convention we assume that the order of child vertices is from left to right. 

Let $T$ be a rooted multi-tree with $n$ vertices.
We can have many different ordered trees isomorphic to $T$.
A canonical form of  $T$ is given by an adequately 
chosen ordered tree on~$T$. 
Let $\tau$ be an ordered tree rooted-isomorphic to~$T$.
We assume that when we conduct a depth-first-search on $\tau$, we visit children from left to right. 
We denote the vertices in $V$ by $v_1, v_2, \ldots, v_n$ in the order 
visited by the depth-first-search starting from the root.
Let $\delta(\tau)$ denote the alternating sequence $(c_1, d_1, \ldots , c_n, d_n)$
 that consists of the color $c_i$ and depth $d_i$ of the $i$-th vertex $v_i$ 
 in the depth-first-search for $i \in [1, n]$, and
 let $\mathrm{M}(\tau)$ denote the sequence $(m_2, m_3, \ldots , m_{n})$ 
of the multiplicity $m_i=\mathrm{mul}(v_i, \parent(v_i))$ 
of the edge joining the $i$-th vertex $v_i$ and its parent $\parent(v_i)$ in $T$ for each $i \in [2, n]$.  
 
Let $T$ be an ordered $\Sigma$-colored multi-tree. 
For each vertex $v \in V(T)$, let 
$\mathrm{dfs}(v)=i$ if  $v$ is the $i$-th vertex $v_i$. 
For a vertex $v \in V(T)$, let $\mathrm{left}_T(v)$ 
denote the sibling immediately on the left of $x$, if such a vertex exists. 
Note that for each vertex $v \in V(T)$ such that $\mathrm{left}_T(v)$ exists, 
 it holds that $\parent(v) = \parent(\mathrm{left}_T(v))$. 
Let $\{u, v\} \subseteq V(T)$ be a pair of nonadjacent vertices with $\mathrm{dfs}(u) < \mathrm{dfs}(v)$. 
Let $P_T(u, v)$ denote the path in $T$ between $u$ and $v$. 
Let $\mathrm{lca}_T(u, v)$ denote the least common ancestor of 
$u$ and $v$, i.e., the highest vertex in $P_T(u, v)$. 
We define the respective greatest uncommon ancestors of two vertices $u$ and $v$: 
Let $\mathrm{gua}_T(u, v)$ denote the child of $\mathrm{lca}_T(u, v)$ 
that is in $P_T(u, \mathrm{lca}_T(u, v))$, 
 i.e., the vertex in $P_T(u, \mathrm{lca}_T(u, v))$ 
 which is closest to $\mathrm{lca}(u, v)$ except for $\mathrm{lca}_T(u, v)$; 
and $\mathrm{gua}_T(v, u)$ denote the child of $\mathrm{lca}_T(u, v)$ 
that is in $P_T(v, \mathrm{lca}_T(u, v))$. 
We see that $\mathrm{gua}_T(u, v) = \mathrm{gua}_T(v, u)$ 
if and only if $\mathrm{lca}_T(u, v) = u$. 
In the following subsection, we define the concept of a 
``left-heavy tree, " and in Section~\ref{sec:avoid-symmetrical-duplication}, 
we present how to leverage the notion of left-heavy trees to efficiently  detect isomorphic rooted trees.

\subsubsection{Left-heavy Trees}

A {\em left-heavy} tree of a rooted multi-tree $T$ is an ordered tree $\tau$ 
that has the maximum code $\delta(\tau)$ among all ordered trees of $T$. 
Note that a left-heavy tree has the following recursive structure: 
For every vertex $v \in V(T)$, the subtree $T_v$ 
is also a left-heavy tree and $\sigma(T_v)$ is a continuous subsequence of $\sigma(T)$. 
Following Nakano and Uno~\cite{NU03_rooted, NU05_colored}, 
we define the canonical form of a rooted multi-tree $T$ to be the left-heavy 
tree $\tau$ that has the maximum sequence $\mathrm{M}(\tau)$
among all left-heavy trees of $T$, and define the signature of $T$ 
to be $\sigma(T) = (\delta(\tau), \mathrm{M}(\tau))$.

Figure~\ref{figrootedtree} gives  an example  of the ordered trees of a rooted multi-tree $T$. 
The integer beside each vertex $v$ represents the value of $\mathrm{dfs}(v)$. 
Since the ordered trees $\tau_2$ and $\tau_3$ have lexicographically maximum sequences
$\sigma(\tau_2) = \sigma(\tau_3)$ among all ordered trees
$\tau$ of the rooted multi-tree $T$, $\tau_2$ and $\tau_3$ are left-heavy
tree representations of $T$. 
The ordered tree $\tau_3$ is the canonical form of $T$ since it is a 
left-heavy tree with lexicographically maximum sequence 
$\mathrm{M}(\tau_3)$ among all left-heavy trees of $T$. 
\begin{figure}[htbp]
\begin{center}
\includegraphics[width=12cm]{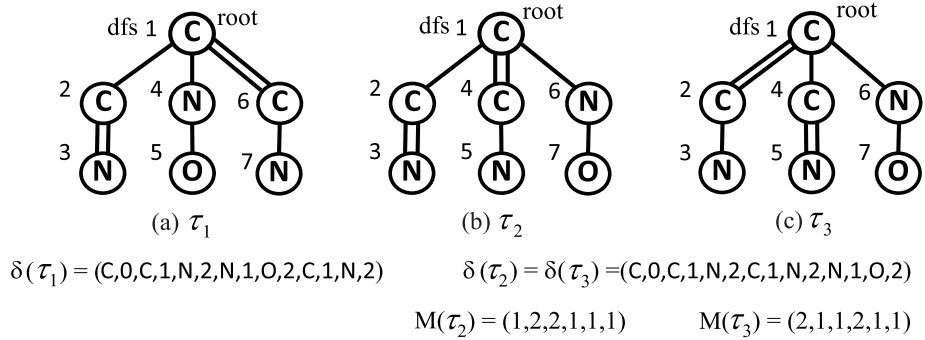}
\caption{An illustration of a rooted multi-tree $T$ and  the ordered trees 
$\tau_1$, $\tau_2$ and $\tau_3$ of $T$. 
It holds that $\sigma(\tau_1) = ({\tt C}, 0, {\tt C} , 1 , {\tt N}, 2, {\tt N}, 1, {\tt O}, 2, {\tt C}, 1, {\tt N}, 2)$, 
$\sigma(\tau_2) = \sigma(\tau_3) =
({\tt C}, 0, {\tt C}, 1, {\tt N}, 2, {\tt C}, 1, {\tt N}, 2, {\tt N}, 1, {\tt O}, 2)$, 
$\mathrm{M}(\tau_2) = (1, 2, 2, 1, 1, 1)$ and 
$\mathrm{M}(\tau_3) = (2, 1, 1, 2, 1, 1)$:
(a)~An ordered tree $\tau_1$ rooted at the centroid of $T$; 
(b)~A left-heavy tree $\tau_2$; and 
(c)~The canonical form $\tau_3$ of $T$. }
\label{figrootedtree}
\end{center}
\end{figure}

Using the canonical form for rooted multi-trees, we can define a canonical
form for unrooted multi-trees $T$ by regarding them 
as trees rooted at their centroids.

\subsubsection{Calculating the Signature of Rooted Multi-Trees}
\label{sec:calculating_sigma}
For two sequences $S_1$ and $S_2$, 
let $S_1 \oplus S_2$ denote the concatenation of~$S_1$ and~$S_2$.
Given an ordered multi-tree $T$ 
on $n$ vertices indexed $v_1, v_2, \ldots, v_n$ as visited in a depth-first
traversal,
let $\delta(T) = (c_1, d_1, c_2, d_2, \ldots, c_n, d_n)$
be its color-depth sequence as defined in Section~\ref{sub:ordered-trees}.
For an integer $k \geq 1$ we define the \emph{$k$-shift} 
$\delta^k(T)$ of the sequence $\delta(T)$ to be the sequence 
$(c_1, d_1 + k, c_2, d_2 + k, \ldots, c_n, d_n + k)$
obtained by adding $k$ 
to each of the depth entries of~$\delta(T)$.

Let $T$ be an ordered multi-tree rooted at a vertex $r$,
and let $u_1, u_2, \ldots, u_{\deg(r)}$ denote
the children of $r$ indexed according to their left-to-right ordering.
Let $\mm(u_i)$ denote $\mul(u_i, \parent(u_i)) \oplus \mm(T(u_i))$ for each $i \in [1,\deg(r)]$.
Given the signatures $\sigma(T(u_i)) = (\delta(T(u_i)), \mm(T(u_i)))$ for all $i \in [1,\deg(r)]$, 
we devise a way to represent $\sigma(T)$ by $\sigma(T(u_i))$ via the following observation.

\begin{observation}
\label{obs:representation}
Let $T$ be an ordered multi-tree rooted at a vertex $r$, and let 
$\Ch(r) = \{u_1, u_2, \ldots, u_{\deg(r)} \}$ 
denote the set of children of $r$ indexed  according to their left-to-right ordering.
Given the sequences $\delta(T_{u})$ and $\mm(T_{u})$ 
for all $u \in \Ch(r)$,
for the sequences $\delta(T)$ and $\mm(T)$ it holds that:
\begin{align*}
\delta(T) &= (\col(r), 0) \oplus \delta^1(T(u_1)) \oplus  
\delta^1(T(u_2)) \oplus  \cdots \oplus  \delta^1(T(u_{\deg(r)})) ; \\
\mm(T) &= \mm(u_1) \oplus \mm(u_2) \oplus  \cdots \oplus \mm(u_{\deg(r)}). 
\end{align*}
\end{observation}


By Observation~\ref{obs:representation}, 
we show an algorithm to calculate the signature of a given rooted 
multi-tree in Procedure~{\procTreeSignature}.
Note that Procedure~{\procTreeSignature} calculates the signatures of all rooted subtrees
of a given tree.

\bigskip
{\bf Procedure~{\procTreeSignature}} {\sc SubTreeSignature}
\label{alg:signature_rmt}
\begin{algorithmic}[1]
\Require A $\Sigma$-colored multi-tree $T$
		    with multiplicity at most $d$ rooted at a vertex $r \in V(T)$.
\Ensure The signatures $\sigma(T_v)$ of each rooted tree $T_v$, $v \in V(T)$.
\State {$s := \emptyset$};
\For {{\bf each} $v \in V(T)$ in DFS-post order} \label{line:for1-begin}
	\If{ $v$ is a leaf}
		\State {$\delta[v] := (\col(v), 0); \mm[v] := \emptyset$}
	\Else
		\Statex {~~~/* The signatures $s[u]$ of all children of $v$ are already obtained */}
		\For{{\bf each} $u \in \ch(v)$} \label{line:for2-begin}
			\State {$\delta' :=$ 1-shift of $\delta[u]$}; \label{line:get-delta}
			\State {$\mm':= \mul(u, \parent(u)) \oplus \mm[u]$}; \label{line:get-mul}
			\State {$s'[u] := (\delta', \mm')$}
		\EndFor; \label{line:for2-end}
		\State {$S := (s'[u] \mid u \in \ch(v) )$};
		\State {Let $k :=  |\ch(v)|$};
		\Statex {~~~/* Represent $S$ as $( s_i = (\delta_i, \mm_i) \mid i \in [1, k] )$ */}
		\State {Sort $S$ in lexicographically descending order $\pi$}; \label{line:lex-sort}
		\State {$\delta[v] := (\col(v), 0) \oplus \delta_{\pi(1)} \oplus \delta_{\pi(2)} \oplus 
			      \cdots \oplus \delta_{\pi(k)}$};
		\State {$\mm[v] := \mm_{\pi(1)} \oplus \mm_{\pi(2)} \oplus \cdots \oplus \mm_{\pi(k)}$}
	\EndIf
	\State {$s[v] := (\delta[v], \mm[v])$}
\EndFor; \label{line:for1-end}
\State {\bf for each} $v \in V(T)$ {\bf output} $s[v]$ as $\sigma(T_v)$.
\end{algorithmic}
\bigskip

\begin{lem}
 \label{lem:complexity_of_signature}
 Given a $\Sigma$-colored rooted tree $T$ on $n$ vertices
 and multiplicity at most~$d$, 
 Procedure~{\procTreeSignature} computes the signatures $\sigma(T_v)$
 of all rooted subtrees $T_v$, $v \in V(T)$, of $T$
 in $O(n \cdot(n + |\Sigma| + d))$ time.
\end{lem}
\begin{proof}
Let $n_v$ denote the number of vertices in the subtree $T_v$ rooted at vertex~$v$,
and ${\rm d}_v$ the maximum depth of a leaf in the rooted tree $T_v$,
where~$v$ is taken to have depth~$0$, and it holds that ${\rm d}_v \leq n_v$.

The for-loop of lines~\ref{line:for1-begin} to~\ref{line:for1-end}
is executed for each vertex $v$ in $T$.
Since vertices are iterated in an DFS-post order,
in each iteration, the signatures $\sigma(T_u)$ are already computed for 
each child $u$ of $v$ in $T$.
Then, in the for-loop of lines~\ref{line:for2-begin} to~\ref{line:for2-end},
their signatures are gathered and the depth entries are offset by~1 in line~\ref{line:get-delta}.
This obviously takes at most $O(n_v)$ time.
Then, in line~\ref{line:lex-sort} the gathered sequences are sorted lexicographically.
The total length of the sequences is $O(n_v)$, and they are over the 
alphabets $\Sigma$ for the color of each vertex, $[1, {\rm d}_v]$ for the
depth, and $[1, d]$ for the multiplicity, thus the total alphabet size is
$|\Sigma| + d +  {\rm d}_v$.
By the algorithm for lexicographical sorting due to Aho {\em et al.}~\cite{AHU74},
the lexicographical sorting in line~\ref{line:lex-sort} takes~$O(n_v + |\Sigma| + d +  {\rm d}_v)$ time.
Finally, summing over all vertices $v$ in $T$,
for the computational complexity we get
\begin{align*}
 \sum_{v \in V(T)} O(n_v + |\Sigma| + d +  {\rm d}_v) 
 &= O(\sum_{v \in V(T)} (n_v + {\rm d}_v) + n \cdot (|\Sigma| + d) ) \\
 &= O(n^2 + n \cdot(|\Sigma| + d)),
\end{align*}
as required. 
\end{proof}


\subsubsection{Ranking of Rooted Trees}
\label{sec:ranking_of_trees}

Let $\mathcal{T}$ be a  finite set of rooted multi-trees,
and let $Z = \{\sigma(T) \mid T \in \mathcal{T} \}$
denote the set of signatures of the trees in $\mathcal{T}$.
We define a lexicographical order over $Z$
in the usual sense, i.e., 
for $\sigma_1 = (\delta_1, {\rm M}_1), \sigma_2=(\delta_2, {\rm M}_2) \in Z$
we write $\sigma_1 \prec \sigma_2$ if 
``$\delta_1 \prec \delta_2$'' or 
``$\delta_1 = \delta_2$ and ${\rm M}_1 \prec {\rm M}_2$.''
Then, we use the lexicographical order over the set $Z$ to
define a ranking $\rank_{\mathcal{T}} : \mathcal{T} \to [1, |Z|]$,
such that for two trees $T_1, T_2 \in \mathcal{T}$,
$\rank_{\mathcal{T}} (T_1) < \rank_{\mathcal{T}}(T_2)$ 
if $\sigma(T_1) \prec \sigma(T_2)$,
and $\rank_{\mathcal{T}}(T_1) = \rank_{\mathcal{T}}(T_2)$ 
means that $\sigma(T_1) = \sigma(T_2)$,
i.e. $T_1$ and $T_2$ are isomorphic.
It  follows that having a rank function over a set of multi-trees,
we can check whether two trees in the set 
are isomorphic to each other by comparing their ranks.

There exist algorithms reported in the literature
that calculate the rank of each subtree of a given tree~\cite{DIR99}
and rooted subgraph of an outerplanar graph~\cite{IN12}
in time linear in the number of vertices in the graph.
In our implementation we use simpler algorithms for this purpose
at the cost of a higher time complexity.

For a set $\mathcal{T}$ of rooted trees, let $\mathcal{T}^*$
denote the set of all rooted subtrees of trees in $\mathcal{T}$.
We give a procedure to calculate a ranking 
of a given set $\mathcal{T}$  of rooted trees
in Procedure~{3}.
By  Procedure~{2},
we in fact obtain a ranking in the set $\mathcal{T}^*$ at no additional cost.

\bigskip
{\bf Procedure~{\procTreeRanking}} {\sc TreeRanking}
\begin{algorithmic}[1]
\Require A set $\mathcal{T}$ 
		of $\Sigma$-colored rooted multi-trees with 
		multiplicity at most~$d$.
\Ensure A ranking function $\rank_{\mathcal{T}^*}$ of~$\mathcal{T}^*$.
  \State{$R:=\emptyset$; $h := |\mathcal{T}^*|$};
  \State{$S := (\sigma(T_i) \mid  i \in [1, h])$;  \label{line:get_sigma}
		/* Calculate $\sigma(T)$ by Procedure~{2} in Sec.~\ref{sec:calculating_sigma} */}
  \Statex{/* Treat $S = (s_1, s_2, \ldots, s_h)$ as an ordered set */}
  \State{Sort $S$ in lexicographically ascending order $\pi$;} \label{line:sort_all}
  \State{$R[T_{\pi(1)}] := 1$;} {$r := 1$};
  \For{ {\bf each} $i \in [2, h]$} \label{line:for3-begin}
	  \State{{\bf if} $s_{\pi(i-1)} \prec s_{\pi(i)}$ {\bf then} $r := r+1$ {\bf endif}}; \label{line:lexcompare}
	  \State{$R[T_{\pi(i)}] := r$;}	
  \EndFor; \label{line:for3-end}
  \State{{\bf output} $R$ as  $\rank_{\mathcal{T}^*}$}.
\end{algorithmic}
\bigskip


\begin{lem}
 \label{lem:ranking_all_subtrees}
 Let $\mathcal{T}$ be a given set of $\Sigma$-colored
 rooted multi-trees with multiplicity at most~$d$,
and let $n$ denote the total number of vertices
over trees in $\mathcal{T}$.
Then, the rank of each rooted subtree of all trees in $\mathcal T$ can be computed in
$O(n (n + |\Sigma| + d ))$ time in total.
\end{lem}
\begin{proof}
Let $\mathcal{T} = \{ T_1, T_2, \ldots, T_k \}$, and
let $n_i$, $i \in [1, k]$ denote the number of vertices in tree~$T_i$,
where $n = \sum_{i \in [1, k]} n_i$.
The signature $\sigma(T_i) = (\delta(T_i), \mm(T_i))$ of each
tree $T_i$ is a sequence with $O(n_i)$ entries
with $|\Sigma| + n_i + d$ possible values,
for the color of vertices and depth in a tree in $\delta(T_i)$,
and multiplicity with the parent in $\mm(T_i)$, respectively.

 By Lemma~\ref{lem:complexity_of_signature},
 computing the signatures of all rooted subtrees
 in line~\ref{line:get_sigma} takes 
 $O(n_i \cdot (n_i + |\Sigma| + d ))$ time for each tree $T_i$,
 and therefore $O(n \cdot ( n + |\Sigma| + d ))$ time in total.
 Now, each tree $T_i$ has $n_i$ rooted subtrees, 
 and the total number of vertices over these subtrees is $O(n_i^2)$.
 Therefore, the collection of signatures for the rooted subtrees
 of tree $T_i$ has in total $O(n_i^2)$ elements taking at
 most $|\Sigma| + n_i + d$ different values (alphabet size).
 Over all trees $T_i$, the elements of the subtree signatures
 take at most 
 $|\Sigma| + \max_{i \in [1, k]} \{n_i\} + d \leq |\Sigma| + n+ d$
 different values.
 Summing over all trees $T_i \in \mathcal{T}$,
 we get that the total length of the signatures over all subtrees is
 $\sum_{i \in [1, k]} O(n_i^2) = O(n^2)$.
Then, all these signatures can be lexicographically sorted in 
$O(n^2 + |\Sigma| + d)$ time~\cite{AHU74}, which is dominated by the time
to calculate the signatures.

Finally, having the lexicographically sorted signatures,
we assign rank to trees in a straightforward manner
by iterating over the sorted signatures as in lines~\ref{line:for3-begin} to~\ref{line:for3-end} in 
Procedure~{3},
and the claim follows. 
\end{proof}

\subsection{Signatures of 2-Cycle Trees}
In this subsection, we define codes that represent 
the topological structure of a bi-block 2-augmented tree
using the code $\sigma$ of rooted trees defined previously.  
We will demonstrate how to use these codes to uniquely 
determine the parent of a bi-block 2-augmented tree in the forthcoming section. 

Let $H$ be a bi-block 2-augmented tree with two edge-disjoint cycles $C$ and $C'$. 
For a vertex $v \in V(P(C, C')) \cup V(C) \cup V(C')$, 
let $\rank(H\langle v_j\rangle)$ denote the rank of $H\langle v_j\rangle$ 
among all pendent trees of $H$ 
as defined in Section~\ref{sec:ranking_of_trees}.
For an anchor vertex $v = \anc(C, H)$, we assume without loss of generality 
that the distinct anchor pairs $\{v, v_i\}, \{v, v_j\}$ satisfy 
$(\col(v_i), \deg(v_i), \rank(H\langle v_i\rangle))
\preceq (\col(v_j), \deg(v_j), \rank(H\langle v_j\rangle))$, 
and we define the \emph{heuristic code} of the anchor vertex $v$ to be
\begin{align*}
\code^*(v)  \triangleq &\, (\col(v), \deg(v), \rank(H\langle v\rangle),  \col(v_i), \col(v_j),  \\
 & \quad\deg(v_i), \deg(v_j), \rank(H\langle v_i\rangle), 
      \rank(H\langle v_j\rangle)).  
\end{align*}
For a path $P = v_1v_2\cdots v_p$ in $H$ such that each $v_i$ 
is in $V(P(C, C')) \cup V(C) \cup V(C')$, we define the code of the path $P$ to be 
\begin{align*}
\code_{\rm P}(P) \triangleq &\, (\rank(H\langle v_1\rangle), 
\mul(v_1 v_2), \rank(H\langle v_2\rangle), \\
& ~~~ \mul(v_2 v_3), \ldots, 
\mul(v_{p-1} v_p), \rank(H\langle v_p\rangle)). \nonumber
\end{align*}
For a cycle $C=v_1v_2\cdots v_pv_1$ in $H$, where $v_1$ is the anchor vertex $\anc(C, H)$, 
let $P$ be the path $v_1v_2\cdots v_p$, and let $P_{\rev}$ be the path $v_1v_pv_{p-1}\cdots v_2$.
We define the code of the cycle $C$ to be $\code_{\rm C}(C) \triangleq \min\{\code_{\rm P}(P), \code_{\rm P}(P_{\rev})\}$, 
and we define the \emph{heuristic code} of the cycle $C$ to be 
\begin{align*}
 \code_{\rm C}^*(C) \triangleq &\, (|V(H)| - |V(Q(C, H))|, |C|, \\
    &\code^*(\anc(C, H)), \code_{\rm P}(P(C, C')), \code_{\rm C}(C))).
\end{align*}

\section{A Parent-Child Relationship Between \\Monocyclic Graphs 
			and Bi-block 2-Augmented Trees}
\label{sec:parent-child-relationship}

Let $H$ be a bi-block 2-augmented tree with cycles $C$ and $C'$, 
and let $u = \anc(C, H)$ and $v = \anc(C', H)$ 
denote the anchor vertices of~$H$.
Without loss of generality, we assume that $\code_{\rm C}^*(C) \preceq \code_{\rm C}^*(C')$ holds.
In addition, let $C =u_1 u_2 \cdots u_p u_1$, where $u_1 = u$, 
let $P$ be the path $u_1 u_2 \cdots u_p$, and let $P_{\rev}$ 
be the path $u_1 u_p u_{p-1} \cdots u_2$.
We assume that $\code_{\rm P}(P) \preceq \code_{\rm P}(P_{\rev})$ holds.
We define \emph{the parent} of $H$ to be the graph $H - u u_2$, 
that is, the graph obtained by deleting the edge between the anchor pair $\{u, u_2\}$.

\subsection{Necessary and Sufficient Conditions for Generating Children}
\label{sec:child_conditions}

Let $G$ be a monocyclic chemical graph with a cycle $C$, 
let $H$ be the graph obtained by adding an edge $xy$ between 
non-adjacent vertices $x, y \in V(G)$, and let $C'$ be the cycle containing the edge $xy$ in $H$.
Now, we devise a necessary and sufficient condition to determine whether $H$ is a child of $G$ or not.

\begin{lem}
\label{lem:child_condition}
Let $G$ be a monocyclic chemical graph and let $C$ denote the unique cycle of $G$.
Let $x$ and $y$ be two non-adjacent vertices in $V(G)$, and $q$ be a positive integer
such that $q \leq \min \{\res(x), \res(y)\}$.
Let $H$ denote the graph $G + q \cdot xy$, and let $C'$ be a cycle containing the edge $xy$ in $H$.
Then, $H$ is a child of $G$ if and only if the following conditions are satisfied:
\begin{description}
\item[(i)] $y \in V(G \langle x \rangle)$;
\item[(ii)] $\code_{\rm C}^*(C') \preceq \code_{\rm C}^*(C)$; and
\item[(iii)] For the path $P = v_1v_2\cdots v_p$ such that 
$v_1 = x$, $v_2 = y$, and $\{v_1, v_p\}$ is the other anchor pair of $C'$, and
$P_{\rev} = v_1v_pv_{p-1}\cdots v_2$, it holds that $\code_{\rm P}(P) \preceq \code_{\rm P}(P_{\rev})$.
\end{description} 
\end{lem}
\begin{proof}
{\bf Necessity}.
We prove that if the pair $\{x, y\}$ of vertices does not 
satisfy any of the conditions (i) to (iii) of the lemma, then
the graph $H = G + xy$ is not a child of $G$.
Suppose that $\rootG{x} \neq \rootG{y}$ holds.
In this case, the graph $G + xy$ will not be a bi-block 2-augmented tree. 
Hence, in order to obtain a bi-block 2-augmented tree, $\rootG{x} = \rootG{y}$ must hold.
From the definition of parent, in order to obtain the parent of $H$, 
an edge between an anchor pair must be deleted. 
Since the anchor vertex $\anc(C', H)$ is a common ancestor of $x$ and $y$, 
if $x$ (resp., $y$) is not an ancestor of $y$ (resp., $x$), then the pair $\{x, y\}$ is not an anchor pair.
Hence, in this case, $H = G + xy$ is not a child of $G$.
As a result, without loss of generality, we assume that $x$ is an ancestor of $y$.
Now, the graph $H = G + xy$ is a bi-block 2-augmented tree. 

Suppose that $\code_{\rm C}^*(C) \prec \code_{\rm C}^*(C')$ holds.
The parent of $H$ is obtained by deleting an edge in the cycle $C$.
This implies that the parent of $H$ is not isomorphic to $G$, and we get that $H$ is a child of $G$.
Hence $\code_{\rm C}^*(C') \preceq \code_{\rm C}^*(C)$ must be satisfied.

Finally, for the path $P = v_1v_2\cdots v_p$ 
such that $v_1 = x$, $v_2 = y$, and $\{v_1, v_p\}$ is the other anchor pair of $C'$, 
let $P_{\rev} = v_1v_pv_{p-1}\cdots v_2$.
If $\code_{\rm P}(P_{\rev}) \prec \code_{\rm P}(P_{\rev})$ holds, 
then from the definition of the parent of a bi-block 2-augmented tree, 
the parent of $H$ is obtained by deleting the edge $xv_p$, and it is not isomorphic to $G$.
This contradicts that $H$ is a child of $G$.

{\bf Sufficiency}.
Assume that $G$ is indeed the parent of the 2-cycle graph $H = G + q \cdot xy$.
By the definition of the parent-child relationship, $\{x, y\}$ is an anchor pair
in $H$, and by the preceding arguments, this is only possible if Condition~(i)
of the lemma is satisfied.
Conditions~(ii) and~(iii) are trivially satisfied by the definition of the parent-child relationship, 
from which the claim follows.
\end{proof}

From Condition~(ii) of Lemma~\ref{lem:child_condition}, we have the following theorem.
\begin{theorem}
\label{thm:unique_subtree}
If a monocyclic chemical graph $G$ with $n$ vertices has at least one child bi-block 2-augmented tree, 
then there exists a unique pendent tree $T \subseteq G$ with at least $n/2$ vertices.
\end{theorem}

In order to give a proof of Theorem~\ref{thm:unique_subtree}, we show three lemmas.

\begin{lem}
\label{lem:anchor_vertices}
Let $G$ be a monocyclic chemical graph.
For vertices $x \in V(G)$ and a descendant $y \in V(G\langle x \rangle)$ 
of $x$ such that $xy \not\in E(G)$, 
the anchor vertices of $G + xy$ are $x$ and $\rootG{x}$.
\end{lem}
\begin{proof}
Let $H$ denote the graph $G + xy$.
Since we have $\rootG{x} = \rootG{y}$, the graph $H$ is a bi-block 2-augmented tree.
The $x, y$-path in $G$ and the edge $xy$ form a cycle in $H$.
The other cycle of $H$ is the cycle in $G$.
The path connecting the two cycles in $H$ is the $x, \rootG{x}$-path  in $G$.
Hence, the anchor vertices are~$x$ and~$\rootG{x}$.
\end{proof}

\begin{lem}
\label{lem:vertices_Q}
Let $G$ be a monocyclic chemical graph, and let $C$ denote the cycle in $G$.
Let $x$ be a vertex in $G$, let $y \in V(G\langle x \rangle)$ such that
$xy \notin E(G)$ be a descendant of $x$, 
and let $v$ denote the child of $\rootG{x}$ such that $y \in V(G\langle v \rangle)$.
Let $H$ denote the graph $G + xy$, and let $C'$ denote the cycle in $H$ containing the edge $xy$.
Then, it holds that $|V(Q(C, H))| = |V(G)| - |V(G \langle v \rangle)|$ 
and $|V(Q(C', H))| = |V(G \langle x \rangle )|$.
\end{lem}
\begin{proof}
From Lemma~\ref{lem:anchor_vertices}, the anchor vertices of $H$ are $x$ and $\rootG{x}$.
First, observe that  $|V(Q(C', H))| = |V(G\langle x \rangle)|$
is the number of vertices reachable from $x$ in $G - (E(P(C, C')) \cup E(C))$.
On the other hand, $|V(Q(C, H))|$ is the number of vertices reachable
from $\rootG{x}$ in $G - (E(P(C, C')) \cup E(C'))$.
Since there is no vertex in $V(G \langle v \rangle)$ that is reachable from $\rootG{x}$ 
in the graph $H - (E(P(C, C')) \cup E(C'))$, 
it follows that $|V(Q(C, H))| = |V(G)| - |V(G\langle v \rangle)|$.
\end{proof}

\begin{lem}\label{lem:child_big_tree}
Let $G$ be a monocyclic chemical graph, and let $C$ denote the cycle in $G$.
For a vertex $v \in V(C)$ with $|V(G\langle v \rangle)| < |V(G)|/2$, 
let $x$ be a vertex in $V(G\langle v \rangle)$, 
let $y$ be a vertex in $V(G\langle x \rangle)$, such that $xy \notin E(G)$, 
let $H$ denote the graph $G + xy$, and let $C'$ denote the cycle in $H$ containing the edge $xy$.
Then, it holds that $|V(Q(C', H))| < |V(Q(C, H))|$, and therefore $G + xy$ is not a child of $G$.
\end{lem}
\begin{proof}
For simplicity, let $n$ denote $|V(G)|$.
Let $v'$ denote the child of $v$ in $H \langle v \rangle$ 
such that $H \langle v' \rangle$ contains~$y$.
From Lemma~\ref{lem:vertices_Q}, we have 
$|V(Q(C, H))| = n - |V(G\langle v' \rangle)|$ 
and $|V(Q(C', H))| = |V(G\langle x \rangle)|$.
Since it always holds that 
$|V(G\langle v' \rangle)| < |V(G\langle v \rangle)|$ and 
$|V(G\langle x \rangle)| \leq |V(G\langle v \rangle)|$, 
we get that $|V(Q(C, H))| =  n - |V(G\langle v_y \rangle)| > n/2$ 
and $|V(Q(C', H))| = |V(G\langle x \rangle)| < n/2$.
Therefore, it follows that $|V(Q(C', H))| < |V(Q(C, H))|$.

If $|V(Q(C', H))| < |V(Q(C, H))|$ holds, then $\code_{\rm C}^*(C') \preceq \code_{\rm C}^*(C)$ always holds.
This contradicts Condition~(ii) of Lemma~\ref{lem:child_condition}.
As a result, we conclude that the graph $G + xy$ is not a child of $G$.
\end{proof} 

Now, we are ready to give a proof of Theorem~\ref{thm:unique_subtree}.
\begin{proof}[Proof of Theorem~\ref{thm:unique_subtree}]
Suppose that a monocyclic chemical graph $G$ has no pendent tree 
with at least $|V(G)|/2$ vertices, then, from Lemma~\ref{lem:child_big_tree}, $G$ has no children.

We prove that any monocyclic chemical graph $G$ has at most one pendent tree with at least $|V(G)|/2$ vertices
by a straightforward observation.
Suppose that $T_1$ and $T_2$ are distinct pendent trees of $G$, each with at least $|V(G)|/2$ vertices.
From the definition of cycle, a cycle contains at least 3 vertices, and
therefore there exists at least one vertex in the cycle of $G$ 
that is neither in $T_1$ or $T_2$.
Hence, it must hold that 
$|V(G)| \geq |V(T_1)| + |(V(T_2)| + 1 \geq |V(G)|/2 + |V(G)|/2 + 1 = |V(G)| + 1$, a contradiction.

Therefore, we conclude that if a monocyclic chemical graph $G$ has at least one child, 
then there exists a unique pendent tree in $G$ with at least $|V(G)|/2$ vertices.
\end{proof}

From Theorem~\ref{thm:unique_subtree}, we have the following
necessary and sufficient conditions for generating the children of a given monocyclic chemical graph.
\begin{lem}\label{lem:child_condition2}
Let $G$ be a monocyclic chemical graph and let $C$ denote the cycle of $G$.
Let $x$ and $y$ be non-adjacent vertices in $V(G)$, 
let $H$ denote the graph $G + xy$, and let $C' \subseteq H$ be a cycle containing the edge~$xy$.
Then $H$ is a child of $G$ if and only if the following conditions are satisfied:
\begin{description}
\item[(i)] $G$ has a unique pendent tree with at least $|V(G)|/2$ vertices;
 \item[(ii)] $y \in V(G \langle x \rangle)$;
\item[(iii)] $\code_{\rm C}^*(C') \preceq \code_{\rm C}^*(C)$; and
\item[(iv)] For the path $P = v_1v_2\cdots v_p$ such that
$v_1 = x$, $v_2 = y$, and $\{v_1, v_p\}$ is the other anchor pair of $C'$, 
and $P_{\rev} = v_1v_pv_{p-1}\cdots v_2$, it holds that $\code_{\rm P}(P) \preceq \code_{\rm P}(P_{\rev})$.
\end{description} 
\end{lem}

We devise an algorithm to determine whether
for a monocyclic graph $G$ and a pair $\{x, y\} \subseteq V(G)$
of nonadjacent vertices the graph $G + xy$ 
is a child of $G$ or not according to Lemma~\ref{lem:child_condition}.
First, we check whether $\code_{\rm C}^*(C') \preceq \code_{\rm C}^*(C)$ or not.
Then,  for the path $P = v_1v_2\cdots v_p$ such that $v_1 = x$, $v_2 = y$, 
and $\{v_1, v_p\}$ is the other anchor pair of $C'$, 
and the reverse path $P_{\rev} = v_1v_pv_{p-1}\cdots v_2$, 
we check whether $\code_{\rm P}(P) \preceq \code_{\rm P}(P_{\rev})$ or not.
We show an algorithm to determine whether $G + xy$ 
is a child of $G$ or not 
as Procedure~{\procChildCheck} \textsc{ChildCheck}.

\bigskip

{\bf Procedure~{\procChildCheck}}
\textsc{ChildCheck}
\begin{algorithmic}[1]
\Require A monocyclic chemical graph $G$ with a cycle $C$ 
and a non-adjacent vertex pair $\{x, y\} \subseteq V(G)$.
\Ensure If $G + xy$ is a child of $G$, then ``Yes'', otherwise ``No''.
\Statex{/* Let $C'$ denote the cycle of $G + xy$ such that $\{x, y\} \subseteq V(C')$ */}
\State ${\rm Answer} := $``No'';
\If{$\code_{\rm C}^*(C') \preceq \code_{\rm C}^*(C)$}
	\State{$P := v_1v_2\cdots v_p$ such that $v_1 = x$, $v_2 = y$, 
	  and $\{v_1, v_p\}$ is the other anchor pair of $C'$};
	\State{$P_{\rev} := v_1v_pv_{p-1}\cdots v_2$};
	\If{$\code_{\rm P}(P) \preceq \code_{\rm P}(P_{\rev})$}
		\State{${\rm Answer} := $ ``yes''}
	\Else
		\State{${\rm Answer} := $ ``no''}
	\EndIf
\Else
	\State{${\rm Answer} := $ ``no''}
\EndIf;
\State {\bf output} ${\rm Answer}$.
\end{algorithmic}

\section{Avoiding Intra-Duplication}
\label{sec:avoid-symmetrical-duplication}

In this section, we introduce a way to avoid 
intra-duplication, that is
generating 
 isomorphic  bi-block 2-augmented trees
 from a single monocyclic chemical graph. 
 
Suzuki~{\em et al.}~\cite{Suzuki14} introduced 
an ``admissibility'' condition, a necessary and sufficient condition
for a pair of nonadjacent vertices to be joined by a  new edge to generate a monocyclic chemical graph 
 from a single multi-tree so that no pair of isomorphic monocyclic chemical graphs will be generated
as long as only ``admissible'' pairs of nonadjacent vertices receive new edges.
We first review the admissibility condition in terms of generating monocyclic chemical graphs. 
Next we explain that this condition can be applied to pairs of nonadjacent vertices
 of monocyclic chemical graphs to avoid generating 
isomorphic bi-block 2-augmented trees.  


Recall that for a tree $T$ and a pair $\{x, y\}$ of  vertices in $V(T)$, 
  $T+xy$ denotes the graph obtained from $T$ by adding a single edge~$xy$.
When $T$ is rooted at a vertex~$r$, we also regard $T+xy$ as a graph rooted at ~$r$.
The next lemma states that if two isomorphic monocyclic chemical graphs are generated 
from a left-heavy multi-tree rooted at its centroid 
by adding edges to different pairs of nonadjacent vertices, 
then these monocyclic chemical graphs are rooted-isomorphic. 

\begin{lem}\label{lem:riso}{\rm (Suzuki }{\em et al.~}{\rm \cite{Suzuki14})}
Let $T$ be a left-heavy tree rooted at its centroid~$c_T$.
For any  two pairs   $ \{s, t\}, \{x, y\}  \subseteq V(T)$ of nonadjacent vertices, 
 if $T+st \approx  T+xy$ then $T+st \riso  T+xy$. 
\end{lem}

\noindent Thus we only need to be aware of rooted isomorphisms in order to avoid the duplication 
by adding edges between different pairs of nonadjacent vertices to a multi-tree. 
 
Let $T$ be a multi-tree rooted at a vertex $r$. 
See Section~\ref{sub:ordered-trees} for the definition of 
$\mathrm{dfs}()$, $\mathrm{left}()$, $\mathrm{lca}()$ and $\mathrm{gua}()$ in $T$. 
To facilitate finding pairs of nonadjacent vertices to be joined to generate monocyclic chemical graphs that are not
rooted-isomorphic to each other, 
Suzuki~{\em et al.}~\cite{Suzuki14} 
%
%
introduced a 0-1 function $\mathrm{copy}$ as follows:
For each non-root vertex $v \in V(T)$, let 
\begin{equation}
\mathrm{copy}_T(v)= \left\{ 
\begin{array}{ll}
1{\ \ }$if$ {\ }\mathrm{left}_T(v){\ }$exists${\ } $and${\ }T_{v\parent(v)} \riso T_{\mathrm{left}_T(v)\parent(v)}\\
0{\ \ }{\rm otherwise}.
\end{array}
\right. 
\end{equation}
Then 
Suzuki~{\em et al.}~\cite{Suzuki14} defined an admissibility condition
 in terms of generating monocyclic chemical graphs as follows: 
\begin{itemize}
\item[{\rm (a-1)}]  For each vertex $w \in V(P_T(\mathrm{lca}_T(u, v), r))$, 
it holds $\mathrm{copy}_T(w)=0$;   
\item[{\rm (a-2)}]   For  each vertex  
$w \in V(P_T(u, \mathrm{gua}_T(u, v)))$ $\cup$ 
$V(P_T(v, \mathrm{gua}_T(v, u))) \setminus \{\mathrm{gua}_T(u, v)\}$, 
 it holds $\mathrm{copy}_T(w)=0$; and
\item[{\rm (a-3)}] If $\mathrm{copy}_T(\mathrm{gua}_T(v, u))=1$, then 
\begin{itemize}
\item[{\rm (i)}] $\mathrm{gua}_T(u, v)=\mathrm{left}_T(\mathrm{gua}_T(v, u))${\rm ;} 
\item[{\rm (ii)}] The copy $\hat{u}$ of vertex $u$ in $T_{\mathrm{gua}(u, v)}$ satisfies   
    $\mathrm{dfs}(u) \geq \mathrm{dfs}(\hat{u})$ \\
~~ $(=\mathrm{dfs}(u)+|V(T_{\mathrm{gua}(u, v)})|)$.
\end{itemize}
\end{itemize}

We call these conditions (a-1), (a-2) and (a-3) the {\em admissibility conditions}. 
For a left-heavy multi-tree $T$, we call an ordered pair $(u, v)$ 
 of nonadjacent vertices $u, v\in  V(T)$
 which satisfies  the admissibility conditions an {\em admissible pair}.   
For a left-heavy tree rooted at a vertex $r$, 
we define $A(T)$ to be the set of  admissible pairs  $(x, y)$ in $T$.
When a tree is rooted at its centroid $c_T$, 
Suzuki~{\em et al.}~\cite{Suzuki14} have shown the following claim. 

\begin{lem}\label{lem:ad}{\rm (Suzuki }{\em et al.~}{\rm \cite{Suzuki14})}
 Let $T$ be a left-heavy tree rooted at its centroid~$c_T$. 
\begin{enumerate}
\item[{\rm (i)}]
For any pair $\{s, t\} \subseteq V(T)$ of nonadjacent vertices, 
there exists an admissible pair $(x, y)\in A(T)$   
such that $T+st \riso T+xy$; and
\item[{\rm (ii)}]
For two   admissible pairs $(s, t), (x, y)\in A(T)$, 
 if $T+st \riso T+xy$ then it holds that $(s, t)=(x, y)$. 
\end{enumerate}
\end{lem}

Lemma~\ref{lem:riso} and Lemma~\ref{lem:ad} indicate that  adding only the edges 
between each admissible pair never generates duplications of 
monocyclic chemical graphs from a single multi-tree.

\begin{figure}[htbp]
\begin{center}
\includegraphics[width=12cm]{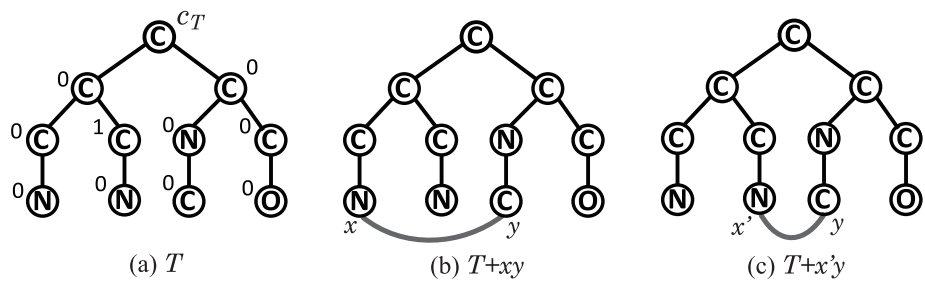}
\caption{An illustration of monocyclic chemical graphs obtained by adding an edge between an admissible pair 
and a non-admissible pair, respectively:
 (a) a left-heavy multi-tree $T$ rooted at its centroid $c_T$, where 
the integer beside each vertex $v$ of $T$ represents the value of $\mathrm{copy}_T(v)$;
 (b) a monocyclic chemical graph obtained by adding a new edge to an admissible pair $(x, y)$; and
 (c) a monocyclic chemical graph obtained by adding  a new edge to a non-admissible pair $(x', y)$. }
\label{figadmissible}
\end{center}
\end{figure}

Figure~\ref{figadmissible}(a) illustrates a left-heavy tree $T$ rooted at its centroid $c_T$. 
Figure~\ref{figadmissible}(b) (resp., (c)) provides an example of a monocyclic chemical graph
 obtained by adding an edge between
  an admissible pair $(x, y)$ (resp.,  a non-admissible pair $(x', y)$), 
where $(x', y)$ is not admissible 
since there is a vertex $v$ such that $\mathrm{copy}_T(v) = 1$ in $P_T(x', c_T)$.

For a $\Sigma$-colored labeled multi-graph $H$ and a set $A$ of ordered pairs of vertices in $V(H)$, 
let $\mathcal{G}(H, A)$ be the set of $\Sigma$-colored labeled multi-graphs $H+xy$
for all $(x, y)\in A$. 
We show that Lemma~\ref{lem:ad} remains valid
for any choice of root in $T$
 restating properties (i) and (ii) in the lemma
as follows.

\begin{lem}\label{lem:ad-any-root} 
 Let $T$ be a left-heavy tree rooted at a vertex $r$.
Then for any pair $\{s, t\} \subseteq V(T)$ of nonadjacent vertices, 
exactly one monocyclic chemical graph $T+xy$ in $\mathcal{G}(H, A(T))$ is isomorphic to $T+st$. 
\end{lem}

\begin{proof} 
Note that $r$ is not necessarily the centroid of $T$. 
To prove this lemma, it suffices to show that,  for any pair $\{s, t\} \subseteq V(T)$ 
of nonadjacent vertices, 
there is exactly one monocyclic chemical graph in $\mathcal{G}(H, A(T))$ isomorphic to $T +st$. 

We modify $T$ by adding some vertices to $r$ to make $r$ the centroid 
$c_{T^*}$ of the resulting tree $T^*$, 
where the left-heavy structure of $T$ and the label of each vertex $v \in V(T)$ are kept. 
We introduce an imaginary color {\tt Ic} such that $\mathrm{val}({\tt Ic}) = 2$ and
 for any color $c \in \Sigma$, it holds that ${\tt Ic} < c$. 
Let $T_1$ be a path that consists of $|V(T)|-1$ vertices of color~{\tt Ic}. 
Let $T^*$ be the resulting ordered-tree obtained by joining $T$ and $T_1$ 
with a new edge $ru$ for an endvertex $u$ of path $T_1$,  and $n^*=|V(T^*)|$, 
where we regard $u$ as the rightmost child of $r$ in $T^*$. 
Clearly the centroid $c_{T^*}$ of $T^*$  is $r$. 
We see that 
the current $\mathrm{dfs}$ order along $T^*$ gives an alternating sequence
 $\delta(T^*)=(c_1, d_1, \ldots, c_{n^*}, d_{n^*})$
of color $c_i$ and depth $d_i$ of the $i$-th vertex
 which is lexicographically maximum among all
other ordered trees of $T^*$, 
since for any color $c \in \Sigma$, it holds that ${\tt Ic} < c$.
Therefore $T^*$ is a left-heavy tree, as required.
\end{proof}

\begin{figure}[htbp]
\begin{center}
\includegraphics[width=11cm]{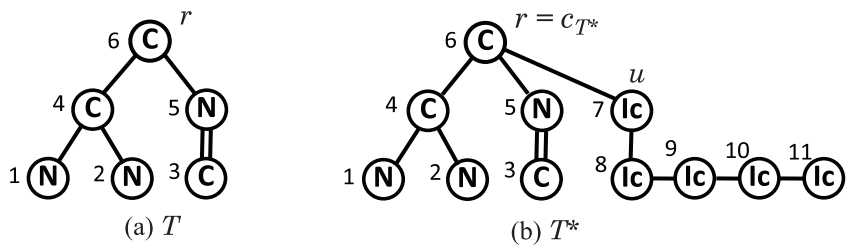}
\caption{An illustration of rooted multi-trees; 
The integer beside each vertex $u$ represents the label $\mathrm{lab}(u)$ of $u$. 
Note that the labeling of the vertices in $T$ is kept in $T^*$. 
 (a) A multi-tree $T$ rooted at a vertex~$r$; and
 (b) A left-heavy multi-tree $T^*$ rooted at its centroid $r = c_{T^*}$ obtained by adding a path $T_1$ 
which consists of $|V(T)| - 1$ vertices of color {\tt Ic} to the rightmost side of~$T$. }
\label{fig-modified-tree}
\end{center}
\end{figure}

Figure~\ref{fig-modified-tree}(a) illustrates a labeled left-heavy multi-tree $T$ rooted at a vertex $r$, 
and \ref{fig-modified-tree}(b) illustrates a labeled left-heavy multi-tree $T^*$
 rooted at its centroid $r = c_{T^*}$.

Suzuki~\textit{et al}.~\cite{Suzuki14} have shown that for a rooted left-heavy tree $T$ 
and a pair $\{x, y\} \subseteq V(T)$ of nonadjacent vertices, 
we can check the admissibility conditions in $\mathcal{O}(1)$ time 
with an adequate data structure. 

\section{Preprocessing for Efficient Computation}

\subsection{Ranking of Pendent Trees in a 2-Cycle Tree}
\label{sec:ranking_pendent_trees_in_2A2B}
For a monocyclic chemical graph $G$, let $\mathcal{H}(G)$
denote the set of all bi-block 2-augmented tree children of $G$, 
and we define the set $\mathcal{T}_{\mathcal{H}}(G)$ 
to be the set of all pendent trees of bi-block 2-augmented trees in $\mathcal{H}(G)$.
In addition, for each $u \in V(G) \setminus V(C)$
we define $T_{G}(u)$ to be the graph 
obtained by deleting the subgraph containing $u$ from $G\langle \parent(u) \rangle$, 
that is, $T_{G}(u) \triangleq  G\langle \parent(u) \rangle - (E(G\langle u \rangle) \cup u \parent(u))$.
Given a monocyclic chemical graph $G$, we devise an algorithm 
to calculate a ranking among all the pendent trees in 
$\mathcal{T}_{\mathcal{H}}(G)$ following the lexicographically
ascending order of signatures for rooted multi-trees 
defined in Sec.~\ref{sec:signature}.
First, we show an important property in the following lemma.

\begin{lem}
\label{lem:property_of_T_H}
Let $G$ be a monocyclic chemical graph, with a unique cycle $C$.
If there exists a vertex $v \in V(C)$ such that
$|V(G\langle v \rangle)| \geq |V(G)|/2$, 
then it holds that
\[
\mathcal{T}_{\mathcal{H}}(G) = 
			\{G\langle u \rangle \mid u \in V(C) \setminus \{v\} \} \cup 
			\{T_{G}(u) \mid u \in G\langle v \rangle \setminus \{v\}\}.
\]
Otherwise, it holds that 
\[
 \mathcal{T}_{\mathcal{H}}(G) = \emptyset.
\]
\end{lem}
\begin{proof}
By Lemma~\ref{lem:child_condition2}, 
unless the monocyclic chemical graph $G$ has a vertex $v \in V(C)$
such that $|V(G\langle v \rangle)| \geq |V(G)|/2$ it has
no children bi-block 2-augmented trees,
 and the latter part of the lemma follows.
Therefore, assume that there exists such a vertex $v \in V(C)$, 
let $x \in V(G\langle v \rangle)$ and $y \in G \langle x \rangle$
be two vertices, 
and $H = G + xy$.
In this case, the graph $H$ is a bi-block 2-augmented tree.
Let $P_{vy}$ denote the rooted tree induced by the vertices in the $v, y$-path in $H$ rooted at $v$.
Then, for any vertex $u \in V(P_{vy})$, we get that $u$ is the root of a pendent tree in $H$.
By a simple observation, we get that $G\langle y \rangle$ is preserved as $H \langle y \rangle$.
Let $u$ be a vertex in $V(P_{vy}) \setminus \{y\}$, and let $u_c$ denote the child of $u$ in $V(P_{vy})$.
The pendent tree $H\langle u \rangle$ is the graph obtained by deleting
the subgraph containing $u_c$ from $G\langle u \rangle$.
Finally, since none of the pendent trees of $G$ rooted at 
vertices other than $v$ are modified in $H$, it holds that
$\{G\langle v \rangle \mid v \in V(C) \setminus \{v\}\} \subseteq  \mathcal{T}_{\mathcal{H}}(G)$, 
from where the claim follows.
\end{proof}

Given a monocyclic chemical graph $G$, 
we obtain the signatures of all pendent trees in 
$\mathcal{T}_{\mathcal{H}}(G)$  by Procedure~{\procTreeSignature}.
Recall that for a collection of trees, Procedure~{\procTreeSignature}
in fact calculates the signature of each rooted subtree
for all trees in the collection.
Then, we obtain their ranking by Procedure~{\procTreeRanking}.

\subsection{Identifying Potential Edges}
\label{sec:potential_edges}

Let $G$ be a monocyclic chemical graph.
For two non-adjacent vertices $x$ and $y$ in $V(G)$, if one of the conditions in
Lemma~\ref{lem:child_condition}
is not satisfied, then the graph $G + xy$ is not a child of $G$.
For a vertex $x \in V(G)$, and a descendant $y$ of $x$ such that $xy \not\in E(G)$, 
let $H$ denote the graph $G + xy$, 
and let $C'$ denote the cycle in $H$ containing the edge $xy$. 
We say that the vertex pair $\{x, y\}$ is a \emph{potential edge}
if $(|V(H)| - |V(Q(C', H))|, |C'|) \preceq (|V(H)| - |V(Q(C, H))|, |C|)$ holds.
We define $E_{\pot}(G)$ to be the set of all potential edges in $G$.
Obviously, $E_{\pot}(G)$ contains all pairs $\{x, y\} \subseteq V(G)$
of non-adjacent vertices such that $G + xy$ is a child of~$G$.

\begin{lem}\label{lem:epot_child}
For a monocyclic chemical graph $G$ and a vertex pair $\{x, y\} \not\in E_{\pot}(G)$, 
the graph $G + xy$ is not a child of $G$.
\end{lem}
\begin{proof}
From the definition of a potential edge, for a vertex pair $\{x, y\} \not\in E_{\pot}(G)$ of $G$, 
at least one of the following properties holds: \\
(i)~ $xy \in E(G)$, \\
(ii)~$x$ is not an ancestor of $y$, and \\
(iii)~$(|V(H)| - |V(Q(C, H))|, |C|) \prec (|V(H)| - |V(Q(C', H))|, |C'|) $.

Suppose that either $xy \in E(G)$ holds or $x$ is not an ancestor of $y$.
From Lemma~\ref{lem:child_condition}, the graph $G + xy$ is not a child of $G$.
Hence, we assume that $xy \not\in E(G)$ and $x$ is an ancestor of $y$.
Now, the graph $G+xy$ is a bi-block 2-augmented tree whose anchor vertices 
are $x$ and $\rootG{x}$ by Lemma~\ref{lem:anchor_vertices}, 
and we denote the bi-block 2-augmented tree by $H$.
Let $C$ denote the cycle in $G$, and let $C'$ denote the cycle in $H$ containing the edge $xy$.
Suppose that $(|V(H)| - |V(Q(C, H))|, |C|) \prec (|V(H)| - |V(Q(C', H))|, |C'|) $ holds.
Since the two leftmost entries of $\code^*_C(C')$ 
(resp., $\code^*_C(C))$ are $|V(H)| - |V(Q(C', H))|$ and  
$|C'|$ (resp., $|V(H)| - |V(Q(C, H))|$ and $|C|$), 
we have $\code_{\rm C}^*(C) \prec \code_{\rm C}^*(C')$.
However, from Lemma~\ref{lem:child_condition}, the graph $G + xy$ is not a child of~$G$. 

As a result, we conclude that for a monocyclic chemical graph $G$ 
and a vertex pair $\{x, y\} \not\in E_{\pot}(G)$, 
the graph $G + xy$ is not a child of $G$.
\end{proof}

Note that, for a vertex pair $\{x, y\} \in E_{\pot}(G)$, the graph $G + xy$ might  not be a child of~$G$.
We construct an algorithm to generate the set of potential edges of a given monocyclic chemical graph.
This can be achieved by checking all vertex pairs in~$G$, 
however, in this section, we propose an efficient algorithm by skipping unnecessary vertex pairs.
In order to develop an efficient algorithm, we show some useful lemmas.

\begin{lem}
\label{lem:epot_big_tree}
Let $G$ be a monocyclic chemical graph.
If $G$ has no pendent tree that contains at least half the number of vertices in $G$, 
then $E_{\pot}(G) = \emptyset$.
\end{lem}
\begin{proof}
Let $n$ denote the number of vertices in $G$.
For a pendent tree $T$ with less than $n/2$ vertices, and a vertex pair in $V(T)$, 
from Lemma~\ref{lem:child_big_tree}, such a vertex pair 
contradicts Condition~(ii) of Lemma~\ref{lem:child_condition}.
Hence, we conclude that if $G$ has no pendent tree with 
at least half the number of vertices of $G$, then $E_{\pot}(G) = \emptyset$ holds.
\end{proof}

Given a monocyclic chemical graph $G$, 
for a vertex $x \in V(G)$ and a vertex $x' \in V(G\langle x \rangle)$, 
since it always holds that 
$|V(G\langle x \rangle)| \geq |V(G\langle x' \rangle)|$, 
we have the following lemma.

\begin{lem}
\label{lem:skipping_x}
Let $G$ be a monocyclic chemical graph, 
and let $C$ denote the cycle in~$G$.
For a vertex $v \in V(C)$, a child $v'$ of $v$,
and a vertex $x \in V(G\langle v' \rangle)$, 
if $|V(G\langle x \rangle)| < |V(G)| - |V(G\langle v' \rangle)|$ holds, 
then there are no vertices $x', y' \in V(G\langle x \rangle)$ such that $\{x', y'\} \in E_{\pot}(G)$.
\end{lem}
\begin{proof}
Suppose that there exists a vertex pair $\{x', y'\} \subseteq V(G\langle x \rangle)$ in $E_{\pot}(G)$.
From the definition of $E_{\pot}(G)$, $y'$ is a descendant of $x$.
Let $H = G + x'y'$, and let $C'$ be the cycle in $H$ containing the edge $x'y'$.
Since it always holds that $|V(G\langle x' \rangle)| \leq |V(G\langle x \rangle)|$, 
we have $|V(Q(C', H))|  = |V(G\langle x' \rangle)| < |V(G)| - |V(G\langle v' \rangle)| = |V(Q(C, H))|$.
This contradicts Condition~(iii) of Lemma~\ref{lem:epot_child}.
Hence, such a vertex pair $\{x', y'\} \not\in E_{\pot}(G)$.
\end{proof}

For a monocyclic chemical graph $G$, a vertex $x$ in $G$, 
and a descendant $y$ of $x$, the length of the cycle in $G + xy$
containing the edge $xy$ is $\depth(y) - \depth(x) + 1$.
From this observation, we have the following lemma.

\begin{lem}~\label{lem:cycle_bound}
Let $G$ be a monocyclic chemical graph, let $x$ be a vertex in $V(G)$, 
and let $y$ be a vertex in $V(G\langle x \rangle)$ such that $xy \not\in E(G)$.
Let $H$ denote the graph $G + xy$, let $C'$ denote the cycle of $H$ 
containing the edge $xy$, and let $C$ denote the other cycle of $H$.
If both (a) and (b) below hold true,  then $\{x, y\} \not\in E_{\pot}(G)$.\\
~~ (a)~$|V(Q(C, H))| = |V(Q(C', H))|$ and \\
~~ (b)~$\depth(y) > |C|  + \depth(x) - 1$.
\end{lem}
\begin{proof}
Assume that both $|V(Q(C, H))| = |V(Q(C', H))|$ and $\depth(y) > |C|  + \depth(x) - 1$ hold.
This contradicts Condition~(iii) of Lemma~\ref{lem:epot_child}.
Hence, the vertex pair $\{x, y\}$ is not in $E_{\pot}(G)$.
\end{proof}

Now, we propose an efficient algorithm to construct the set of potential edges by 
Lemmas~\ref{lem:epot_big_tree}, \ref{lem:skipping_x}, and~\ref{lem:cycle_bound}.
Let $G$ be a given monocyclic chemical graph, and let $C$ denote the cycle in $G$.
By Lemma~\ref{lem:epot_big_tree}, first, we look for the root $v$ of a 
pendent tree that contains at least half of the vertices in $G$ - by Theorem~\ref{thm:unique_subtree}, 
if such a pendent tree of $G$ exists, then it is unique. 
If there exists no such pendent tree, then $E_{\pot}(G) = \emptyset$.
On the other hand, if such a pendent tree exists, then, for each child $v'$ of $v$, 
from Lemma~\ref{lem:skipping_x}, we traverse the subtree induced by $\{v\} \cup V(G\langle v \rangle)$ 
in depth-first order by skipping vertices $x$ such that 
$|V(G\langle x \rangle)| < |V(G)| - V(G\langle v' \rangle)$.
Let $x$ denote the vertex being currently visited in the depth-first traversal.
If $|V(G\langle x \rangle)| < |V(G)| - V(G\langle v' \rangle)$ holds, 
then, for any descendant $y \in V(T^*)$ of $x$ such that $xy \in E(G)$, 
the vertex pair $\{x, y\}$ belongs to $E_{\pot}(G)$.
On the other hand, if $|V(G\langle x \rangle)| = |V(G)| - V(G\langle v' \rangle)$ holds, 
then, from Lemma~\ref{lem:cycle_bound}, we choose $y$ satisfying $\depth(y) > |C|  + \depth(x) - 1$.
We show an efficient algorithm to generate the set of potential edges of a given monocyclic chemical graph $G$
as Procedure~{\procPotEdges} \textsc{GeneratePotentialEdges}.

\bigskip 

{\bf Procedure~{\procPotEdges}}
\textsc{GeneratePotentialEdges}
\begin{algorithmic}[1]
\Require A monocyclic chemical graph $G$ with a cycle~$C$.
\Ensure The set $E_{\pot}(G)$ of potential edges in~$G$.
\State{$S := \emptyset$};
\If{there exists a vertex $v \in V(C)$ such that $|V(G\langle v \rangle)| \geq |V(G)|/2$}
	\For{{\bf each} child $v'$ of $v$}
		\For{{\bf each} $x \in V(G\langle v' \rangle)\cup\{v\}$ such that $|V(G\langle x \rangle)| \geq |V(G)| - |V(G\langle v' \rangle)|$}
			\If{$|V(G\langle x \rangle)| > |V(G)| - |V(G\langle v' \rangle)|$}
				\For{{\bf each} descendant $y$ of $x$}
					\State{$S := S \cup \{x, y\}$}
				\EndFor
			\Else
				\For{{\bf each} descendant $y$ of $x$ such that $\depth(y) > |C| + \depth(x) - 1$}
					\State{$S := S \cup \{x, y\}$}
				\EndFor
			\EndIf
		\EndFor
	\EndFor
\EndIf;
\State{{\bf output} $S$ as $E_{\pot}(G)$}.
\end{algorithmic}

\begin{lem}
\label{lem:proper_set}
Given a monocyclic graph $G$ on $n$ vertices
 with a unique pendent tree $T$
  such that $|V(T)| > n/2$,
 the set $E_{\pot}(G) \cap A(T)$ of pairs of non-adjacent vertices
 is proper.
\end{lem}
\begin{proof}
By Theorem~\ref{thm:unique_subtree}, unless the monocyclic graph
$G$ has a pendent tree with at least $n/2$ vertices, then
it has no children bi-block 2-augmented trees, and such a pendent tree
is unique.
By the definition of the set $A(T)$ and Lemma~\ref{lem:ad-any-root},
for each pair $\{x, y\} \subseteq V(T)$ of non-adjacent
vertices there exists exactly one vertex pair $\{s, t\} \in A(T)$
such that $G + xy \approx G + st$.
On the other hand, the set $E_{\pot}(G)$ contains all vertex pairs $\{x, y\}$
such that $G + xy$ is a child of $G$, from which the claim follows.
\end{proof}

\section{Restricted Monocyclic Chemical Graphs}
\label{sec:generating_restricted_1-augs}

Based on Theorem~\ref{thm:unique_subtree}, 
only monocyclic chemical graphs that have a pendent tree with at least half the number
of vertices in the graph have children bi-block 2-augmented trees.
We examine how to limit an existing process for generating monocyclic chemical graphs (see~\cite{Suzuki14})
to such trees.
Let $T$ be a rooted multi-tree with $n$ vertices.
For two non-adjacent vertices $x, y \in V(T)$, 
the graph $G = T + xy$ is a monocyclic chemical graph whose cycle comprises
the $x, y$-path in $T$ and the edge $xy$.
We have the following lemma for the number of vertices in the pendent trees of $G$.

\begin{lem}\label{lem:restricted:1aug}
Let $T$ be a rooted multi-tree with $n$ vertices. For non-adjacent vertices $x$ and $y$ in $T$, 
let $P_{xy}$ denote the rooted subtree of $T$ induced by the $x, y$-path in 
$T$ rooted at the lowest common ancestor of $x$ and $y$, 
and let $G$ denote the monocyclic chemical graph $T + xy$.
Then, for any vertex $v \in V(P_{xy})$ the following claims hold true. \\
{\rm (i)} If $v$ is a leaf in $P_{xy}$, then $|V(G\langle v \rangle)| = |V(T_{v})|$; \\
{\rm (ii)} Any $v \in P_{xy} \setminus \{x, y, \lca(x, y)\}$ has a 
    unique child $c \in P_{xy}$, and it holds that 
    $|V(G\langle v \rangle)| = |V(T_{v})| -  |V(T(c))|$; and \\
{\rm (iii)} For the root $r = \lca(x, y)$ of $P_{xy}$
	and the set $\ch(r)$  of children of $r$ in $P_{xy}$, 
	it holds that $|V(G \langle r \rangle)| = n - \sum_{c \in \ch(r)} |V(T(c))|$. 
\end{lem}
\begin{proof}
The cycle $C$ in $G$ comprises the $x, y$-path in $T$ and the edge $xy$.
For a vertex $v$ in $C$, the number of vertices in the pendent tree 
rooted at $v$ is the number of vertices reachable from $v$ in $T - E(P_{xy})$.\\
(i) Assume that $v$ is a leaf of $P_{xy}$.
All vertices in $V(T_{v})$ are reachable from $v$ since $E(T_{v}) \cap E(P_{xy}) = \emptyset$.
However, any other vertex in $T$ is not reachable from $v$ in $G - E(C)$ 
since $E(P_{xy}) \subseteq E(C)$ and for the parent $\parent(v)$ of $v$ it holds that  $\parent(v) \in V(P_{xy})$.
Hence we have $|V(G\langle v \rangle)| = |V(T_{v})|$.\\
(ii) 
Since $P_{xy}$ is induced by the $x, y$-path in $T$ and $v$ is not the root of $P_{xy}$, 
$v$ must have a unique child $c$ in $P_{xy}$, 
and it holds that $\{vc, v\parent(v)\} \subseteq E(P_{xy})$.
Hence, the pendent tree $G\langle v \rangle$ consists of all 
descendants of $v$ not in $V(T(c))$, and we have $|V(G\langle v \rangle)| = |V(T_{v})| -  |V(T(c))|$.\\
(iii) 
In this case, the number of children of $r$ is 1 or 2. 
For a child $c$ of $r$, 
none of the vertices in $V(T(c))$ are reachable from $v$ in $T - E(P_{xy})$ since $vc \in E(P_{xy})$. 
However, any other vertices are reachable.
Hence, we have $|V(G\langle r \rangle)| = n - \sum_{c \in \ch(r)} |V(T(c))|$.
\end{proof}

Based on Lemma~\ref{lem:restricted:1aug}, 
a rooted multi-tree $T$ on $n$ vertices
and two non-adjacent vertices $x, y \subseteq V(T)$, 
we devise an algorithm to determine whether the 
monocyclic chemical graph $T + xy$ has a pendent tree with at least $n/2$ vertices or not.
Assume that, for each vertex $v$ in $T$, 
the number of vertices in the subtree $T_{v}$ is already calculated.
Based on Lemma~\ref{lem:restricted:1aug}, we calculate the number 
of vertices of each pendent tree in $T + xy$.
A description of the algorithm is shown 
as Procedure~{\procRestrictedRoot}.

\bigskip

{\bf Procedure~{\procRestrictedRoot}}
\textsc{GenerateRestrictedRootedMonocyclicGraphs}
\begin{algorithmic}[1]
\Require A rooted multi-tree $T$ on $n$ vertices, 
	    and two non-adjacent vertices $\{x, y\} \subseteq T$.
\Ensure If the monocyclic chemical graph $T + xy$ has a pendent tree 
	with at least $n/2$ vertices, then a message ``Yes, '' otherwise, a message ``No.''
\State{$\textrm{Answer} :=$ ``No''};
\State{$v := \lca(x, y)$};
\State{$m := |V(T)|$};
\If{$x \neq v$}
	\Statex{/* $x$ is a leaf of $P_{xy}$ */}
	\If{$|V(T(x))| \geq n$}
		\State{$\textrm{Answer} :=$ ``Yes''}
	\EndIf;
	\State{$c := x; w := \parent(x)$};
	\While{$w \neq v$}
		\If{$|V(T(w))| - |V(T(c))| \geq n$}
			\State{$\textrm{Answer} :=$ ``Yes''}
		\EndIf;
		\State{$c := w; w := \parent(w)$}
	\EndWhile;
	\State{$m := m - |V(T(c))|$}
\EndIf;
\If{$y \neq v$}
	\Statex{/* $y$ is a leaf of $P_{xy}$ */}
	\If{$|V(T(y))| \geq n$}
		\State{$\textrm{Answer} :=$ ``Yes''}
	\EndIf;
	\State{$c := y; w := \parent(y)$};
	\While{$w \neq v$}
		\If{$|V(T(w))| - |V(T(c))| \geq n$}
			\State{$\textrm{Answer} :=$ ``Yes''}
		\EndIf;
		\State{$c := w; w := \parent(w)$}
	\EndWhile;
	\State{$m := m - |V(T(c))|$}
\EndIf;
\LineComment{Here it holds that $m = n - \sum_{c \in \ch(r)} |V(T(c))|$}
\If{$m \geq n/2$}
	\State{$\textrm{Answer} :=$ ``Yes''}
\Else
	\State{$\textrm{Answer} :=$ ``No''}
\EndIf;
\State {\bf output} $\textrm{Answer}$.
\end{algorithmic}

\subsection{The Case of Centroid-Rooted Multi-Trees}
\label{sec:centroid-rooted}

Assume that monocyclic chemical graphs are obtained by adding 
multi-edges to an $n$-vertex multi-tree rooted at its centroid.
In this case, we propose an efficient algorithm to determine
whether a monocyclic chemical graph obtained in this way 
has a pendent tree with at least $n/2$ vertices or not.

Let $T$ be a multi-tree with $n$ vertices
rooted at its centroid,
let $x$ and $y$ be non-adjacent vertices in $T$, 
and let $G = T + xy$.
Let $P_{xy}$ denote the rooted subtree of $T$ induced by the $x, y$-path in $T$
rooted at the lowest common ancestor of $x$ and~$y$.
Suppose that $T$ has a unicentroid, ~$r$.
In this case, for each vertex $v \in V(T) - \{r\}$, we have $|V(T_{v})| \leq \lfloor (n-1)/2 \rfloor$.
Hence, for each vertex $v$ in $P_{xy} \setminus \{\lca(x, y)\}$, 
we have $|V(G\langle v \rangle)| < n/2$.
This implies that it is sufficient to check Condition~(iii) 
of Lemma~\ref{lem:restricted:1aug} to determine whether $G$ 
has a unique pendent tree with at least $n/2$ vertices or not.
Note that, if the root of $P_{xy}$ is not the unicentroid $r$, 
or the root of $P_{xy}$ is $x$ or $y$, then it immediately follows 
that $T + xy$ has a unique pendent tree with at least $n/2$ vertices.

Suppose that $T$ has a bicentroid, ~$e$.
If $e \notin E(P_{xy})$, 
then by Lemma~\ref{lem:restricted:1aug}-(iii), 
for the root $r = \lca(x, y)$ of $P_{xy}$, 
we have $|V(G\langle r \rangle)| = n - \sum_{c \in \ch(r)} |V(T(c))| \geq n - (n/2 - 1) \geq n/2$.
On the other hand, assume that $e \in E(P_{xy})$.
If $x$ (resp., $y$) is an endpoint of the bicentroid, then we have $|V(G\langle x \rangle)| = n/2$ (resp., $|V(G\langle y \rangle)| = n/2$).
Otherwise, for each vertex $v$ in $V(P_{xy})$, it holds that $|V(G\langle x \rangle)| < n/2$.

From the above observations, 
we propose an algorithm to determine whether a monocyclic chemical graph obtained 
by adding an edge to a multi-tree rooted at its centroid has 
a unique pendent tree with at least $n/2$ vertices or not 
as Procedure~{\procRestrictedCentroid}.

\bigskip 

{\bf Procedure~{\procRestrictedCentroid}}
\textsc{GenerateRestrictedCentroid-RootedMonocyclicGraphs}
\begin{algorithmic}[1]
\Require A multi-tree $T$ with $n$ vertices rooted at its centroid, and two non-adjacent vertices $x$ and $y$ in $T$.
\Ensure If the monocyclic chemical graph $T + xy$ has a pendent tree with at least $n/2$ vertices, 
	      then a message ``Yes, '' otherwise, a message ``No.''
\State{$\textrm{Answer} :=$ ``No''};
\State{$v := \lca(x, y)$};
\If{$T$ has a unicentroid}
	\If{$v$ is not the unicentroid}
		\State{$\textrm{Answer} :=$ ``Yes''}
	\ElsIf{$v$ is $x$ or $y$}
		\State{$\textrm{Answer} :=$ ``Yes''}	
	\EndIf;
	\State{$p := $ the child of $v$ that is $x$ or an ancestor of $x$};
	\State{$q := $ the child of $v$ that is $y$ or an ancestor of $y$};
	\If{$n - |V(T_p)| - |V(T_q)| \geq n/2$}
		\State{$\textrm{Answer} :=$ ``Yes''}	
	\EndIf
\Else ~~/* {$T$ has a bicentroid} */
	\State{$e := $ the bicentroid of $T$};
	\State{$P_{xy} := $ the $x, y$-path in $T$};
	\If{$e \notin E(P_{xy})$}
		\State{$\textrm{Answer} :=$ ``Yes''}
	\ElsIf{either $x$ or $y$ is an endpoint of the bicentroid}
		\State{$\textrm{Answer} :=$ ``Yes''}
	\EndIf
\EndIf;
\State{{\bf output} $\textrm{Answer}$}.	
\end{algorithmic}

 \section{Experimental results}
 \label{sec:experiments}
 
 To test the effectiveness of our algorithm for enumerating bi-block 2-augmented trees, 
 we have implemented it 
 and performed computational comparison with MOLGEN~\cite{MOLGEN5}, 
 a generator for chemical graphs.

 In particular, we did experiments for two different types of instances, 
 named EULF-$L$-A and  EULF-$L$-P, by considering
 a set $\pathset$ of colored sequences with length at most a given integer~$N$, 
 given lower and upper bounds, 
 $\mathbf{g}_{a}: \pathset \to \mathbb{Z}_+$ and
 $\mathbf{g}_{b}: \pathset \to \mathbb{Z}_+$, respectively, 
 on the path frequencies of the paths in $\pathset$, and integers $L$ and~$d$.
 For a given set $\pathset$ of colored sequences and a graph $G$, 
 let $\fvset_{\pathset}(G): \pathset \to \mathbb{Z}_+$ 
 denote the number $\freq(t, G)$ of rooted paths $P \subseteq G$
 such that $\gamma(P) = t \in \pathset$.
 Assuming that 
 $\mathbf{g}_{a} \leq \mathbf{g}_b$, 
 and in particular, that
 $\mathbf{g}_{a}[t] = \mathbf{g}_b[t]$ 
is satisfied for each colored sequence $t \in \pathset \cap \Sigma^{0, d}$, 
each of the instance types EULF-$L$-A and  EULF-$L$-P asks 
to enumerate chemical graphs $G$ such that 
$\mathbf{g}_{a} \leq \freq_{\pathset}(G) \leq \mathbf{g}_{b}$, 
and for any $P \subseteq G$ such that $\gamma(P) \notin \pathset$, 
it holds that $|P| > L$ and $|P| \leq L$, for
 instance types EULF-$L$-A and  EULF-$L$-P, respectively.

We have chosen six compounds 
 from the PubChem database 
 which when represented
 as hydrogen-suppressed chemical graphs have
 bi-block 2-augmented tree structure,  and constructed
 feature vectors based on the path frequencies
 of the paths  in the chemical graphs.
 All compounds have~13 non-hydrogen atoms, 
 maximum path length~11, 
 and maximum bond multiplicity $d \in \{2, 3\}$.
 All compounds include the three chemical
 elements {\tt C} (carbon), {\tt O} (oxygen), and {\tt N} (nitrogen).
 The information on the chosen compounds, 
 identified by their Compound ID (CID) number
 in the PubChem database is given in Table~\ref{table:6mols}.
\begin{table}[!hb]
\centering
 \caption{Information on the six compounds chosen from 
 the PubChem database for our experiments}
 \label{table:6mols}
   \begin{tabular}{@{} l l l @{} }
  \toprule
    \parbox{20mm}{Molecular formula }
	    & $d$ & CID \vspace{2mm} \\
   \toprule
   \multirow{2}{20mm}{\tt C$_9$N$_1$O$_3$} 
						    & 2 &  130964701 	 	 \\ 
						    & 3 &  131152558 		 \\  \midrule
  \multirow{2}{20mm} {\tt C$_9$N$_2$O$_2$} 
 						     & 2 &  742733 		 \\
						    & 3 &  23461643 	\\ \midrule
  \multirow{2}{20mm} {\tt C$_9$N$_3$O$_1$} 
  						    & 2 &  10307896 		 \\
						    & 3 &  10307899  		\\  \bottomrule
  \end{tabular}
\end{table}

We construct 
instances of types EULF-$L$-A and EULF-$L$-P 
for different values of parameter $L$ in the following way.
We take a set $\Sigma$ of colors to be $\Sigma = \{ \mathtt{C}, \mathtt{O}, \mathtt{N} \}$, 
such that $\val(\mathtt{C}) = 4$, $\val(\mathtt{O}) = 2$, and $\val(\mathtt{N}) = 3$.
For each hydrogen suppressed chemical graph $G$
that corresponds to a chemical compound in Table~\ref{table:6mols}, 
we take $d \in \{2, 3\}$ to be the maximum bond multiplicity in the chemical graph, 
and for some choice of values for $N \geq 0$ we construct
a set of colored sequences $\pathset \subseteq \Sigma^{\leq N, d}$
that consists of all colored sequences $t$ with length $|t| \leq N$
such that $G$ contains a rooted path $P$ with $\gamma(P) = t$. 
Finally, for an integer $s \in [0, 2]$ we set lower and upper bounds, 
${\bf g}_a$ and ${\bf g}_b$
on feature vectors
as follows: 
for $t \in \pathset$, if $|t| \in \Sigma^{0, d}$ then ${\bf g}_a[t] = {\bf g}_b[t] = \freq(t, G)$, 
otherwise ${\bf g}_{a}[t] = \max \{0, \freq(t, G) - s \}$
and  ${\bf g}_{b}[t] = \freq(t, G) + s$.
The parameter $s$ effectively serves to ``relax''
the path frequency specification.

On the other hand, we used MOLGEN~\cite{MOLGEN5}
without aromaticity detection
by specifying the hydrogen suppressed formula, 
the number of cycles to be two in enumerated structures
- thereby enumerating chemical graphs with 2-augmented tree structures
with a maximum allowed bond multiplicity.
Note that there is no option in MOLGEN to specify
whether the enumerated structures have a bi-block 
structure or not.

 We implemented our algorithm in the 
 C++ programming language, and compiled and executed on
the Linux~14.04.6 operating system by the gcc compiler  version~4.8.4 and optimization level~O3.
All experiments were done on a PC with Intel Xeon CPU E5-1660 v3 
running at 3.00~GHz, with 32~GB memory.

 \subsection{Experimental Results for EULF-$L$-A}
 \label{sec:experiments_EULF1}

To test the behavior of our algorithm for 
instance types EULF-$L$-A, especially the effect
the choice of problem parameters have on the running time
and the number of enumerated chemical graphs, 
we choose values for parameter $N \in [2, 6]$, 
and 
we took values for the parameter $L \in \{2, \lceil {N/2} \rceil, N\}$.
 
The results from our experiments 
for  EULF-$L$-A are summarized in 
Figs.~\ref{fig:result_graphs_1} to~\ref{fig:result_graphs_6}.
We observe that our algorithm has a clear advantage
when we are given a path frequency specification
for instances of type EULF-$L$-A
over using MOLGEN to generate molecules with a specified formula.
 We also observe some trends over the values of 
 the parameters $N$, $L$, and~$s$.
 Namely, the number of generated molecules, as well as the time it takes our algorithm, 
 reduces as the length $N$ of the longest path given in the 
 set of paths, as well as the parameter~$L$ increases, 
 but grows with an increasing value~$s$
 that we choose to relax the path frequency specification.

 \begin{figure}[!ht]
  \begin{minipage}{0.45\textwidth}
   \centering
   \includegraphics[width=1.1\textwidth]{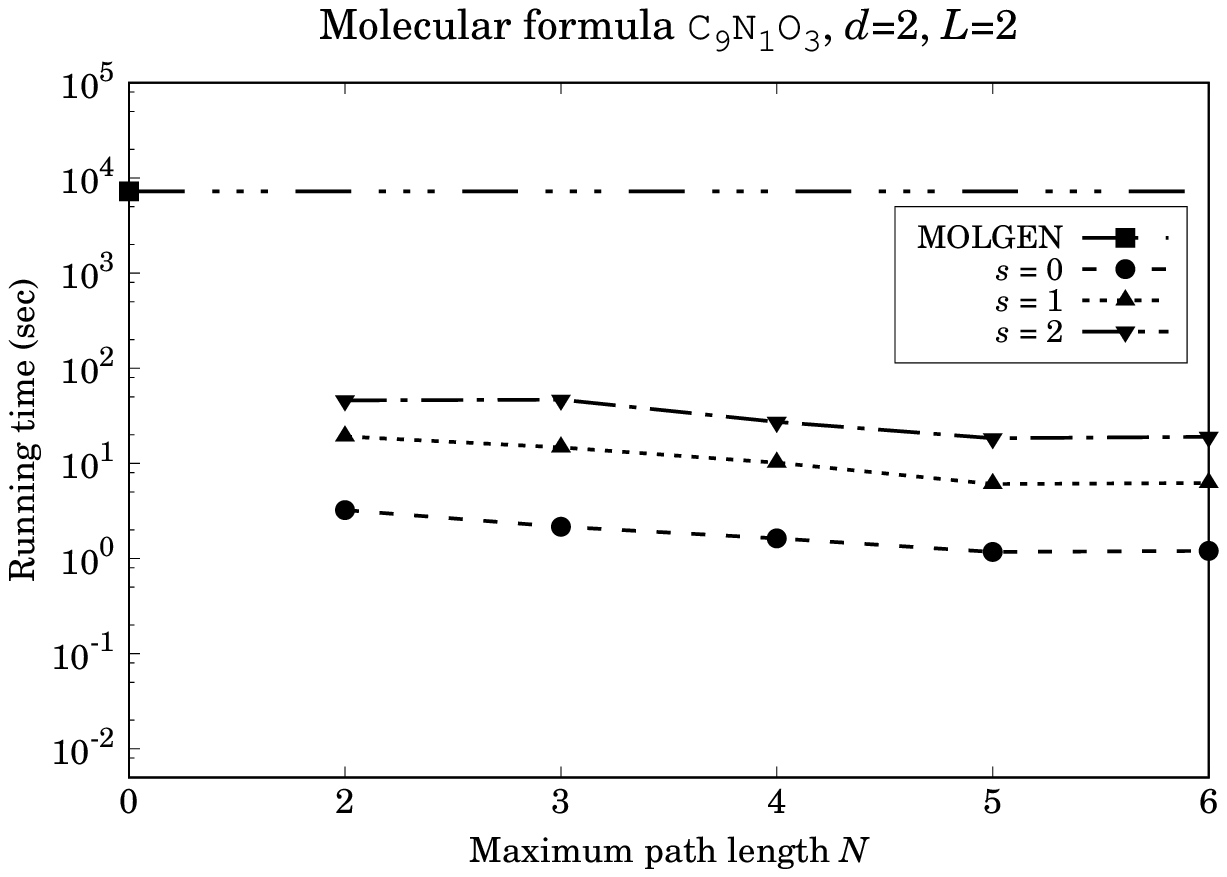}\\
   {\footnotesize (a)}\\
  \end{minipage}
\hfill
  \begin{minipage}{0.45\textwidth}
   \centering
   \includegraphics[width=1.1\textwidth]{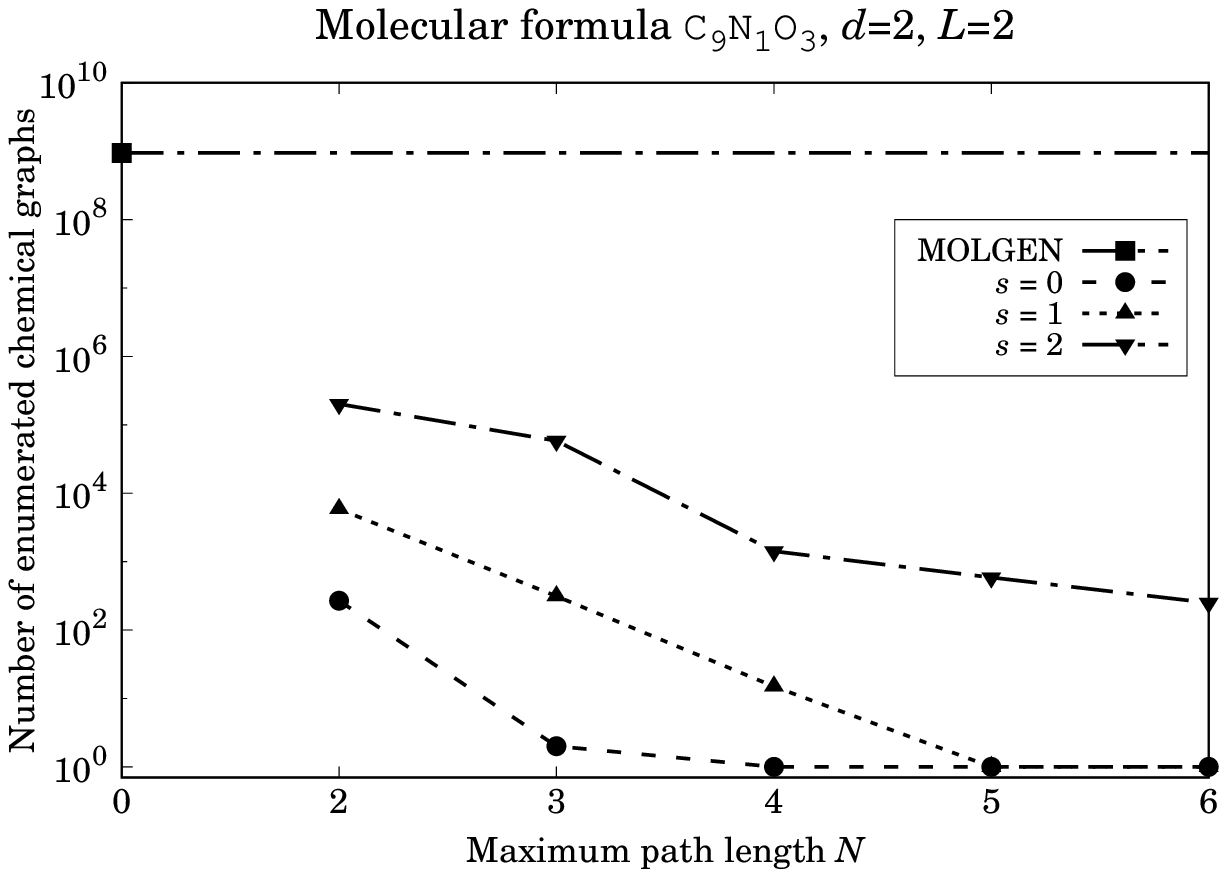}\\
   {\footnotesize (d)}\\
  \end{minipage} 
  \medskip

  \begin{minipage}{0.45\textwidth}
   \centering
      \includegraphics[width=1.1\textwidth]{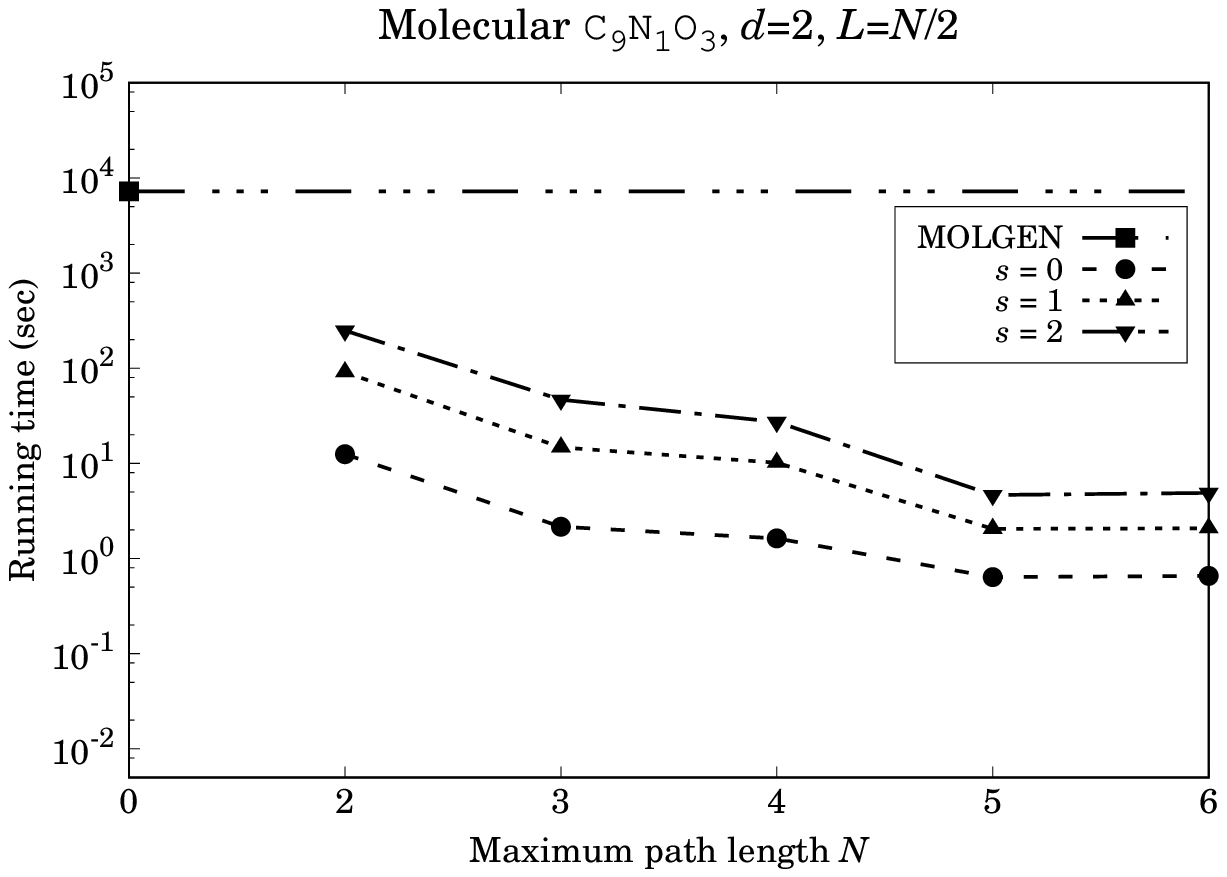}\\
      {\footnotesize (b)}\\
  \end{minipage} 
\hfill
  \begin{minipage}{0.45\textwidth}
   \centering
    \includegraphics[width=1.1\textwidth]{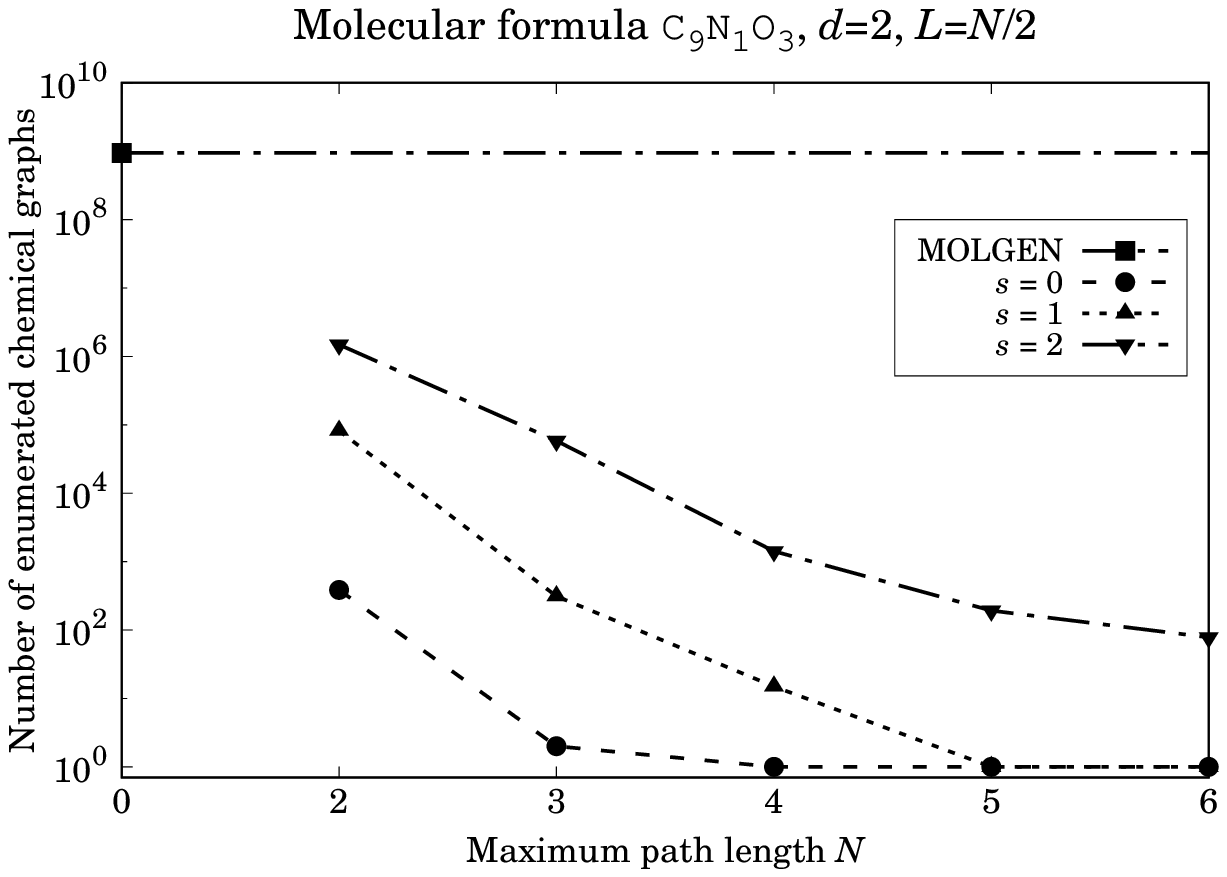}\\
    {\footnotesize (e)}\\
  \end{minipage} 
  \medskip

  \begin{minipage}{0.45\textwidth}
   \centering
      \includegraphics[width=1.1\textwidth]{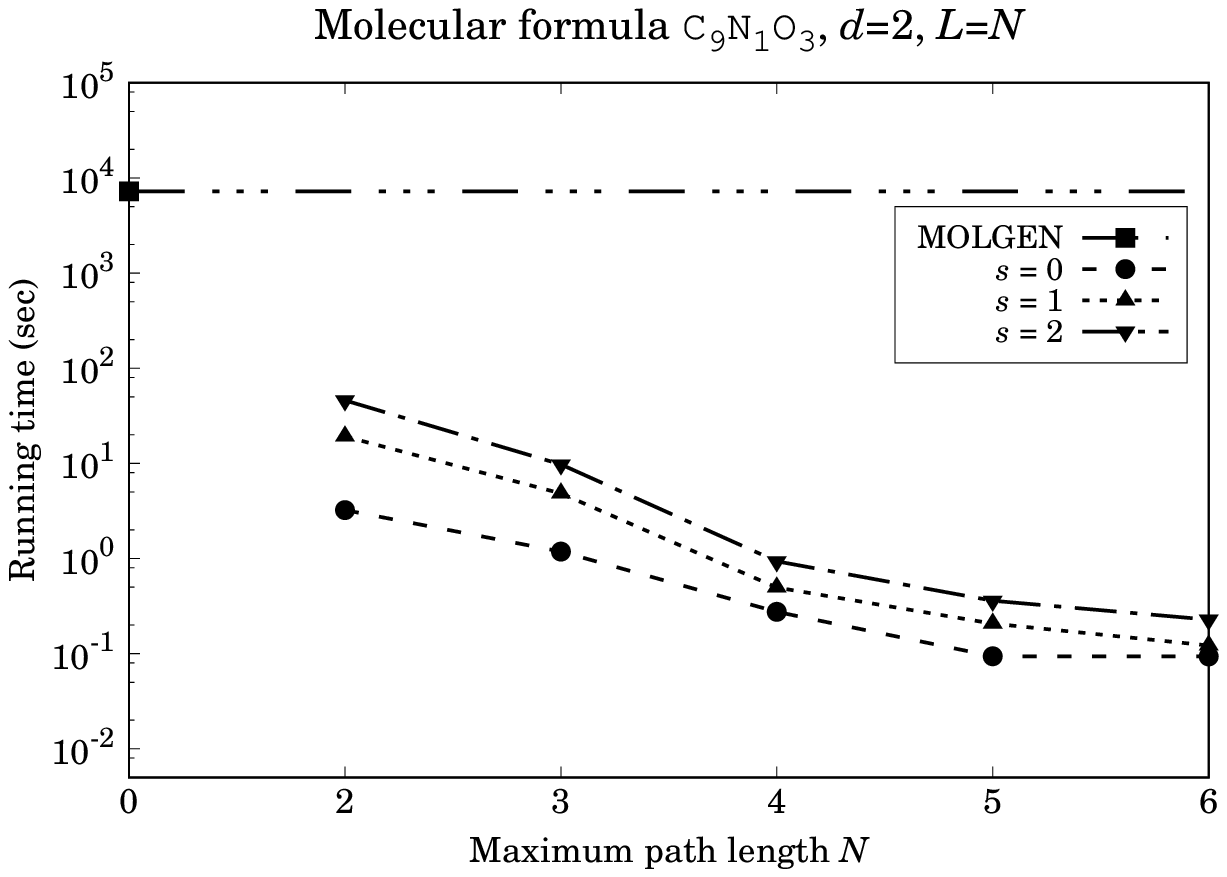}\\
      {\footnotesize (c)}\\
  \end{minipage} 
\hfill
  \begin{minipage}{0.45\textwidth}
   \centering
    \includegraphics[width=1.1\textwidth]{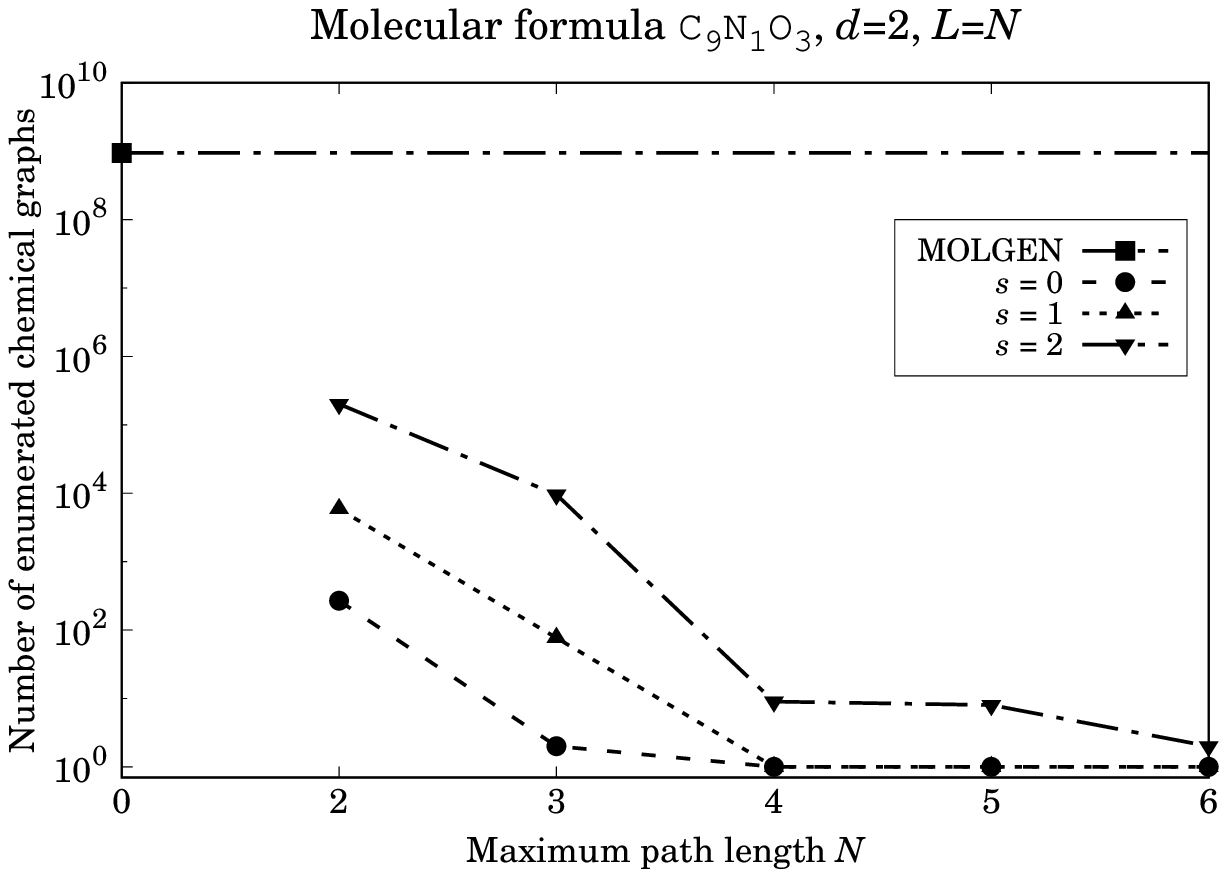}\\
    {\footnotesize (f)}\\
  \end{minipage} 
  \vspace{1cm}
  
  \caption{
    Plots showing the computation time 
    and number of chemical graphs enumerated by our algorithm
    for instance type EULF-$L$-A, as compared to MOLGEN.
    The sample structure from PubChem is with CID 130964701, 
    molecular formula {\tt C$_9$N$_1$O$_3$}, 
    and maximum bond multiplicity~$d=2$.
    (a)-(c)~Running time;
    (d)-(f)~Number of enumerated chemical graphs.
  }
 \label{fig:result_graphs_1}
 \end{figure}

  \begin{figure}[!ht]
  \begin{minipage}{0.45\textwidth}
   \centering
   \includegraphics[width=1.1\textwidth]{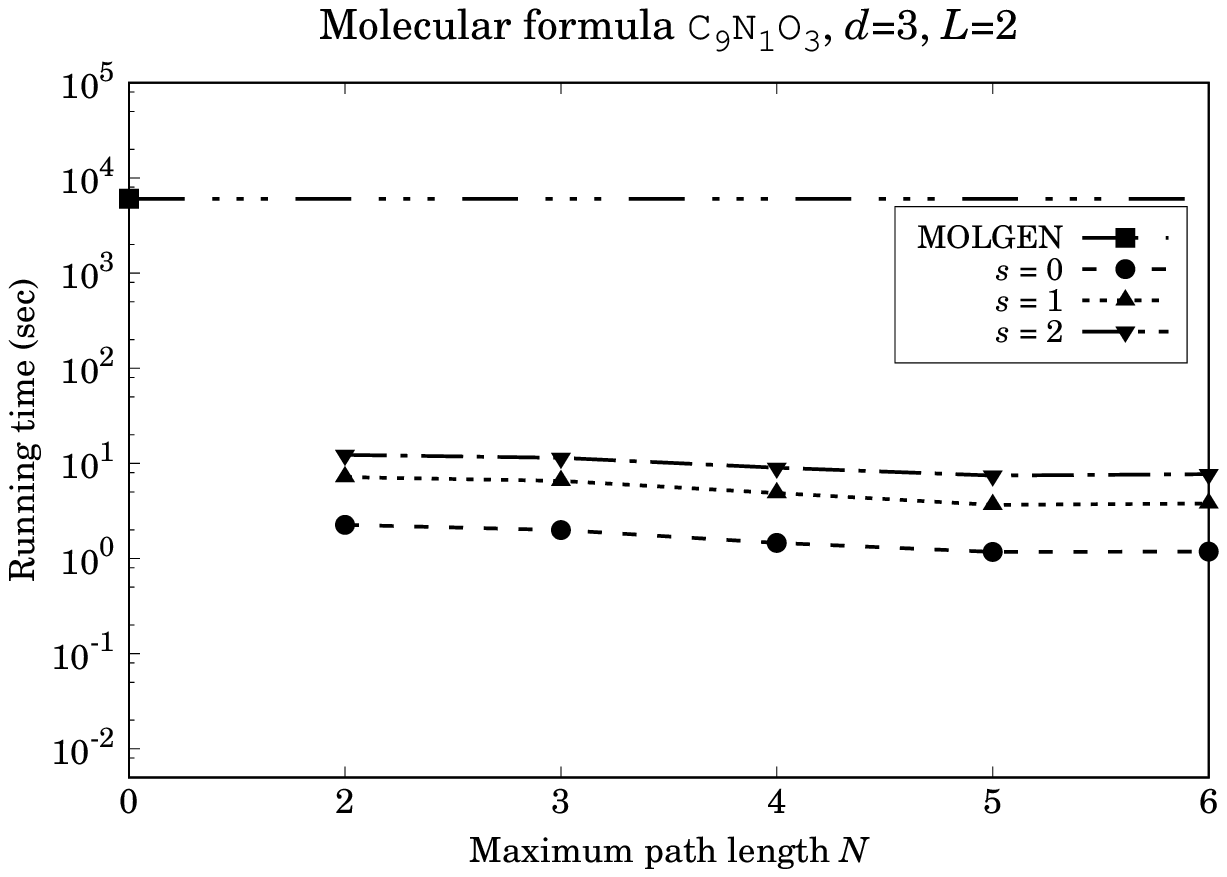}\\
   {\footnotesize (a)}\\
  \end{minipage}
\hfill
  \begin{minipage}{0.45\textwidth}
   \centering
   \includegraphics[width=1.1\textwidth]{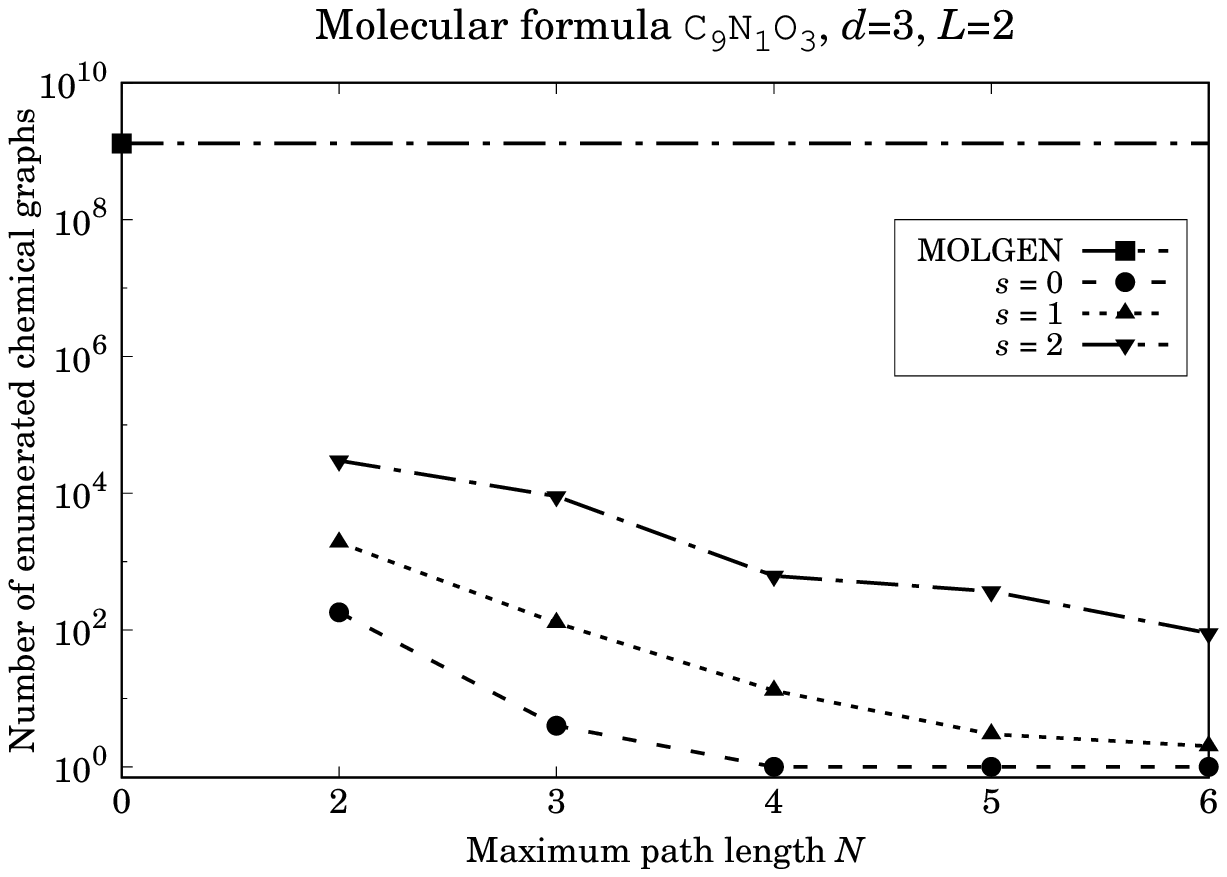}\\
   {\footnotesize (d)}\\
  \end{minipage} 
  \medskip

  \begin{minipage}{0.45\textwidth}
   \centering
      \includegraphics[width=1.1\textwidth]{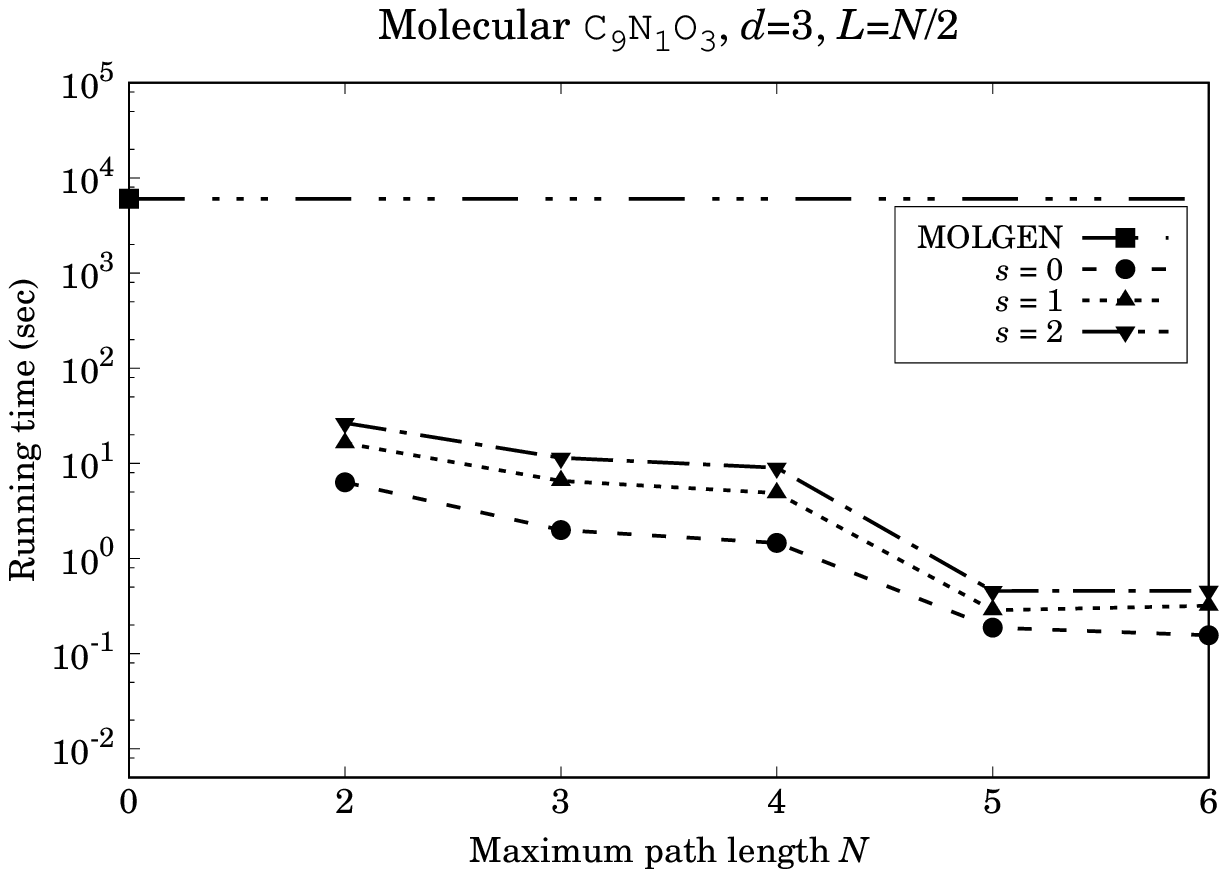}\\
      {\footnotesize (b)}\\
  \end{minipage} 
\hfill
  \begin{minipage}{0.45\textwidth}
   \centering
    \includegraphics[width=1.1\textwidth]{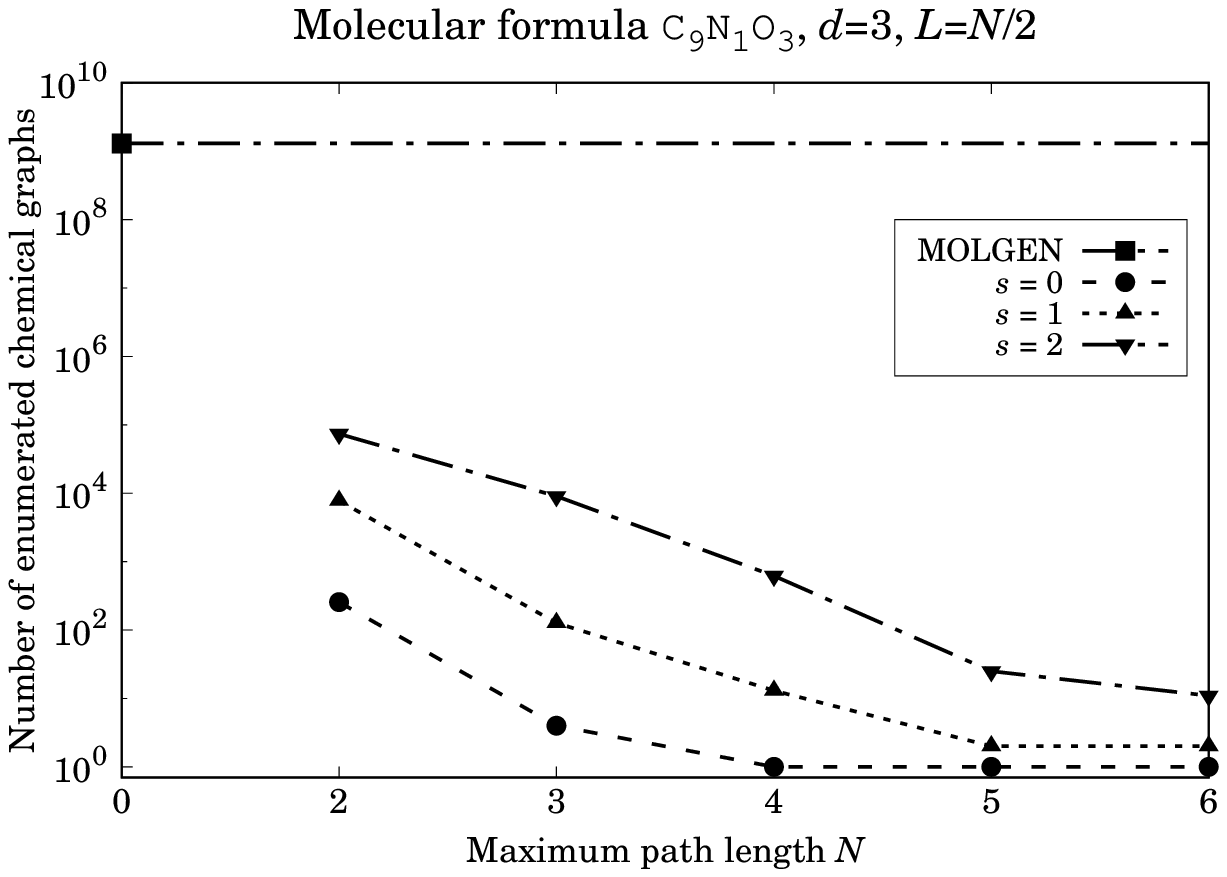}\\
    {\footnotesize (e)}\\
  \end{minipage} 
  \medskip

  \begin{minipage}{0.45\textwidth}
   \centering
      \includegraphics[width=1.1\textwidth]{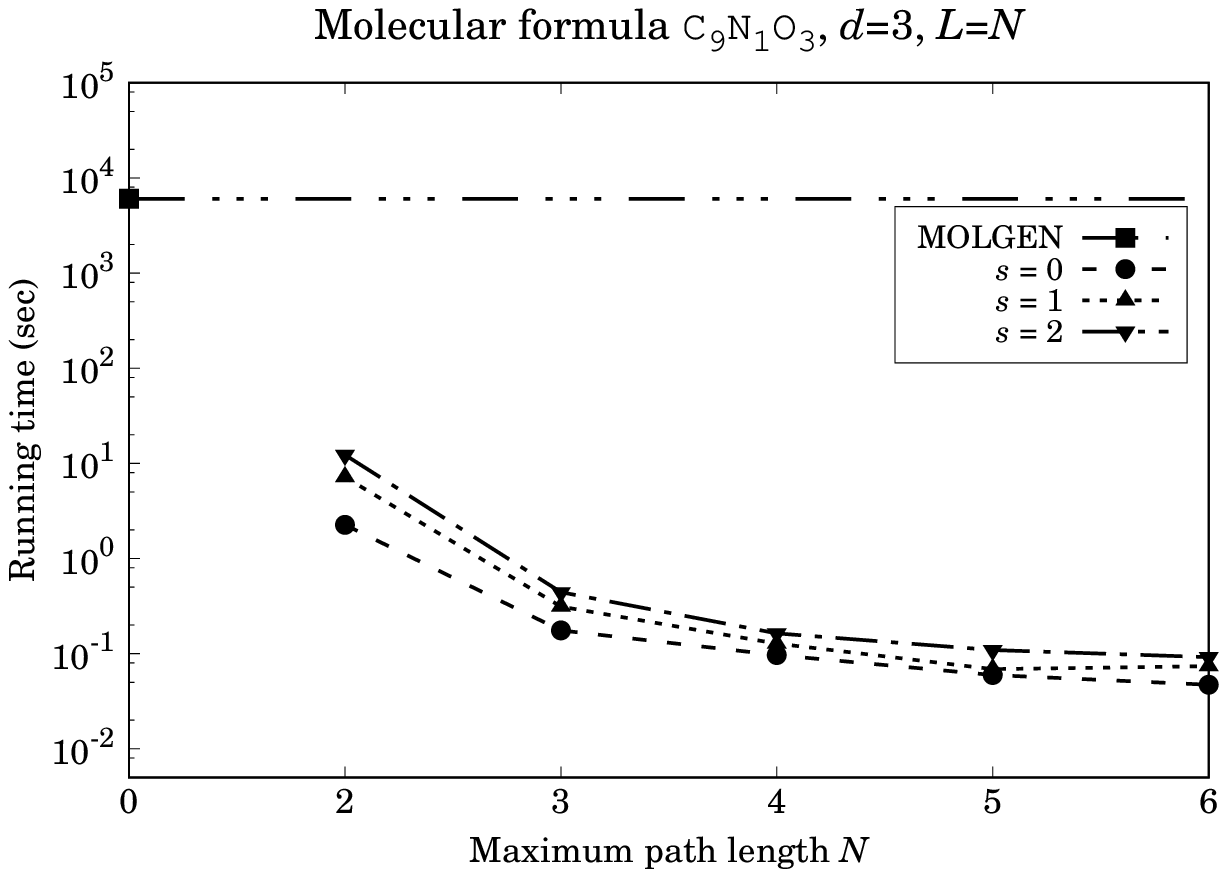}\\
      {\footnotesize (c)}\\
  \end{minipage} 
\hfill
  \begin{minipage}{0.45\textwidth}
   \centering
    \includegraphics[width=1.1\textwidth]{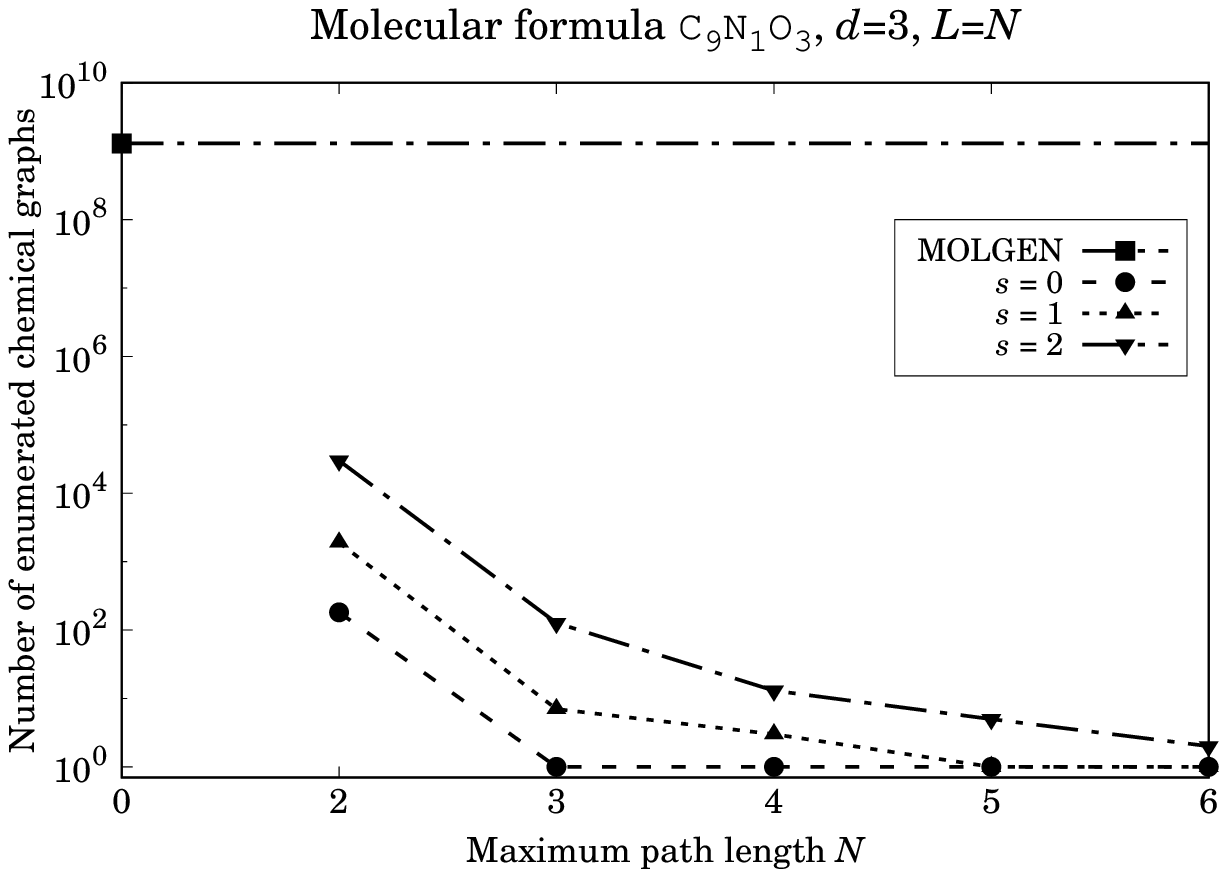}\\
    {\footnotesize (f)}\\
  \end{minipage} 
  \vspace{1cm}
  
  \caption{
    Plots showing the computation time 
    and number of chemical graphs enumerated by our algorithm
    for instance type EULF-$L$-A, as compared to MOLGEN.
    The sample structure from PubChem is with CID~131152558, 
    molecular formula {\tt C$_9$N$_1$O$_3$}, 
    and maximum bond multiplicity~$d=3$.
    (a)-(c)~Running time;
    (d)-(f)~Number of enumerated chemical graphs.
  }
 \label{fig:result_graphs_2}
 \end{figure}

  \begin{figure}[!ht]
  \begin{minipage}{0.45\textwidth}
   \centering
   \includegraphics[width=1.1\textwidth]{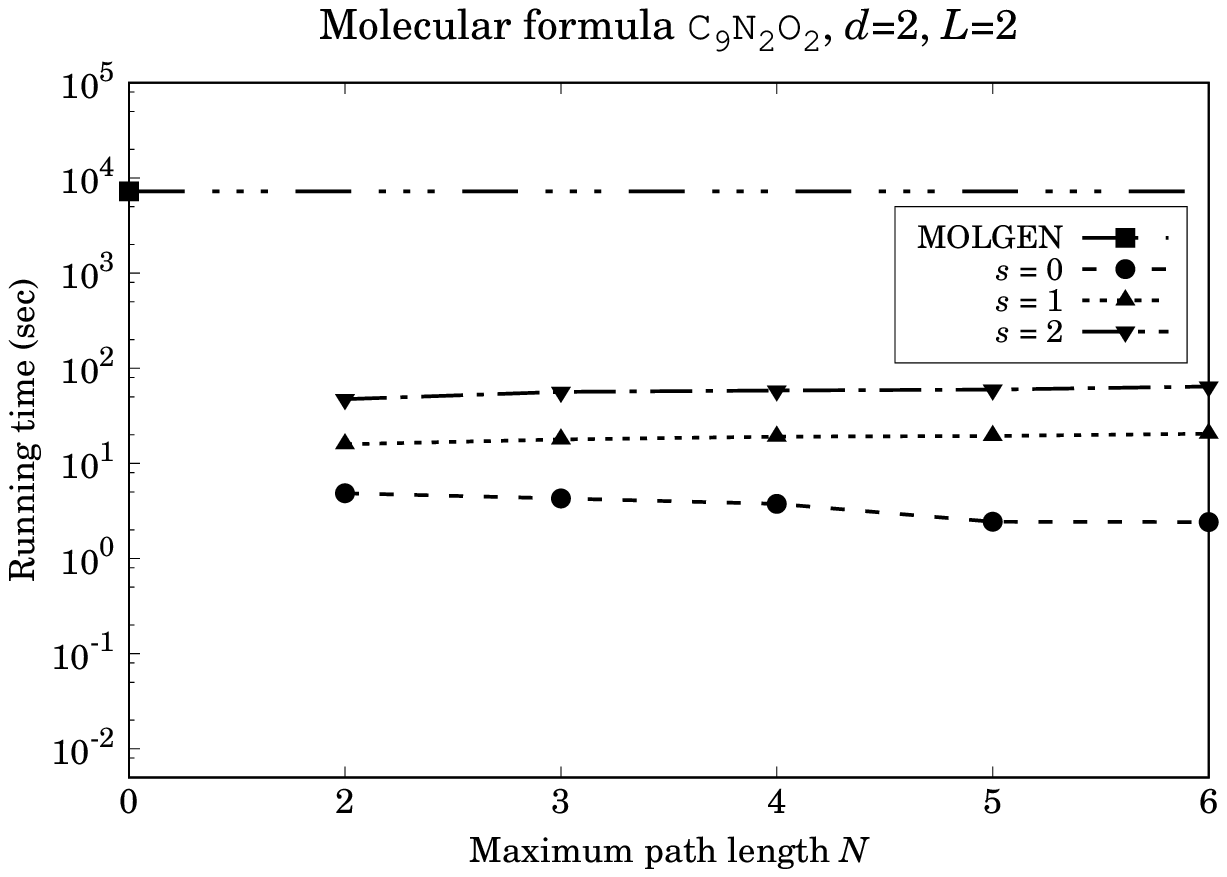}\\
   {\footnotesize (a)}\\
  \end{minipage}
\hfill
  \begin{minipage}{0.45\textwidth}
   \centering
   \includegraphics[width=1.1\textwidth]{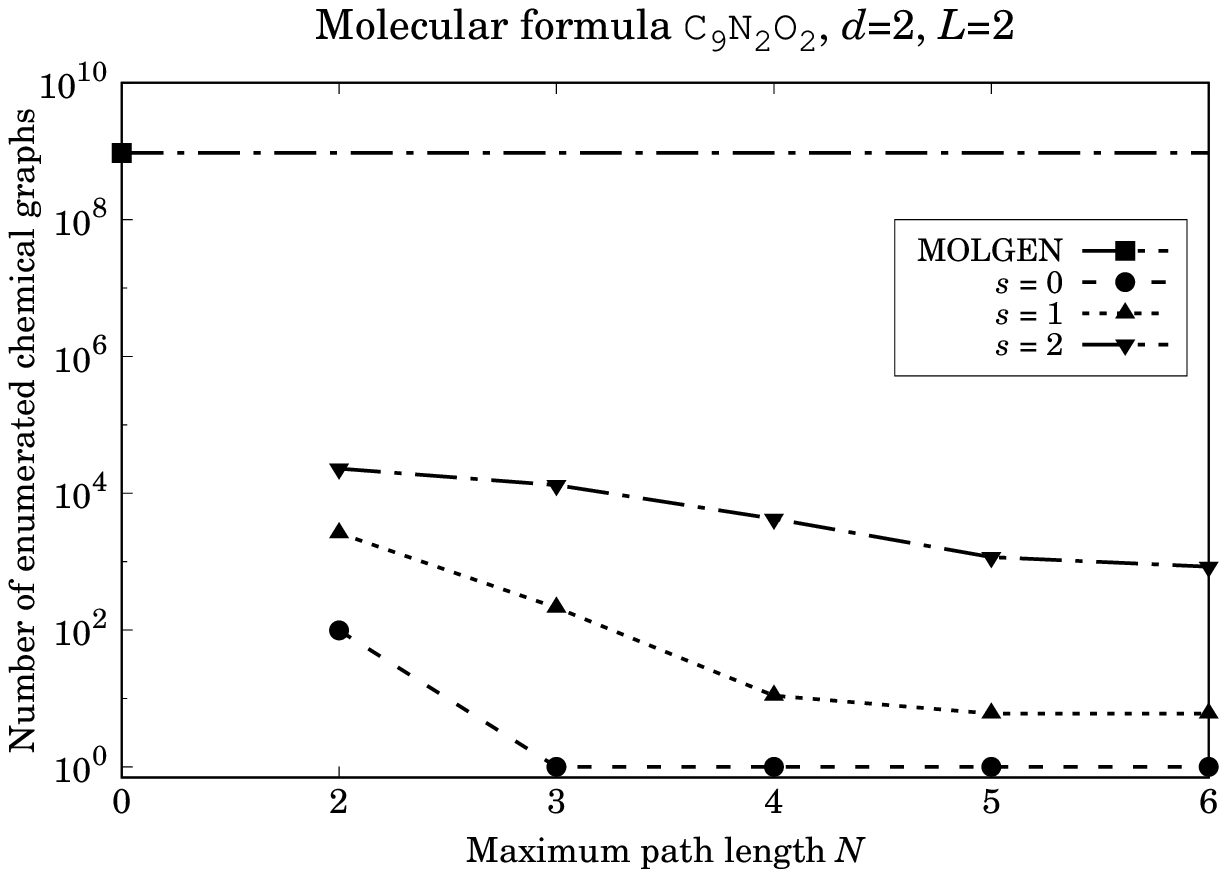}\\
   {\footnotesize (d)}\\
  \end{minipage} 
  \medskip

  \begin{minipage}{0.45\textwidth}
   \centering
      \includegraphics[width=1.1\textwidth]{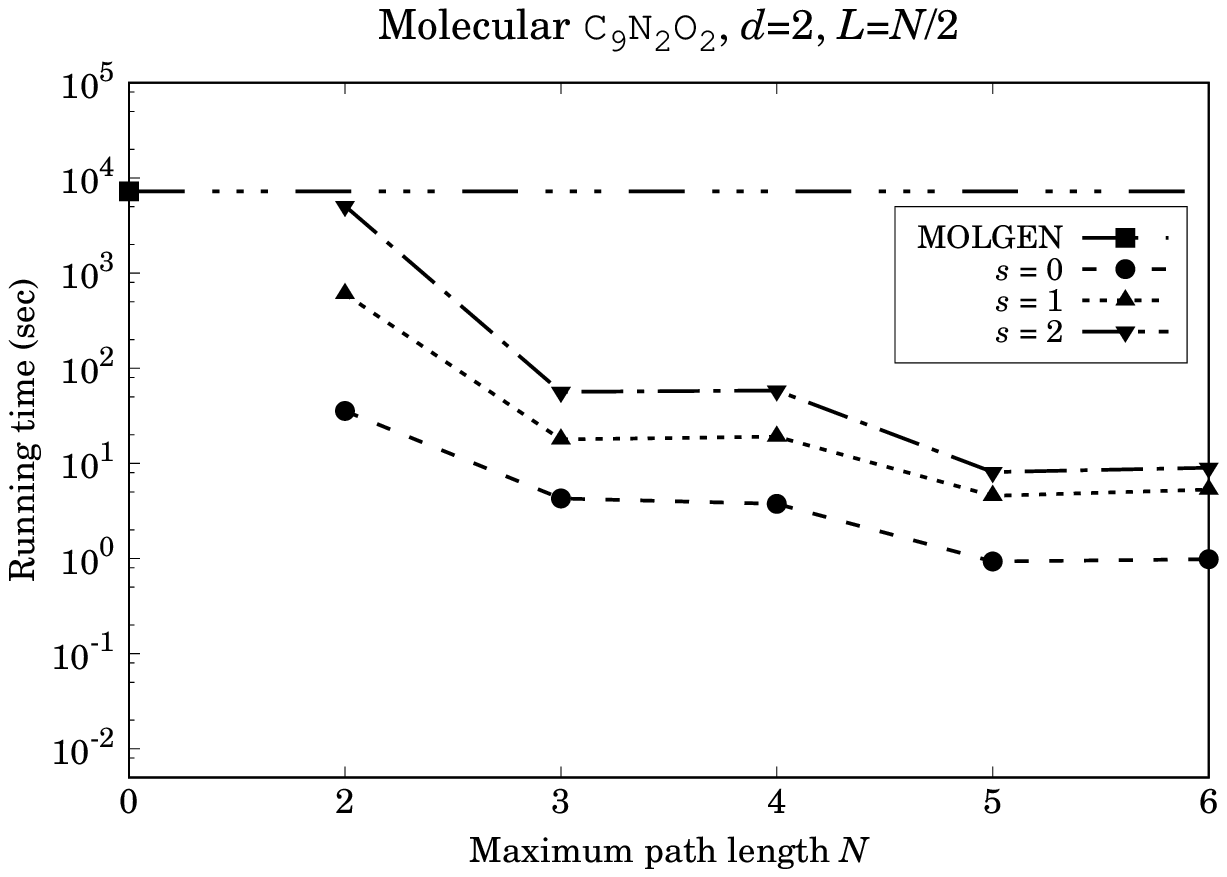}\\
      {\footnotesize (b)}\\
  \end{minipage} 
\hfill
  \begin{minipage}{0.45\textwidth}
   \centering
    \includegraphics[width=1.1\textwidth]{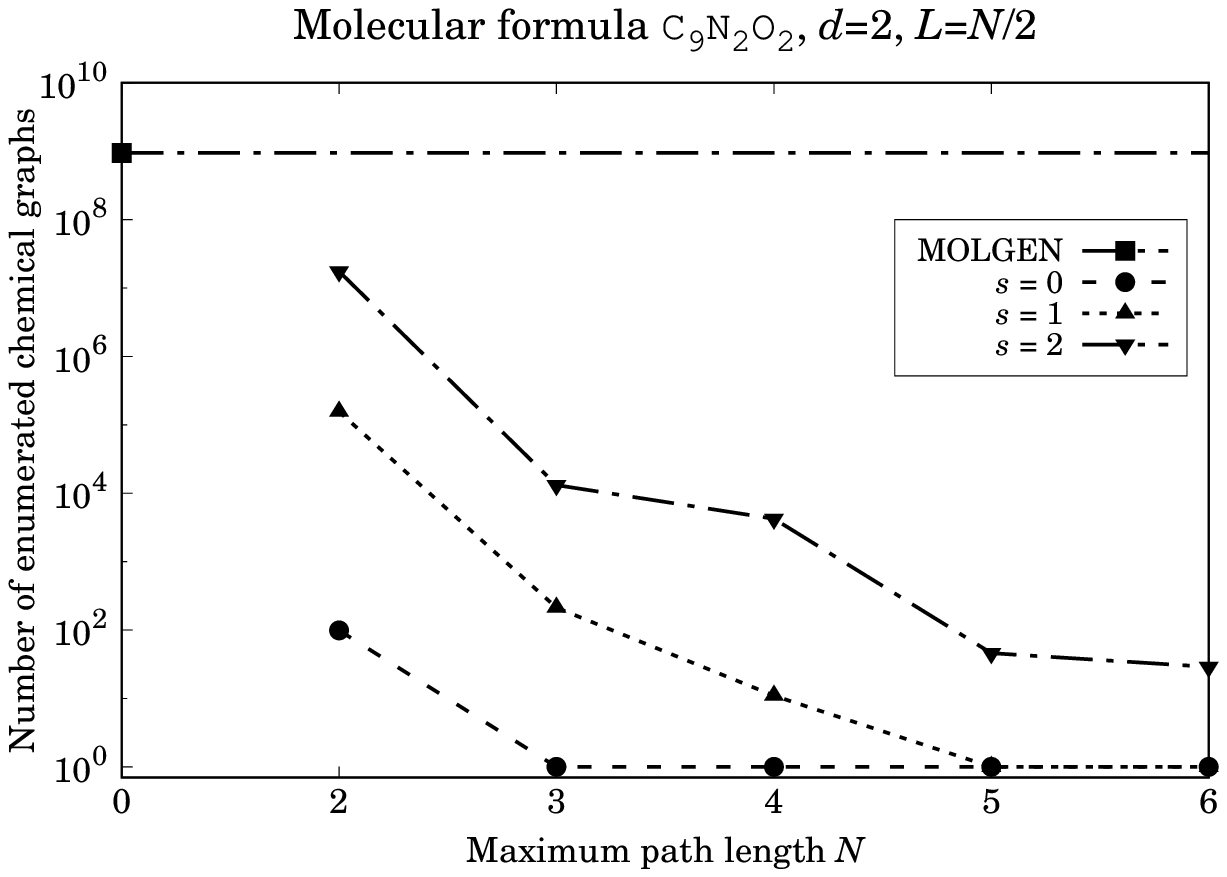}\\
    {\footnotesize (e)}\\
  \end{minipage} 
  \medskip

  \begin{minipage}{0.45\textwidth}
   \centering
      \includegraphics[width=1.1\textwidth]{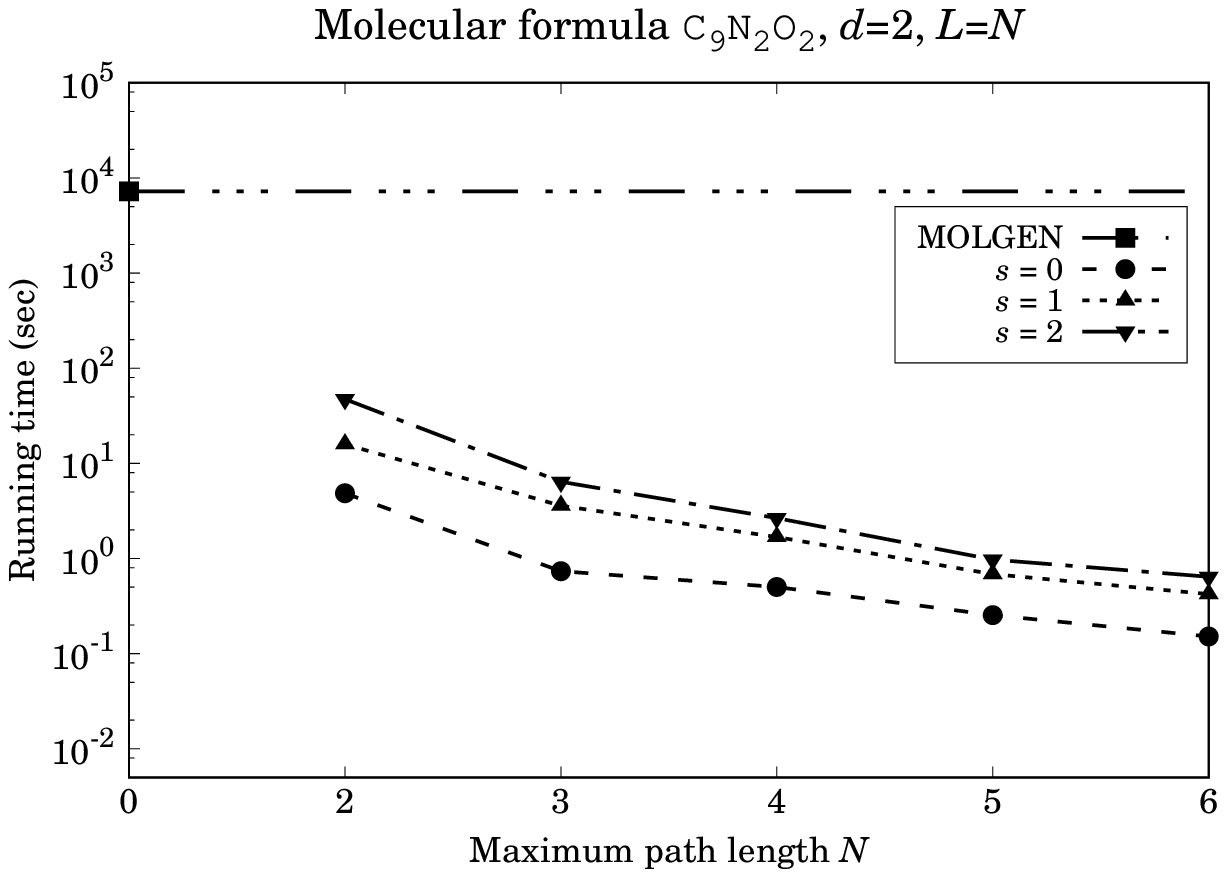}\\
      {\footnotesize (c)}\\
  \end{minipage} 
\hfill
  \begin{minipage}{0.45\textwidth}
   \centering
    \includegraphics[width=1.1\textwidth]{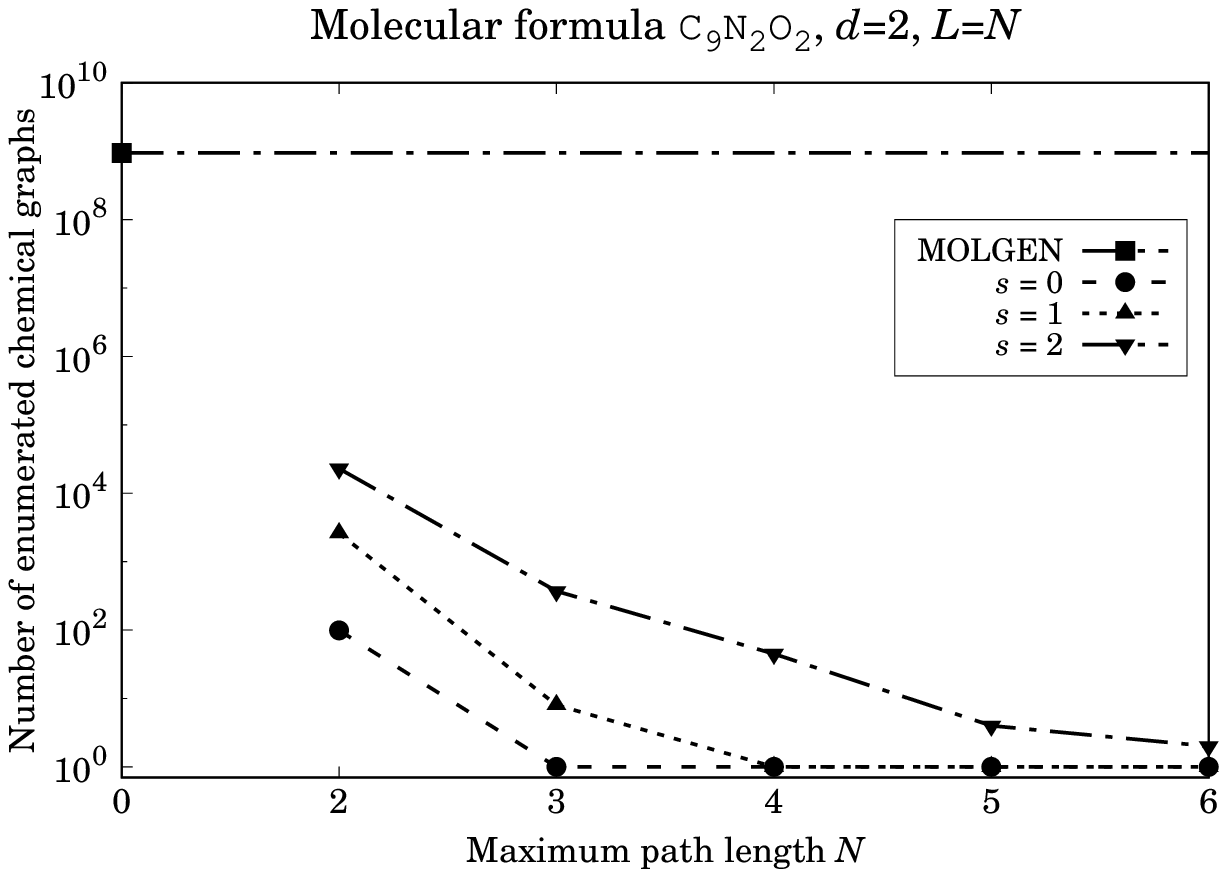}\\
    {\footnotesize (f)}\\
  \end{minipage} 
  \vspace{1cm}
  
  \caption{
    Plots showing the computation time 
    and number of chemical graphs enumerated by our algorithm
    for instance type EULF-$L$-A, as compared to MOLGEN.
    The sample structure from PubChem is with CID~742733, 
    molecular formula {\tt C$_9$N$_2$O$_2$}, 
    and maximum bond multiplicity~$d=2$.
    (a)-(c)~Running time;
    (d)-(f)~Number of enumerated chemical graphs.
  }
 \label{fig:result_graphs_3}
 \end{figure}

 \begin{figure}[!ht]
  \begin{minipage}{0.45\textwidth}
   \centering
   \includegraphics[width=1.1\textwidth]{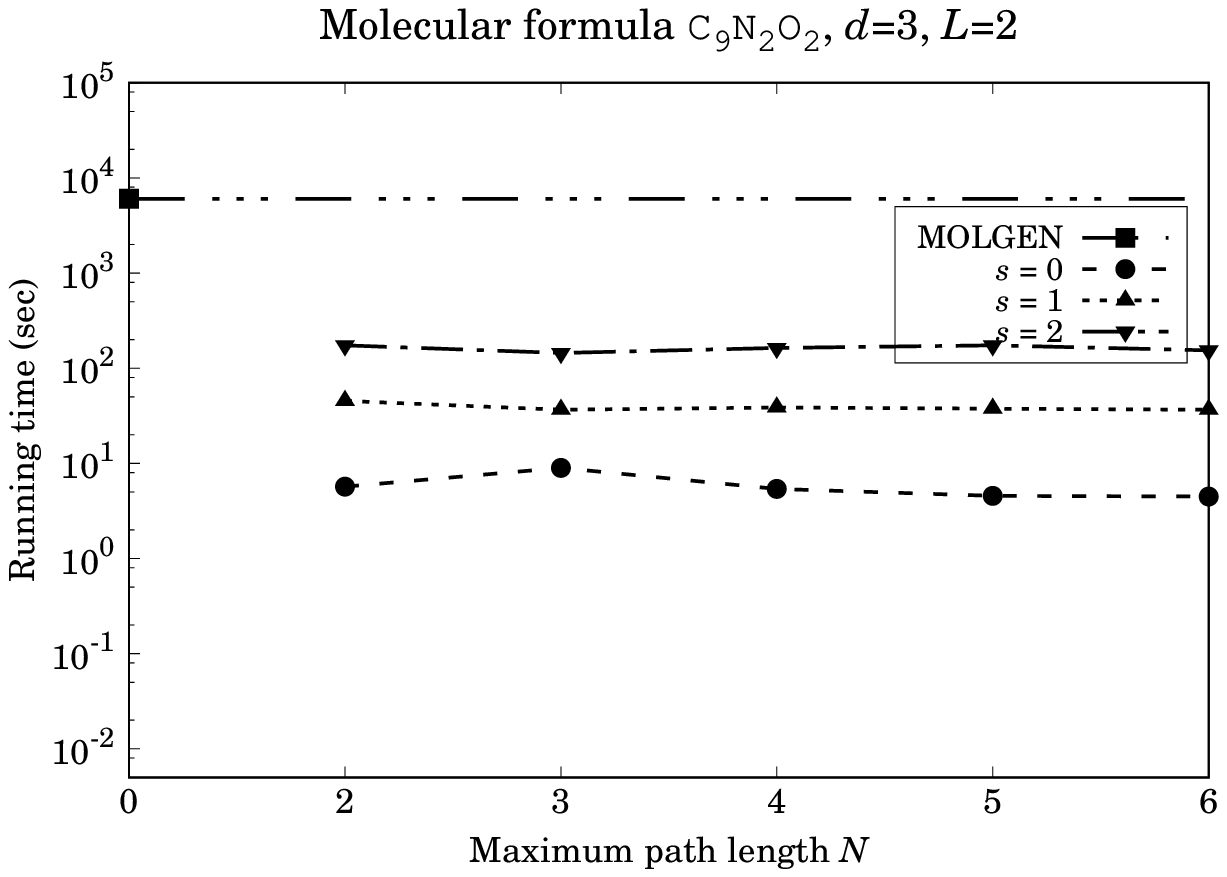}\\
   {\footnotesize (a)}\\
  \end{minipage}
\hfill
  \begin{minipage}{0.45\textwidth}
   \centering
   \includegraphics[width=1.1\textwidth]{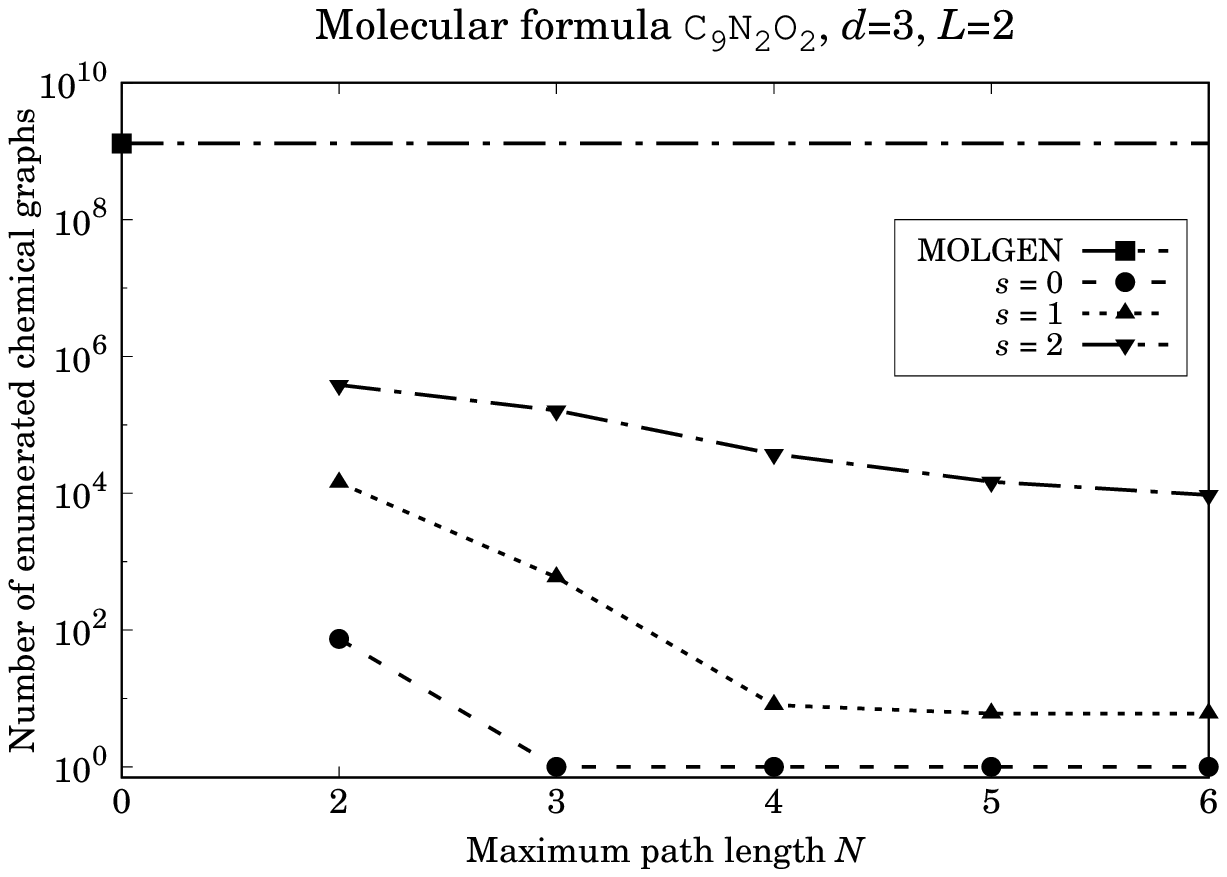}\\
   {\footnotesize (d)}\\
  \end{minipage} 
  \medskip

  \begin{minipage}{0.45\textwidth}
   \centering
      \includegraphics[width=1.1\textwidth]{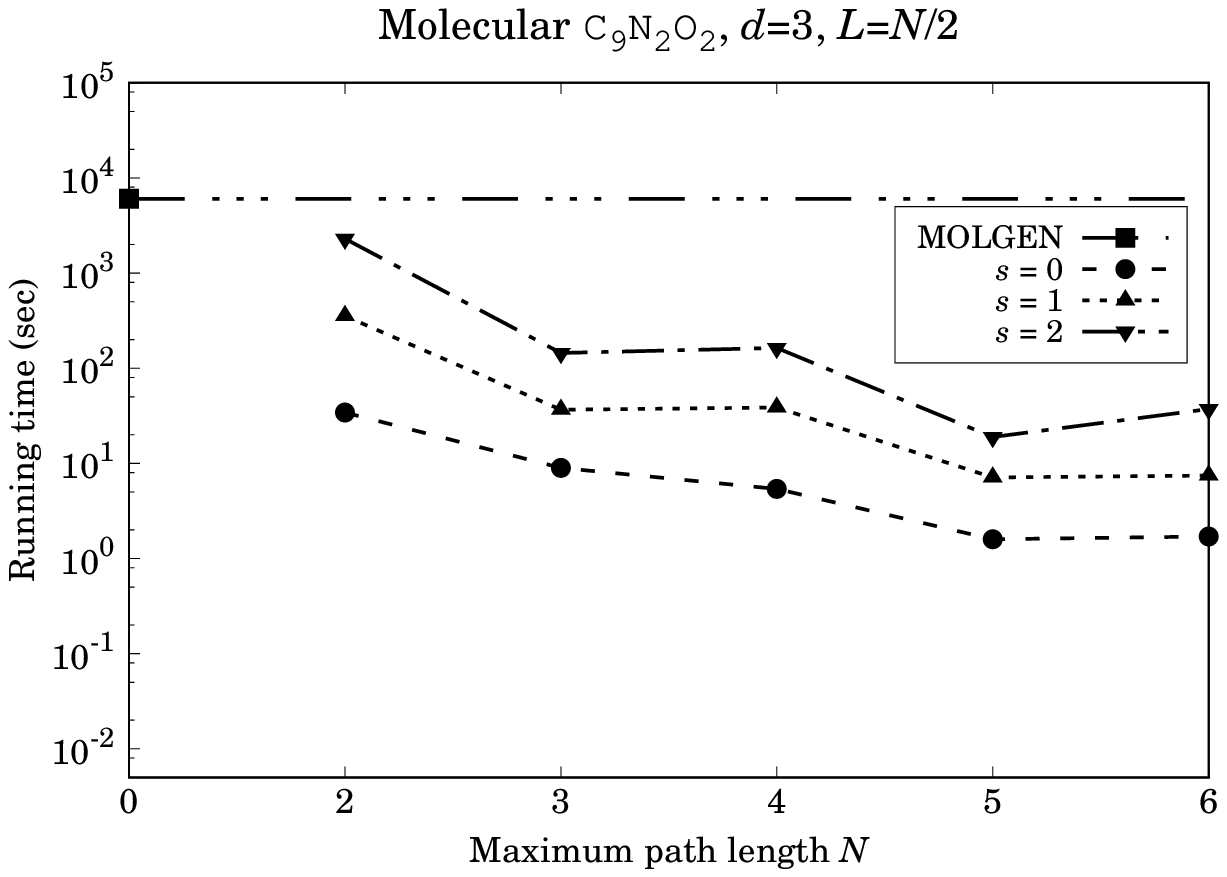}\\
      {\footnotesize (b)}\\
  \end{minipage} 
\hfill
  \begin{minipage}{0.45\textwidth}
   \centering
    \includegraphics[width=1.1\textwidth]{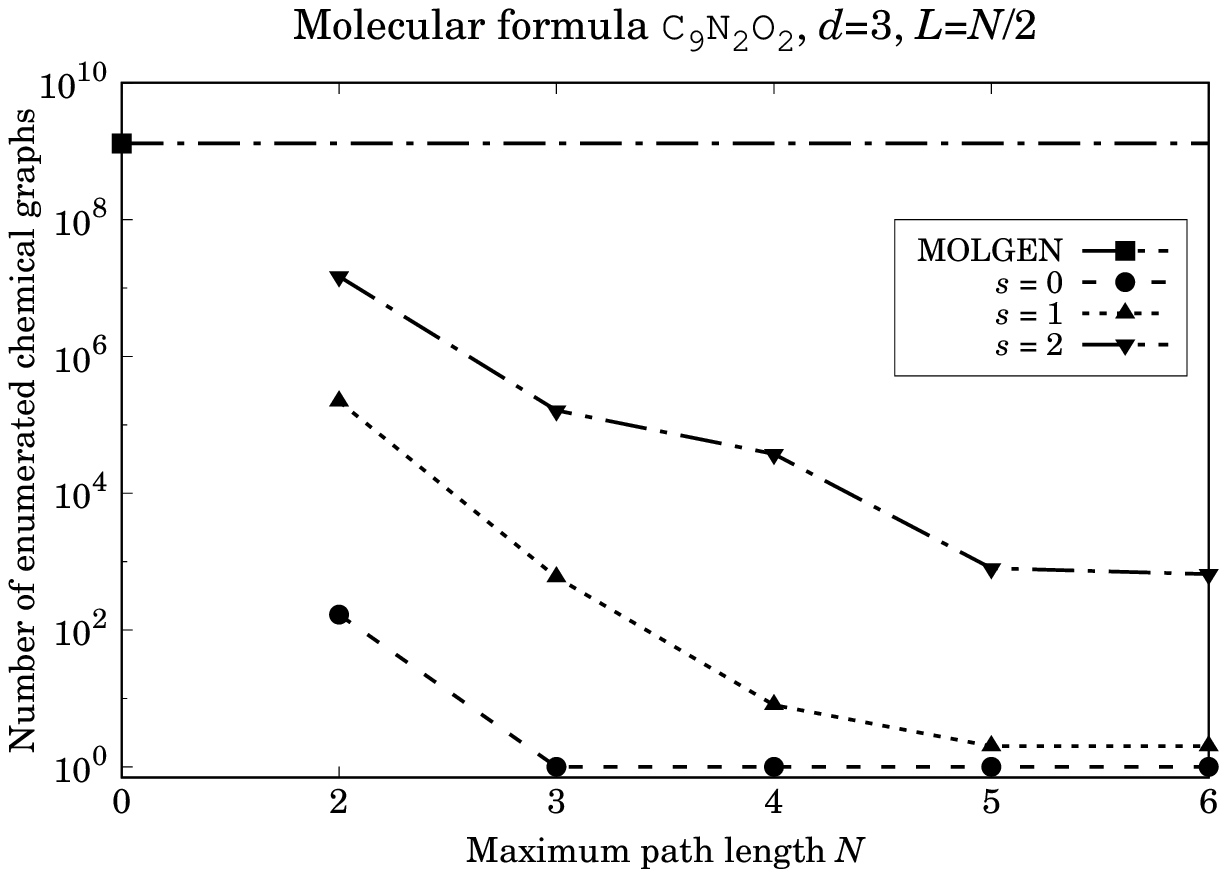}\\
    {\footnotesize (e)}\\
  \end{minipage} 
  \medskip

  \begin{minipage}{0.45\textwidth}
   \centering
      \includegraphics[width=1.1\textwidth]{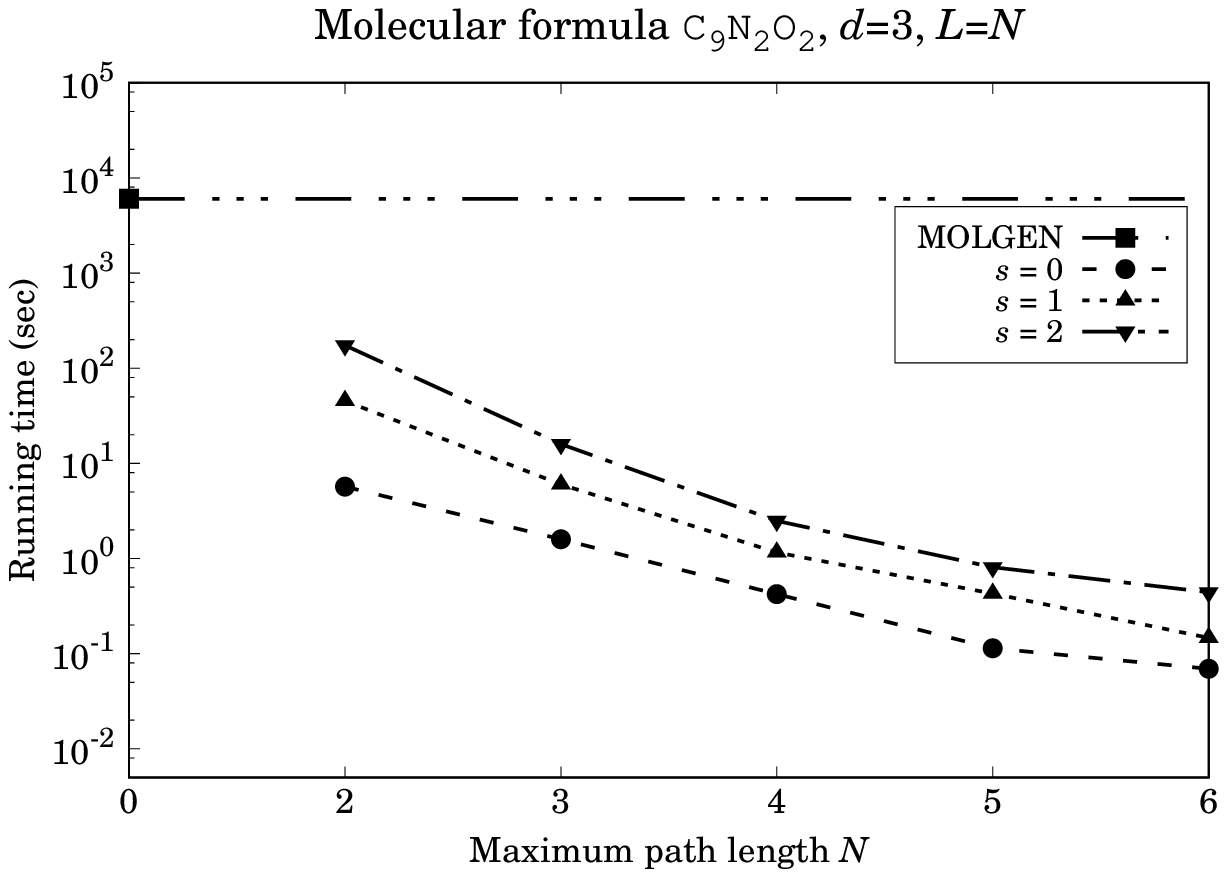}\\
      {\footnotesize (c)}\\
  \end{minipage} 
\hfill
  \begin{minipage}{0.45\textwidth}
   \centering
    \includegraphics[width=1.1\textwidth]{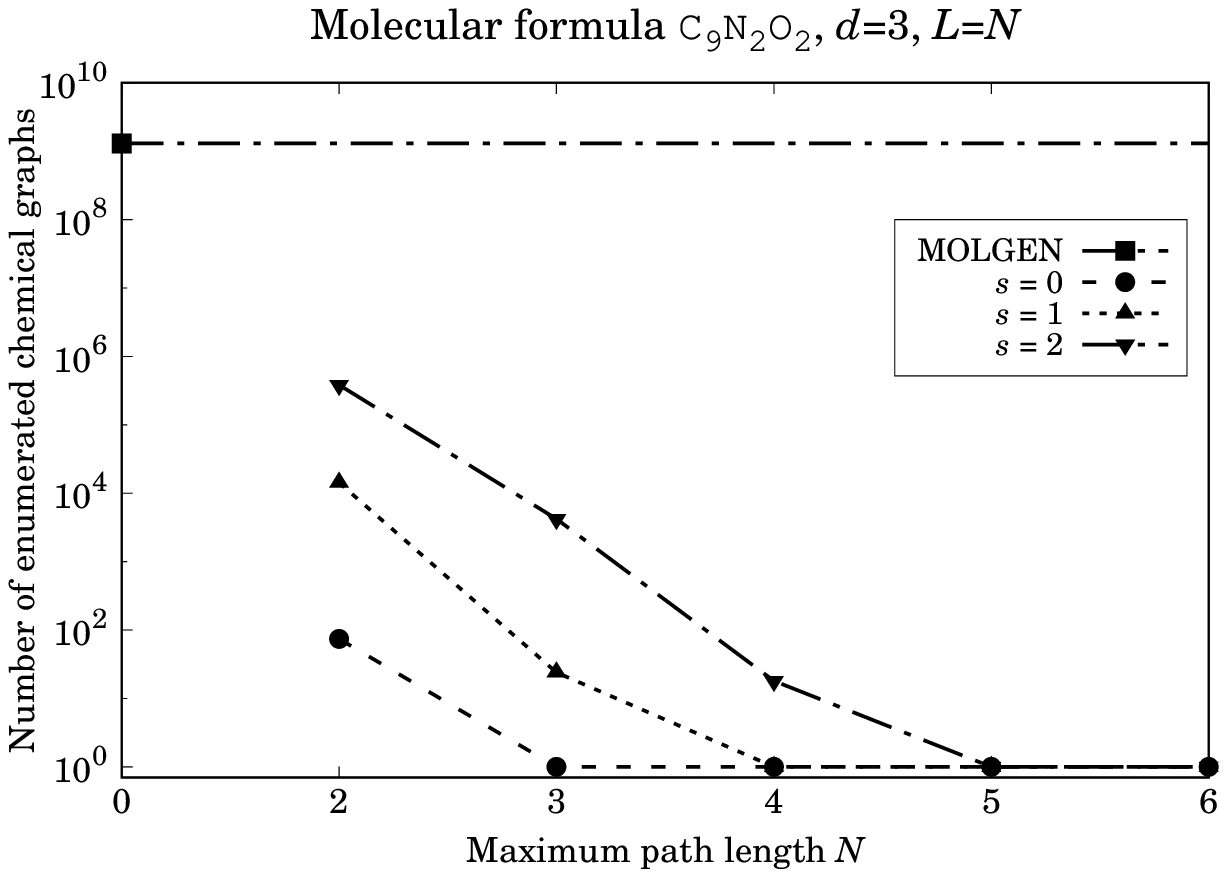}\\
    {\footnotesize (f)}\\
  \end{minipage} 
  \vspace{1cm}
  
  \caption{
    Plots showing the computation time 
    and number of chemical graphs enumerated by our algorithm
    for instance type EULF-$L$-A, as compared to MOLGEN.
    The sample structure from PubChem is with CID~23461643, 
    molecular formula {\tt C$_9$N$_2$O$_2$}, 
    and maximum bond multiplicity~$d=3$.
    (a)-(c)~Running time;
    (d)-(f)~Number of enumerated chemical graphs.
  }
 \label{fig:result_graphs_4}
 \end{figure}

 \begin{figure}[!ht]
  \begin{minipage}{0.45\textwidth}
   \centering
   \includegraphics[width=1.1\textwidth]{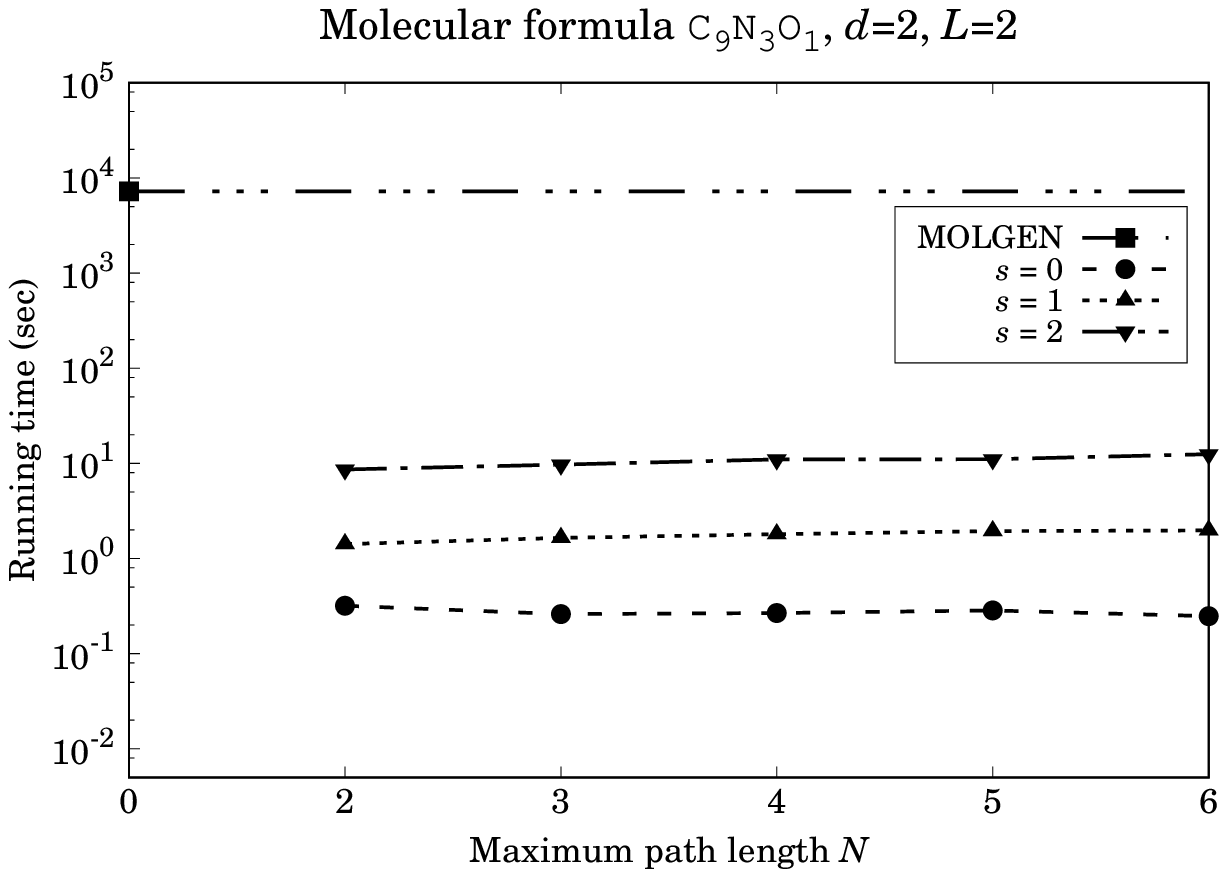}\\
   {\footnotesize (a)}\\
  \end{minipage}
\hfill
  \begin{minipage}{0.45\textwidth}
   \centering
   \includegraphics[width=1.1\textwidth]{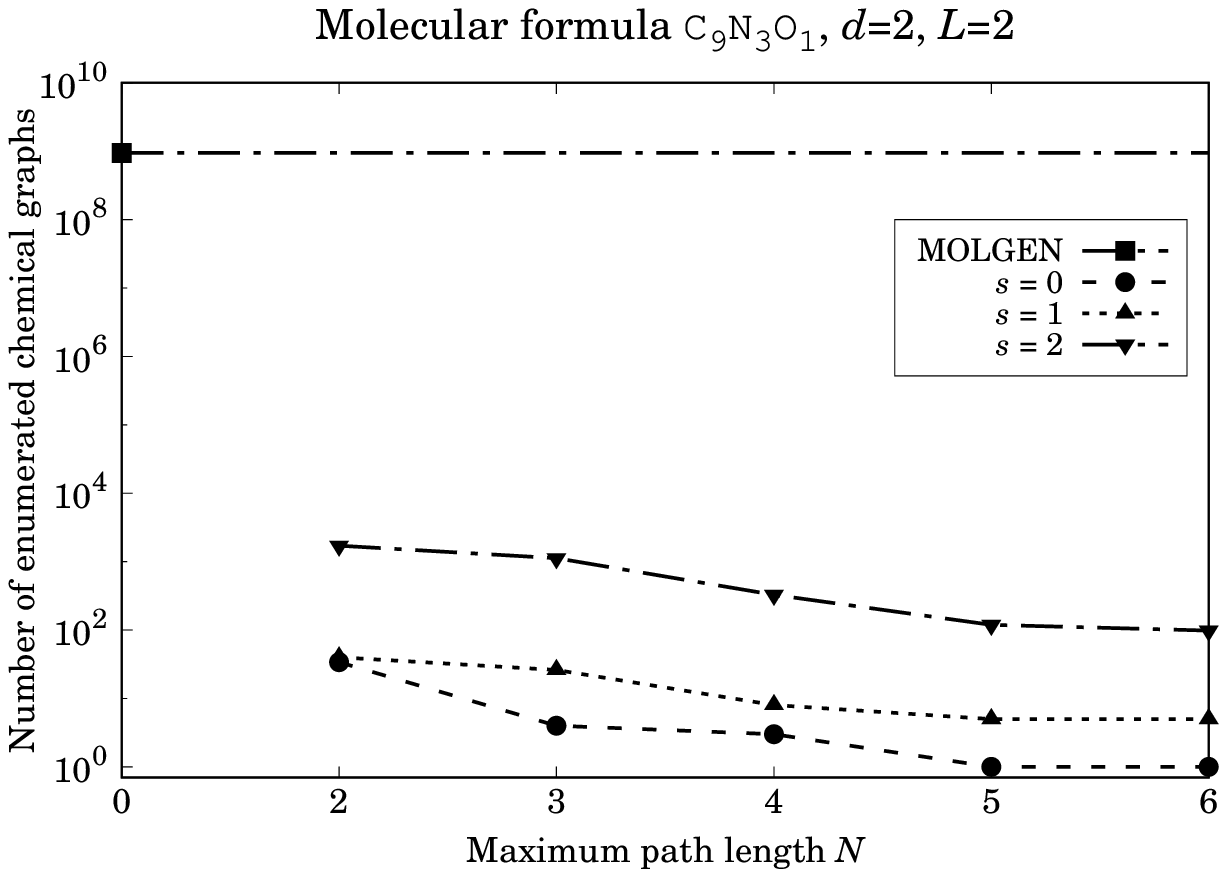}\\
   {\footnotesize (d)}\\
  \end{minipage} 
  \medskip

  \begin{minipage}{0.45\textwidth}
   \centering
      \includegraphics[width=1.1\textwidth]{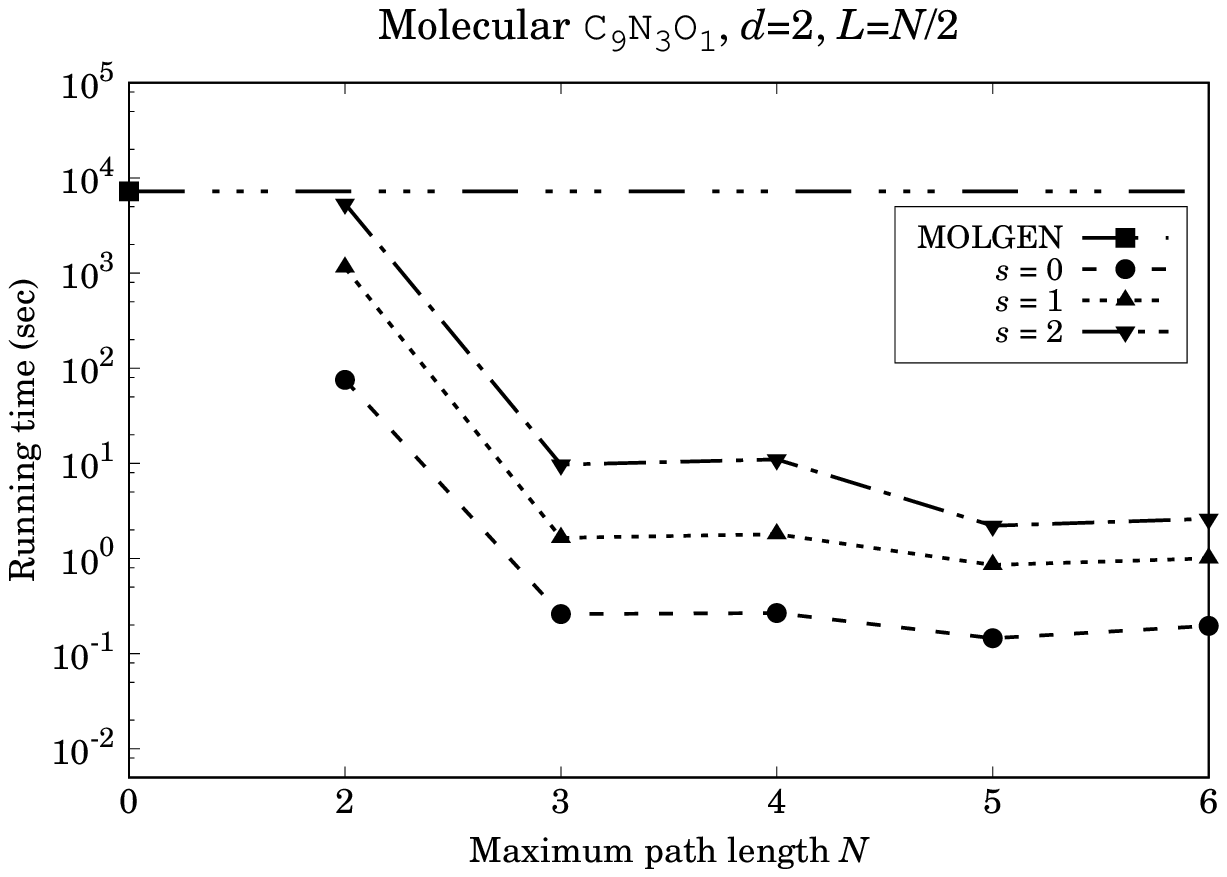}\\
      {\footnotesize (b)}\\
  \end{minipage} 
\hfill
  \begin{minipage}{0.45\textwidth}
   \centering
    \includegraphics[width=1.1\textwidth]{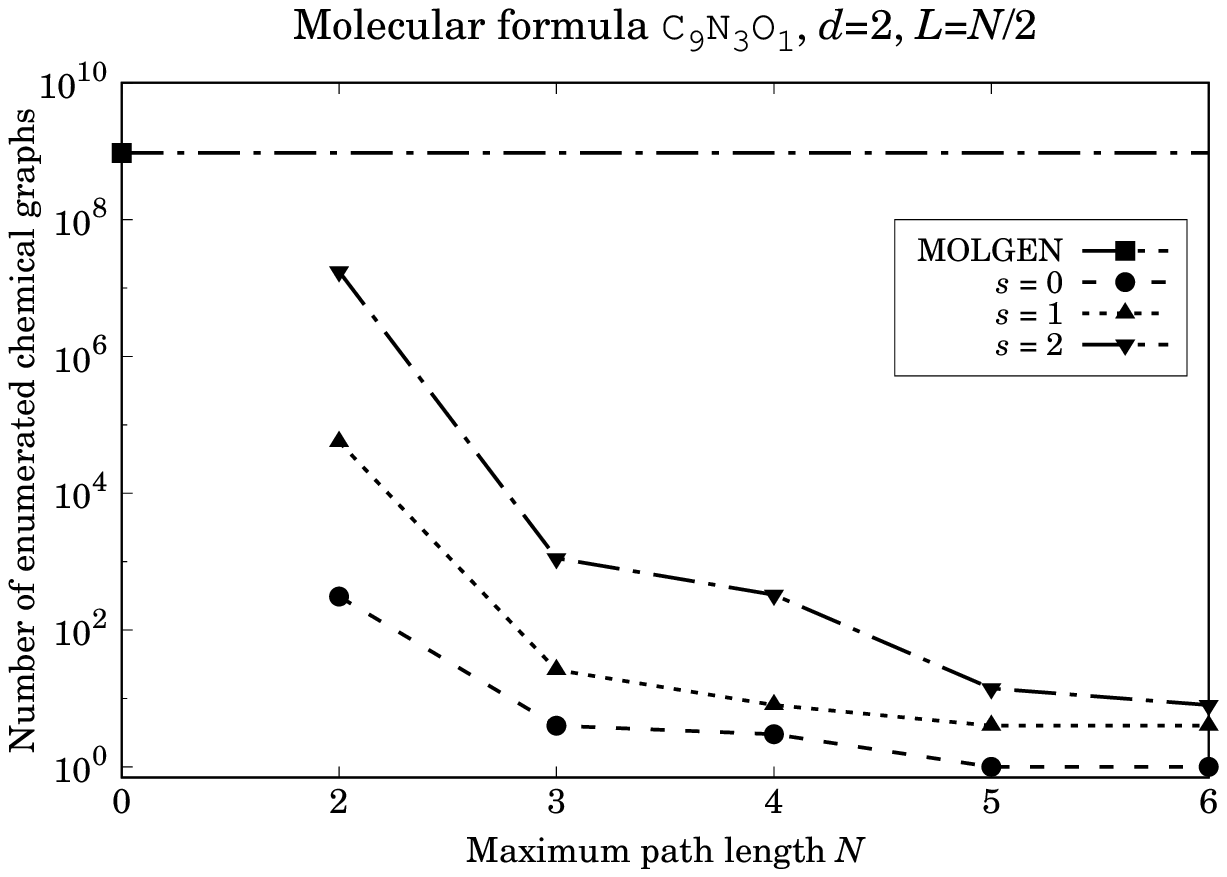}\\
    {\footnotesize (e)}\\
  \end{minipage} 
  \medskip

  \begin{minipage}{0.45\textwidth}
   \centering
      \includegraphics[width=1.1\textwidth]{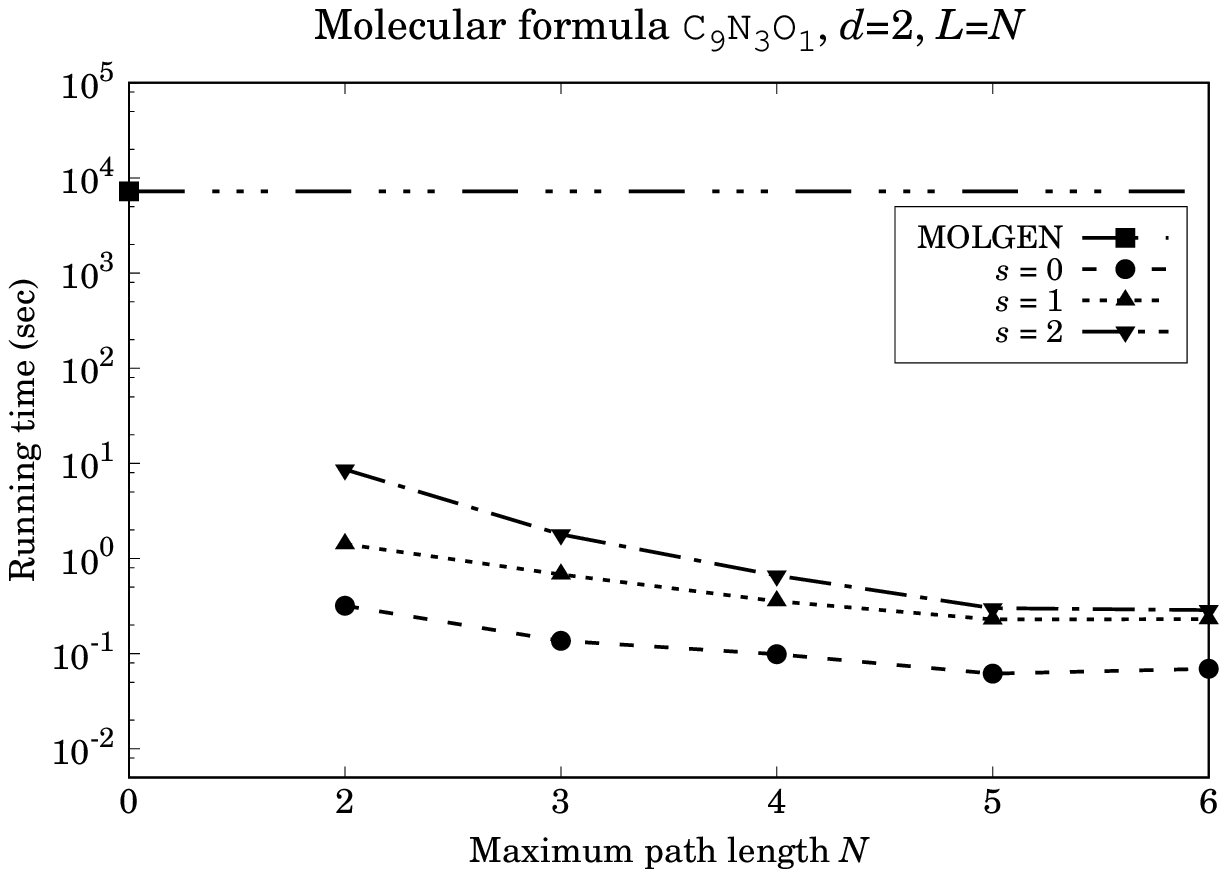}\\
      {\footnotesize (c)}\\
  \end{minipage} 
\hfill
  \begin{minipage}{0.45\textwidth}
   \centering
    \includegraphics[width=1.1\textwidth]{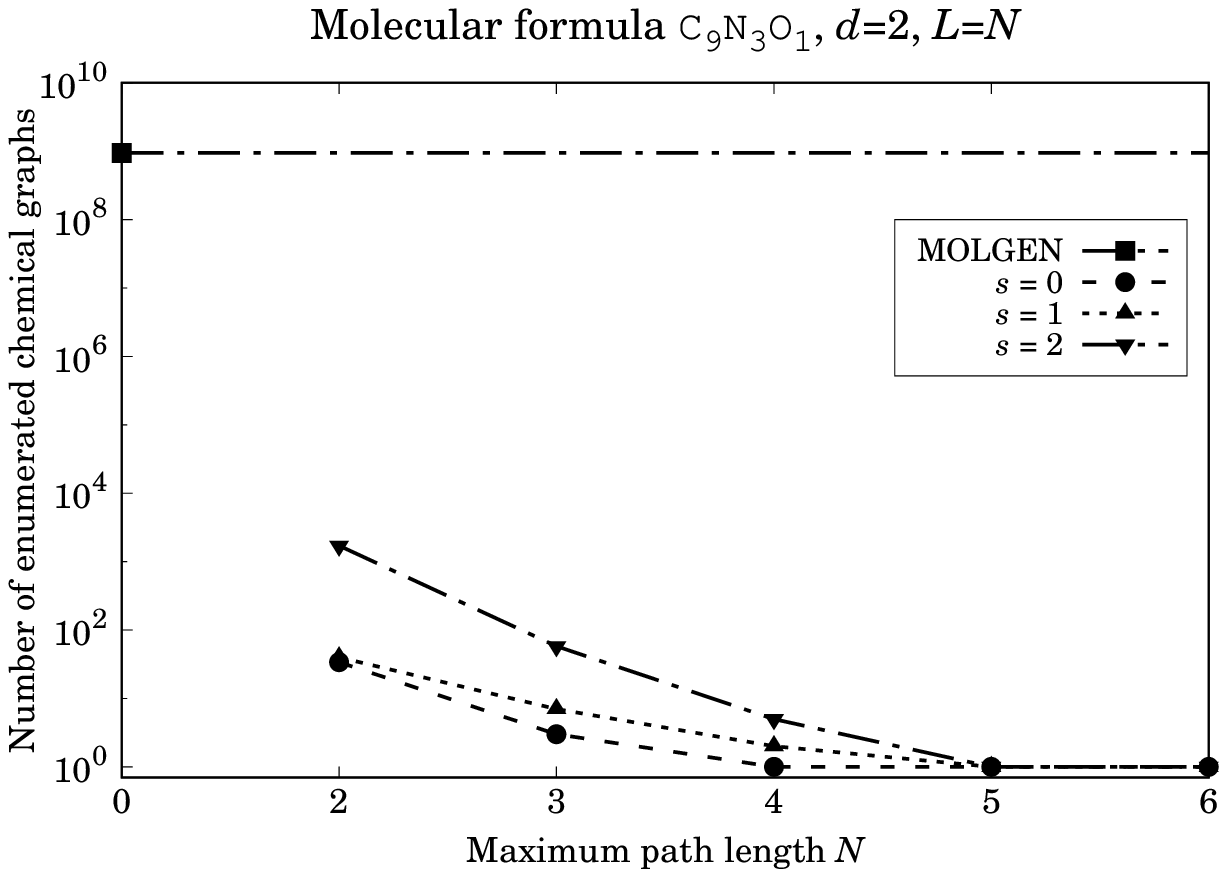}\\
    {\footnotesize (f)}\\
  \end{minipage} 
  \vspace{1cm}
  
  \caption{
    Plots showing the computation time 
    and number of chemical graphs enumerated by our algorithm
    for instance type EULF-$L$-A, as compared to MOLGEN.
    The sample structure from PubChem is with CID~10307896, 
    molecular formula {\tt C$_9$N$_3$O$_1$}, 
    and maximum bond multiplicity~$d=2$.
    (a)-(c)~Running time;
    (d)-(f)~Number of enumerated chemical graphs.
  }
 \label{fig:result_graphs_5}
 \end{figure}

 \begin{figure}[!ht]
  \begin{minipage}{0.45\textwidth}
   \centering
   \includegraphics[width=1.1\textwidth]{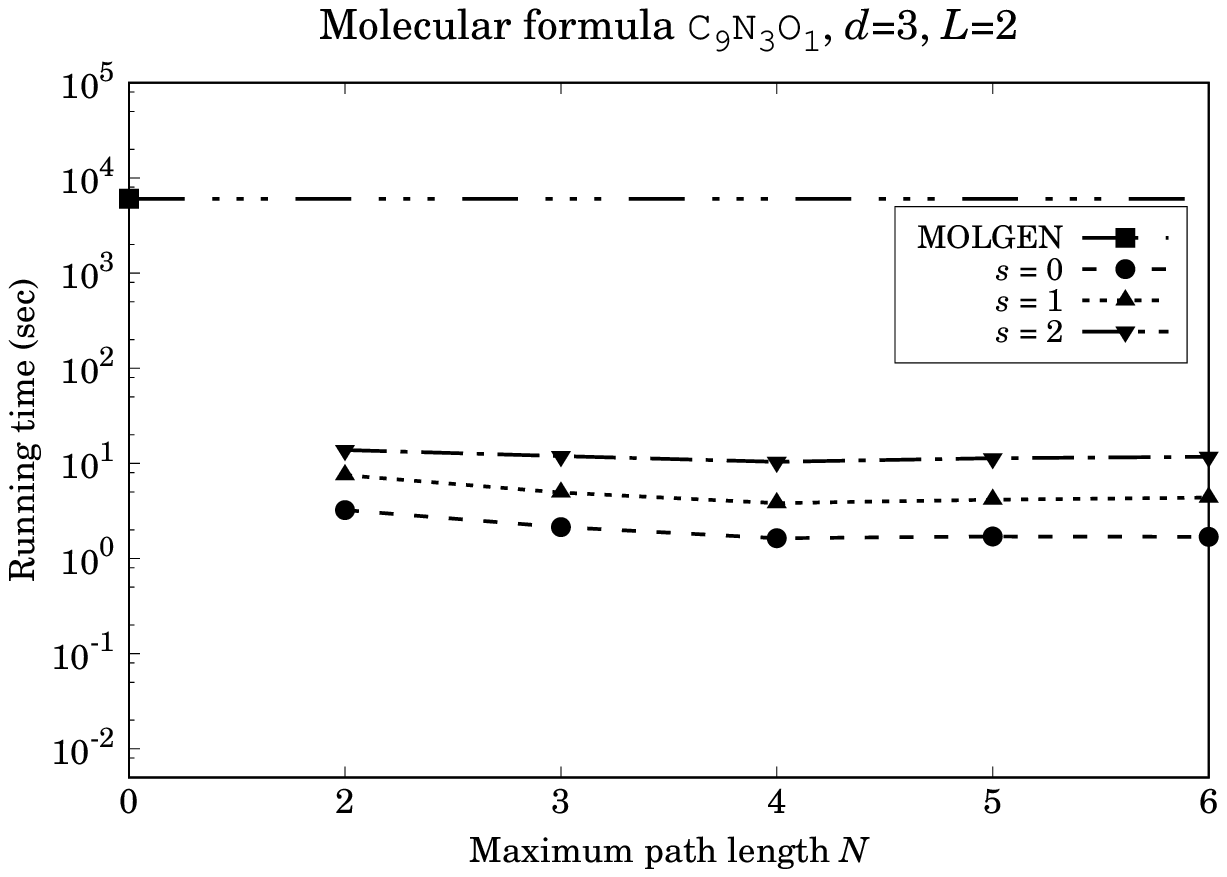}\\
   {\footnotesize (a)}\\
  \end{minipage}
\hfill
  \begin{minipage}{0.45\textwidth}
   \centering
   \includegraphics[width=1.1\textwidth]{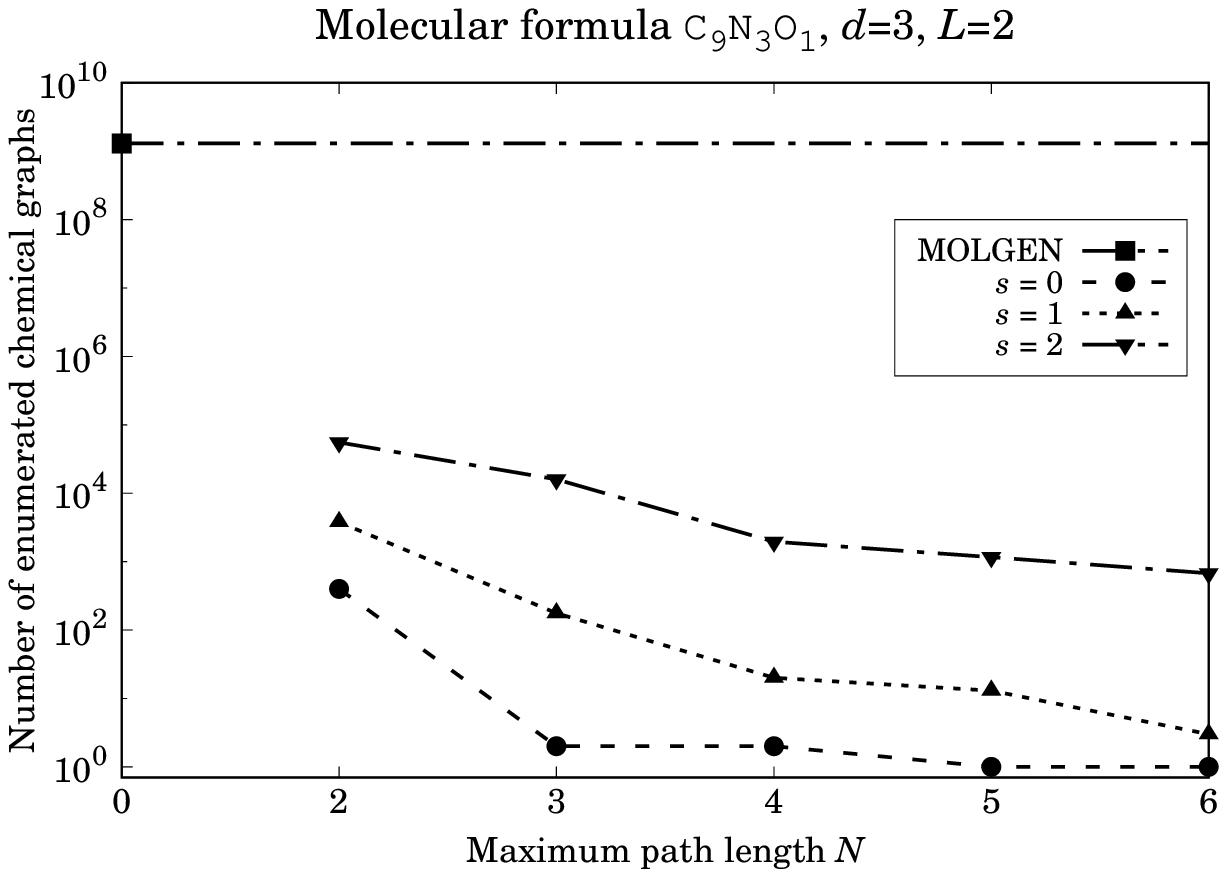}\\
   {\footnotesize (d)}\\
  \end{minipage} 
  \medskip

  \begin{minipage}{0.45\textwidth}
   \centering
      \includegraphics[width=1.1\textwidth]{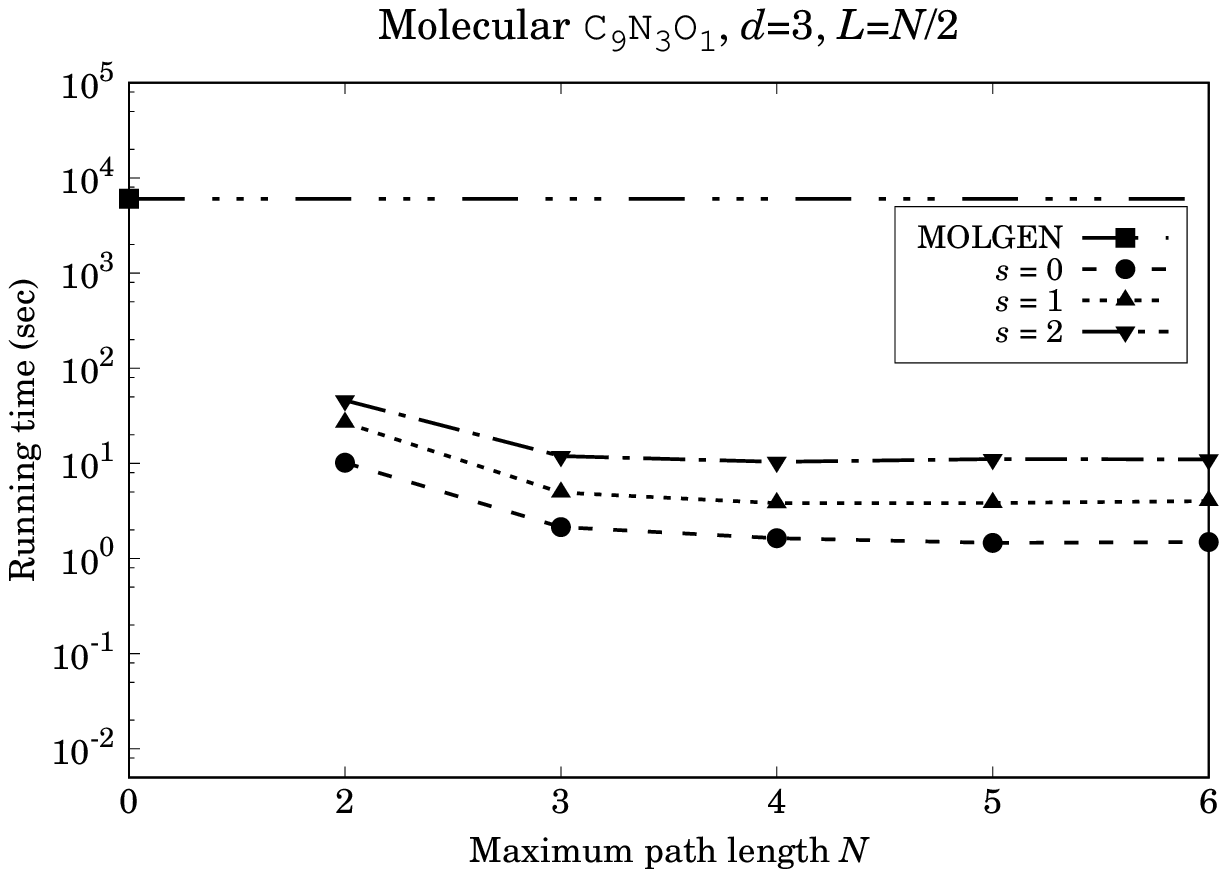}\\
      {\footnotesize (b)}\\
  \end{minipage} 
\hfill
  \begin{minipage}{0.45\textwidth}
   \centering
    \includegraphics[width=1.1\textwidth]{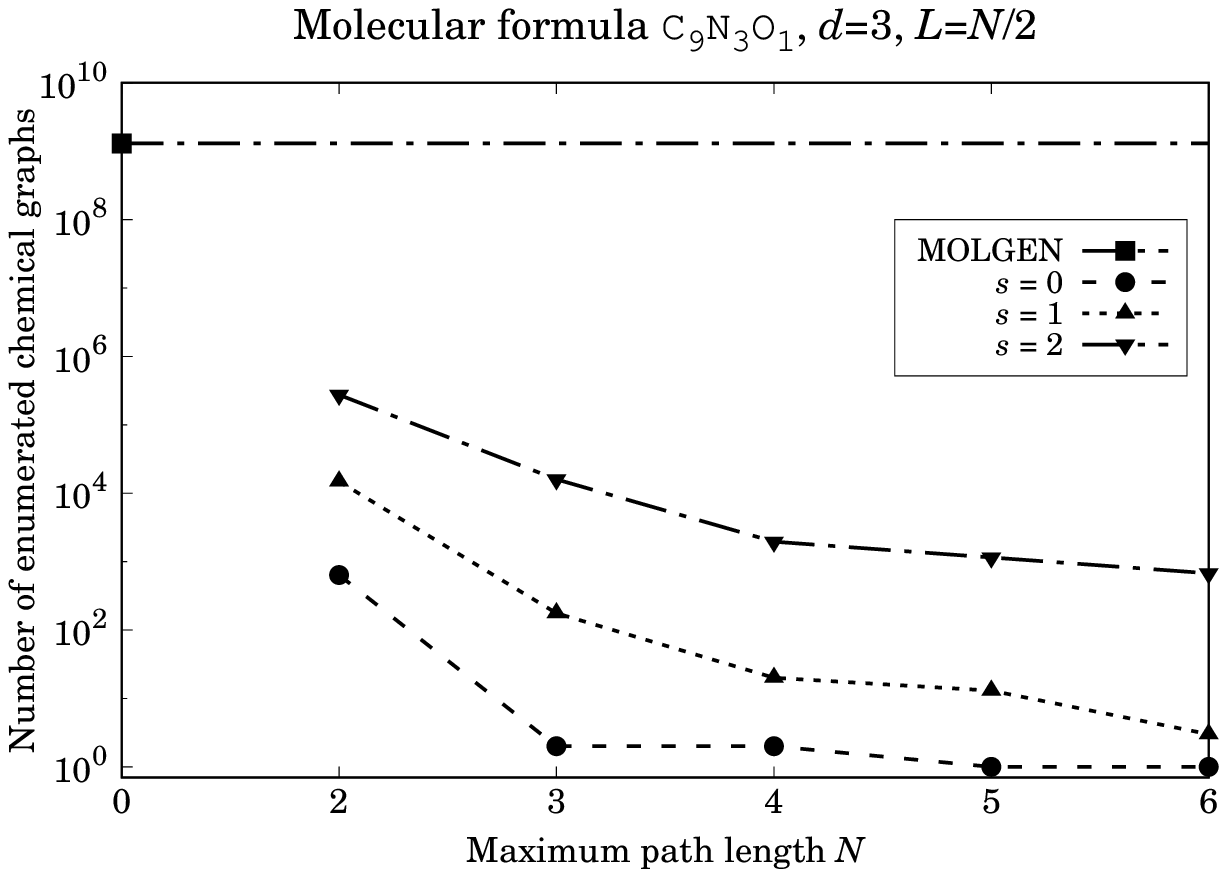}\\
    {\footnotesize (e)}\\
  \end{minipage} 
  \medskip

  \begin{minipage}{0.45\textwidth}
   \centering
      \includegraphics[width=1.1\textwidth]{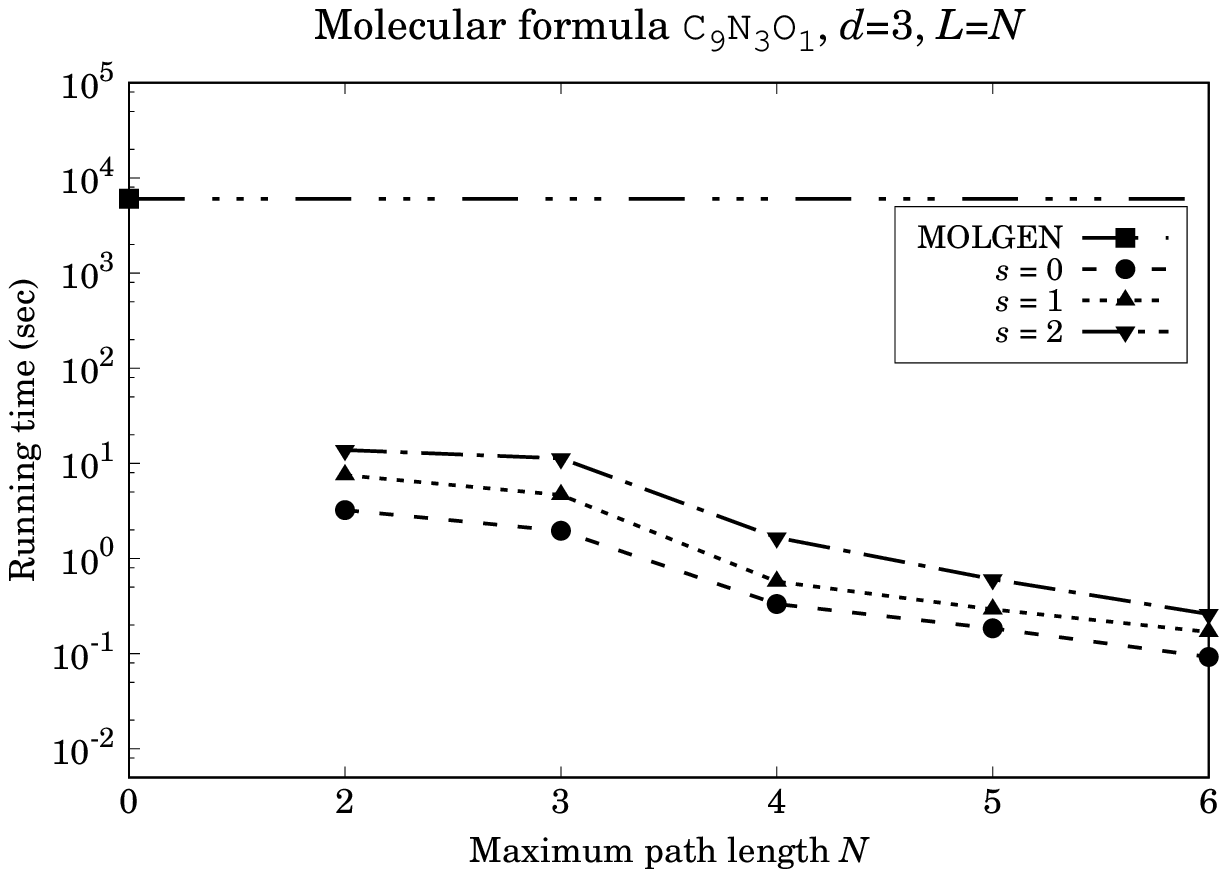}\\
      {\footnotesize (c)}\\
  \end{minipage} 
\hfill
  \begin{minipage}{0.45\textwidth}
   \centering
    \includegraphics[width=1.1\textwidth]{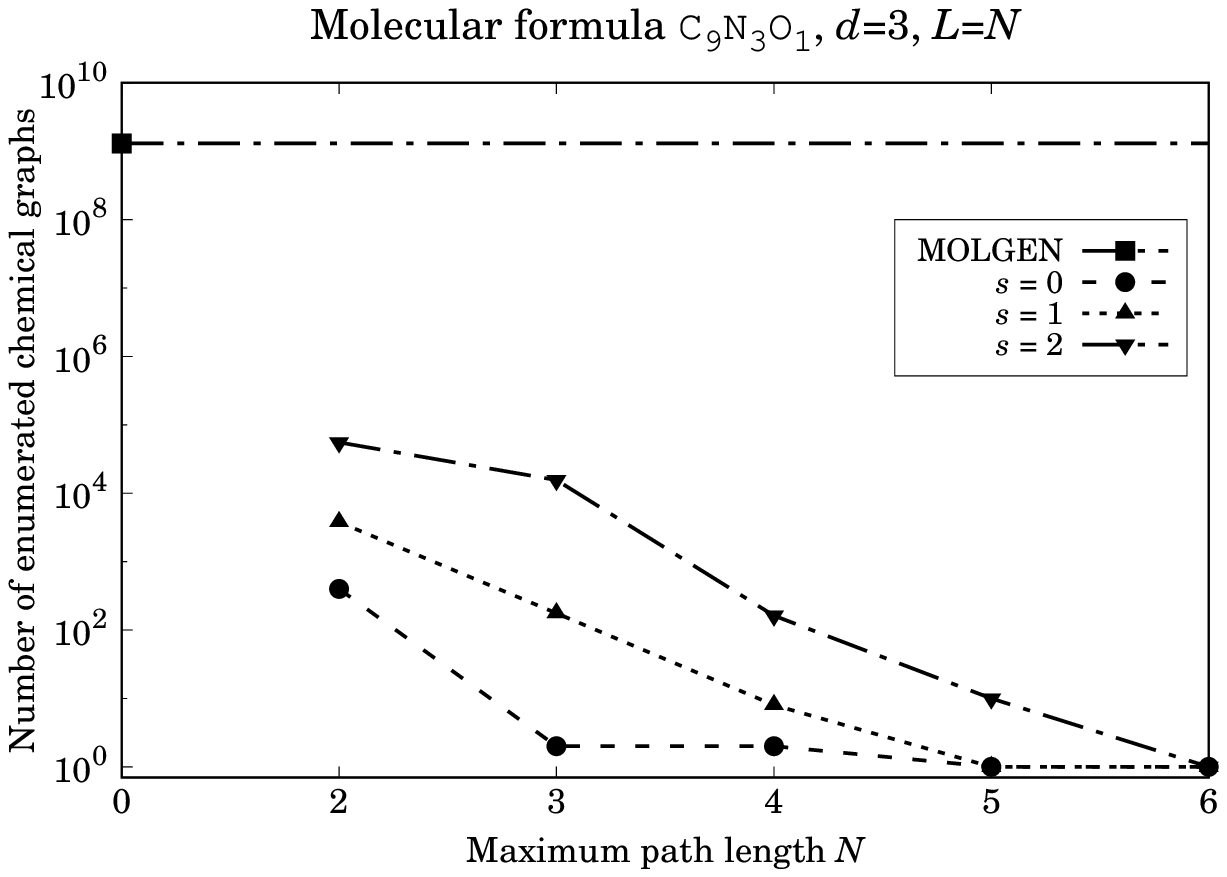}\\
    {\footnotesize (f)}\\
  \end{minipage} 
  \vspace{1cm}
  
  \caption{
    Plots showing the computation time 
    and number of chemical graphs enumerated by our algorithm
    for instance type EULF-$L$-A, as compared to MOLGEN.
    The sample structure from PubChem is with CID~10307899, 
    molecular formula {\tt C$_9$N$_3$O$_1$}, 
    and maximum bond multiplicity~$d=3$.
    (a)-(c)~Running time;
    (d)-(f)~Number of enumerated chemical graphs.
  }
 \label{fig:result_graphs_6}
 \end{figure}

 In addition, to check the limits as to the maximum number of vertices
 in graphs that can be enumerated in a reasonable time.
 Preliminary experiments indicated that the time limit of $3, 600$
 is reached when the number $n$ of 
 vertices is around 15, and
 we conducted experiments over a range $n \in [15, 45]$ for the number 
 of vertices in a target chemical graph.
 For a fixed number $n$ of vertices, 
 we tested two types of instances, one with molecular formula $\mathtt{C}_n$, 
 and the other with molecular formula  $\mathtt{C}_{n-4} \mathtt{N}_2 \mathtt{O}_2$, 
 and set an execution time limit of $3, 600$ seconds.
 The results are summarized in Fig.~\ref{fig:result_graphs_num_limit}.
 From Fig.~\ref{fig:result_graphs_num_limit}\, (a) and~(b), we see that the
 program still enumerates 
 structures within the time limit up to $n=30$, after which there are cases when
 not even a single chemical graph is enumerated during the time limit.

   \begin{figure}[!ht]
  \begin{minipage}{0.45\textwidth}
   \centering
   \includegraphics[width=1.1\textwidth]{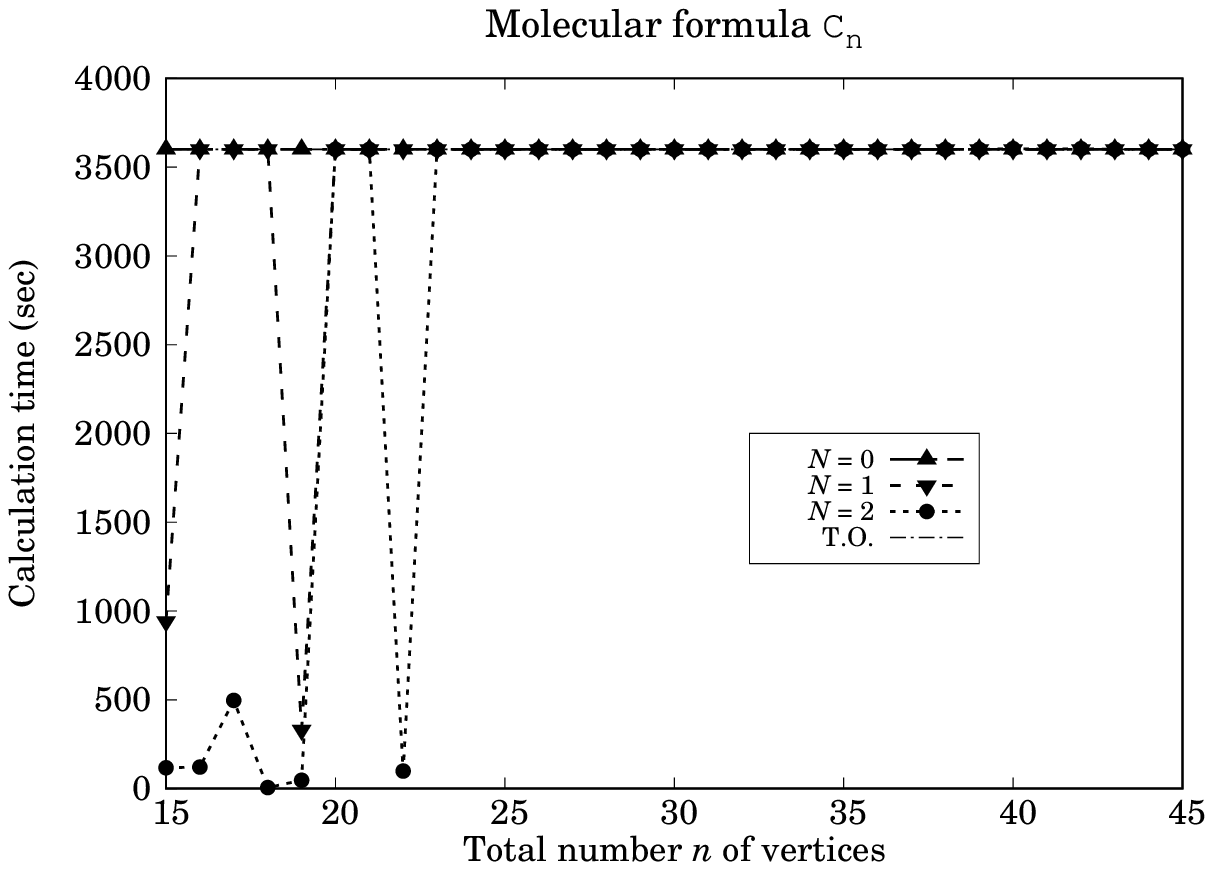}\\
   {\footnotesize (a)}\\
  \end{minipage}
\hfill
  \begin{minipage}{0.45\textwidth}
   \centering
   \includegraphics[width=1.1\textwidth]{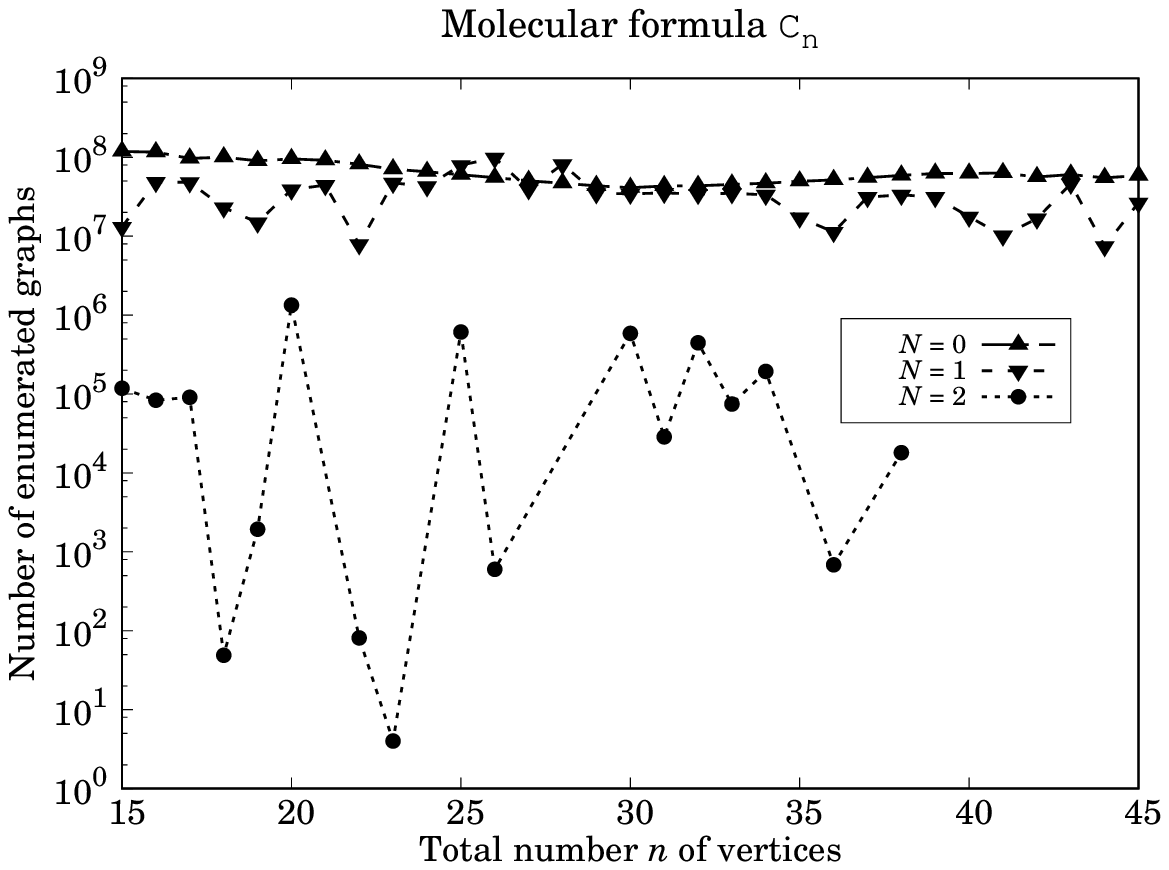}\\
   {\footnotesize (c)}\\
  \end{minipage} 
  \medskip

  \begin{minipage}{0.45\textwidth}
   \centering
      \includegraphics[width=1.1\textwidth]{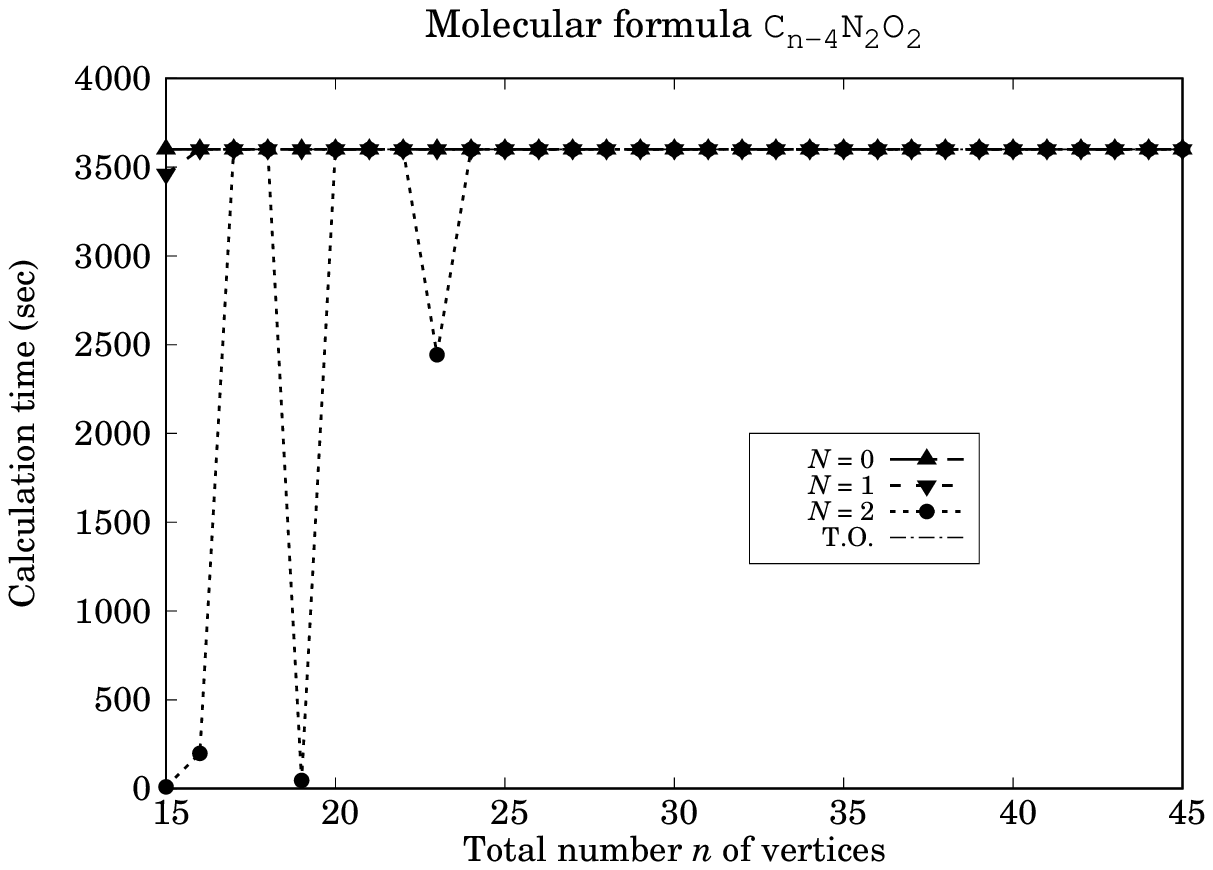}\\
      {\footnotesize (b)}\\
  \end{minipage} 
\hfill
  \begin{minipage}{0.45\textwidth}
   \centering
    \includegraphics[width=1.1\textwidth]{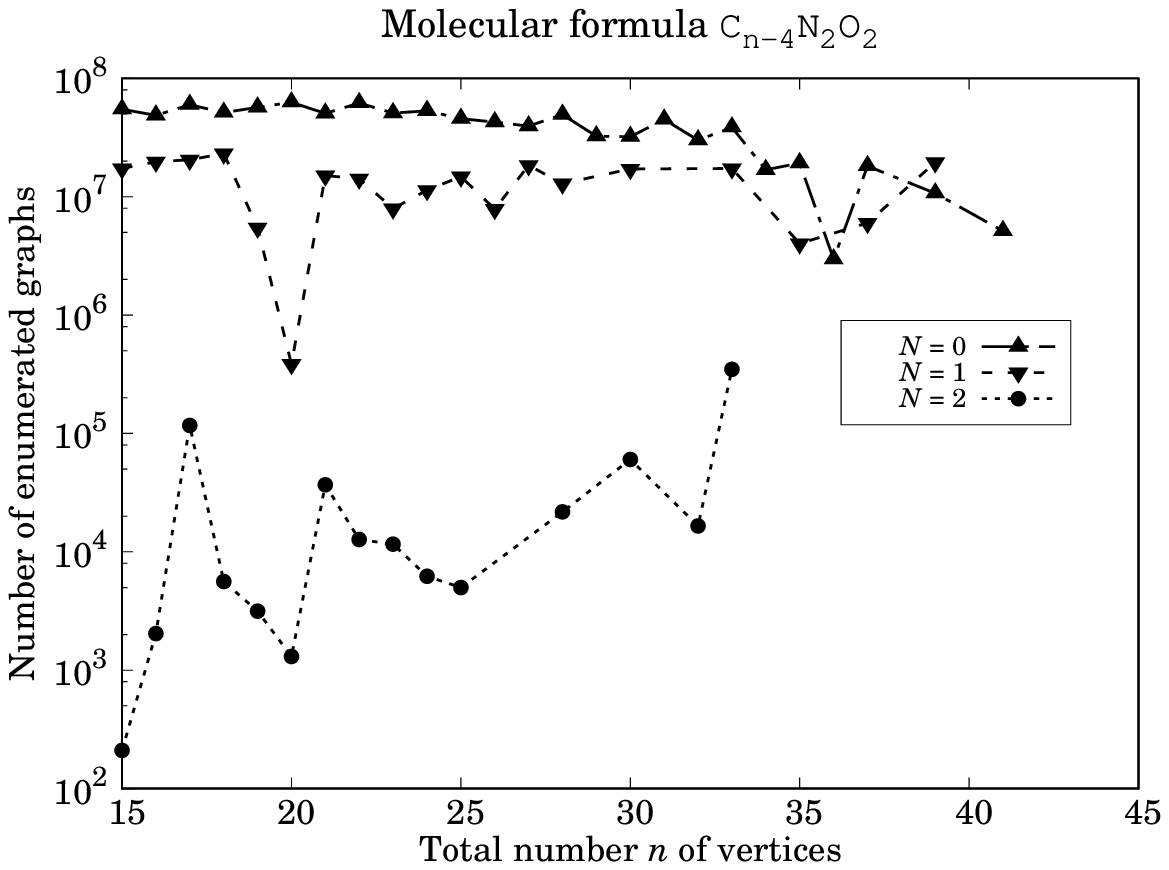}\\
    {\footnotesize (d)}\\
  \end{minipage} 
  \medskip
  
  \vspace{1cm}
  
  \caption{
    Plots showing the computation time 
    and number of chemical graphs enumerated by our algorithm
    for instance type EULF-$L$-A, over ranges for $N \in [0, 2]$, $L = N$, 
    $d = 3$ and $s = 0$.
    (a), (b)~Running time, T.O. stands for ``Time Out'';
    (c), (d)~Number of enumerated chemical graphs.
  }
 \label{fig:result_graphs_num_limit}
 \end{figure}

 \subsection{Experimental Results on EULF-$L$-P}
 \label{sec:experiments_EULF2}

We conducted similar computational 
experiments to test the performance of
our algorithm for Problem EULF-$L$-P
as in Section~\ref{sec:experiments_EULF1}.
We took values for $N \in [8, 11]$, 
and $L \in \{ 2, 3 \}$.

The results from our experiments 
for instance type EULF-$L$-P are summarized in 
Figs.~\ref{fig:result_graphs_1.2} to~\ref{fig:result_graphs_6.2}.
Our results for instance type EULF-$L$-P 
indicate that there in fact do not exist chemical graphs that
satisfy the path frequency specification for
our choice of a set $\pathset$ of colored paths
obtained from the six compounds from the PubChem database, parameter~$L$
for prescribed maximum path length less than the maximum
path length in the sample chemical graph, in our case $N \leq 10$.
This could be due to the structure of bi-block 2-augmented structures, 
namely, a single path frequency specification
has a strong influence on the structure of a chemical graph.

In addition, we observe that the running time of our algorithm, even when there are no
enumerated chemical graphs, grows rapidly, an order of magnitude,
with the value of the parameter~$L$.
It is an interesting idea for future research to 
improve our algorithm in such a way that the non-existence of any
chemical graphs that satisfy a given path frequency specification is determined 
much quicker.

  \begin{figure}[!ht]
  \begin{minipage}{0.45\textwidth}
   \centering
   \includegraphics[width=1.1\textwidth]{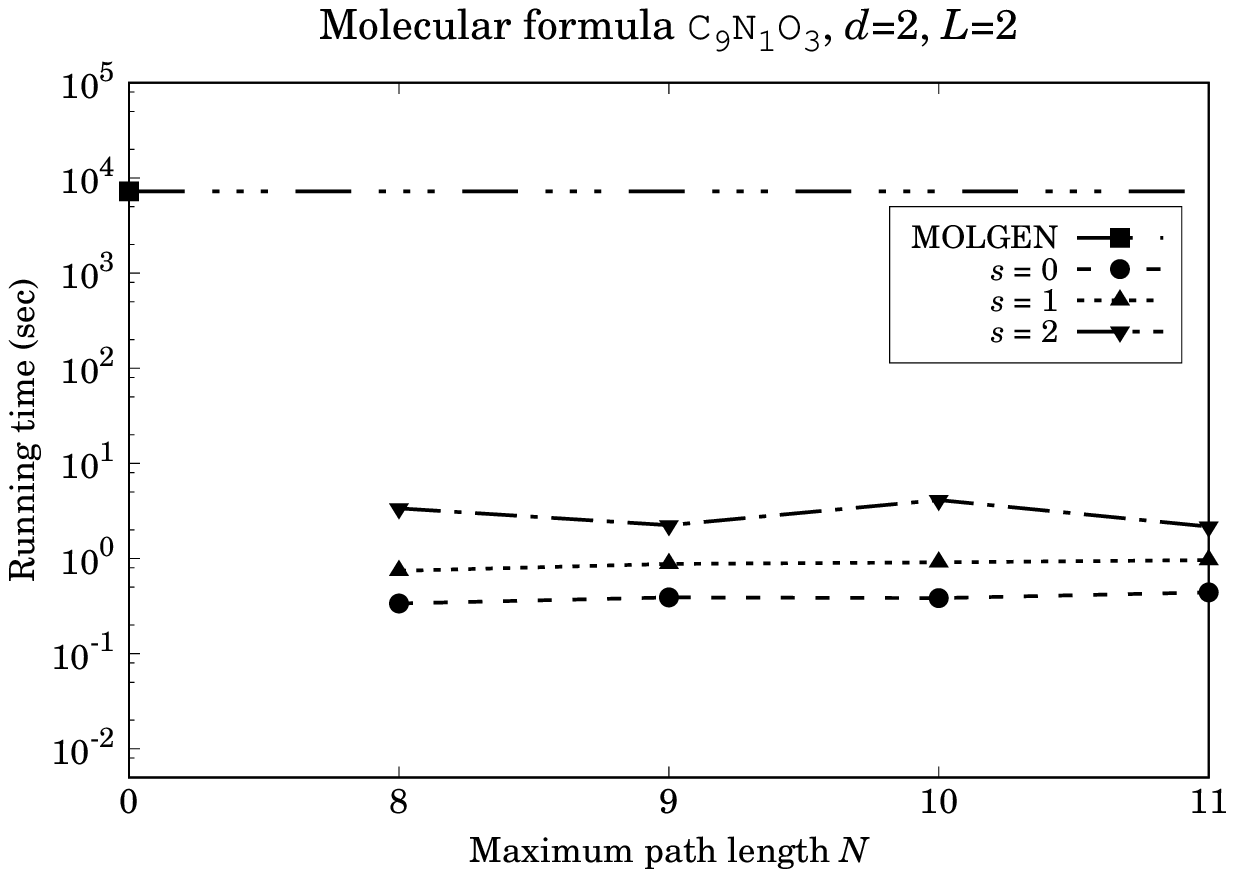}\\
   {\footnotesize (a)}\\
  \end{minipage}
\hfill
  \begin{minipage}{0.45\textwidth}
   \centering
   \includegraphics[width=1.1\textwidth]{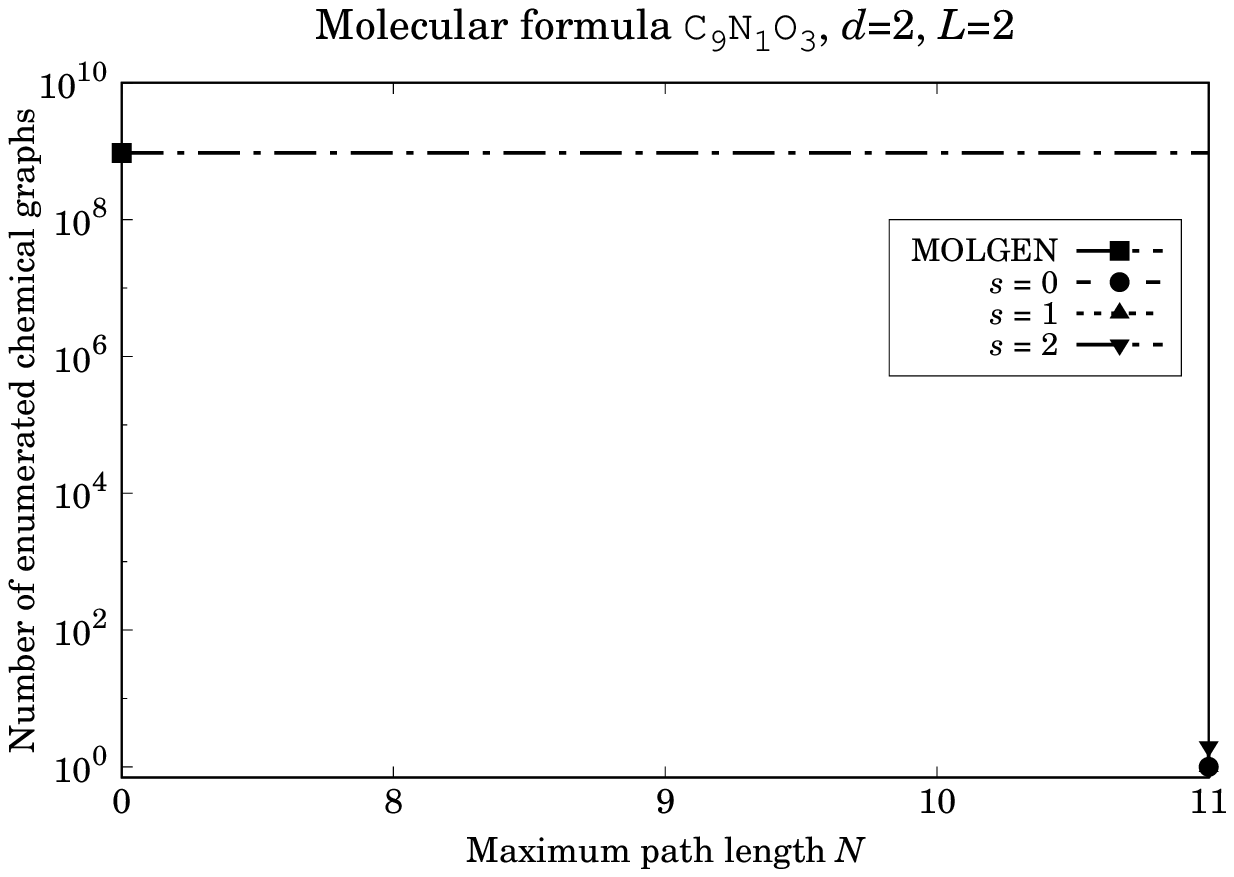}\\
   {\footnotesize (c)}\\
  \end{minipage} 
  \medskip

  \begin{minipage}{0.45\textwidth}
   \centering
      \includegraphics[width=1.1\textwidth]{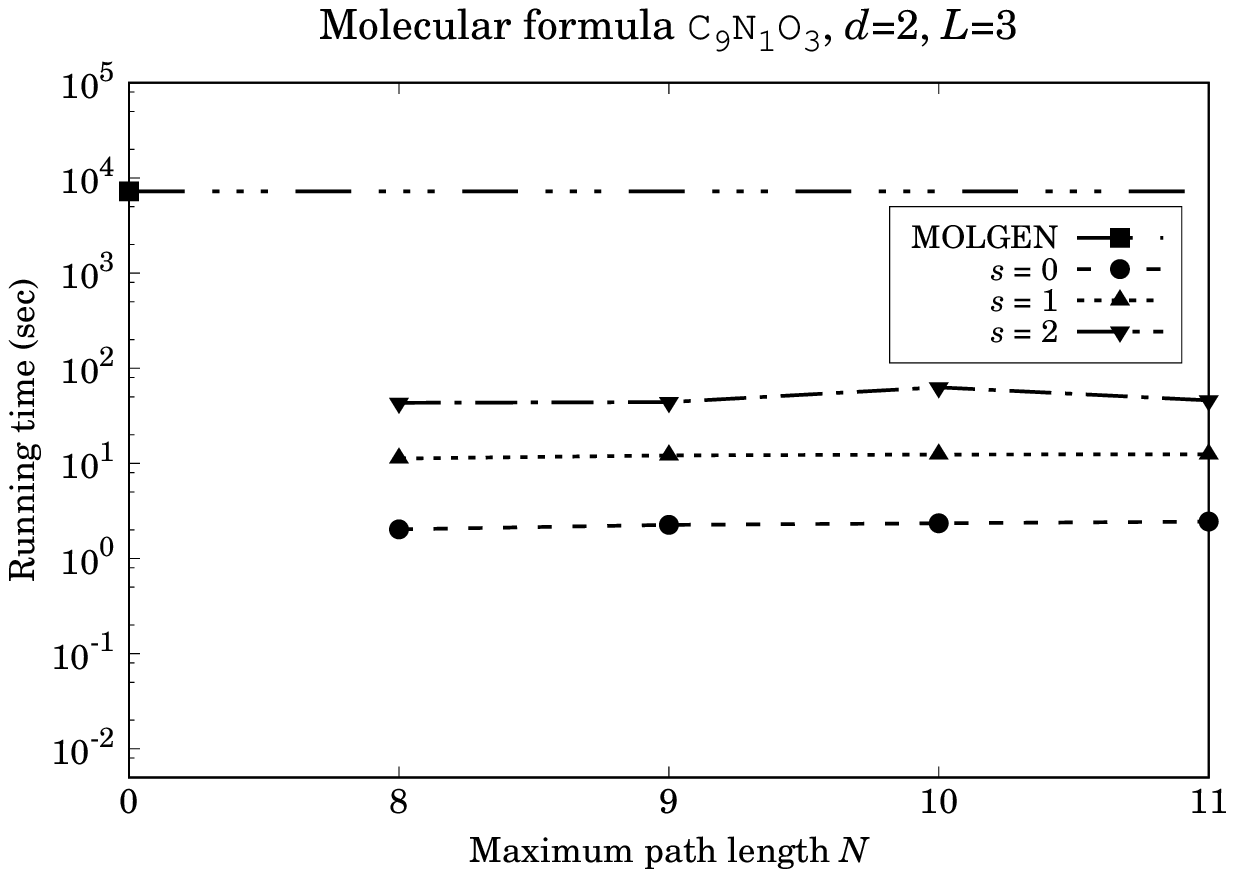}\\
      {\footnotesize (b)}\\
  \end{minipage} 
\hfill
  \begin{minipage}{0.45\textwidth}
   \centering
    \includegraphics[width=1.1\textwidth]{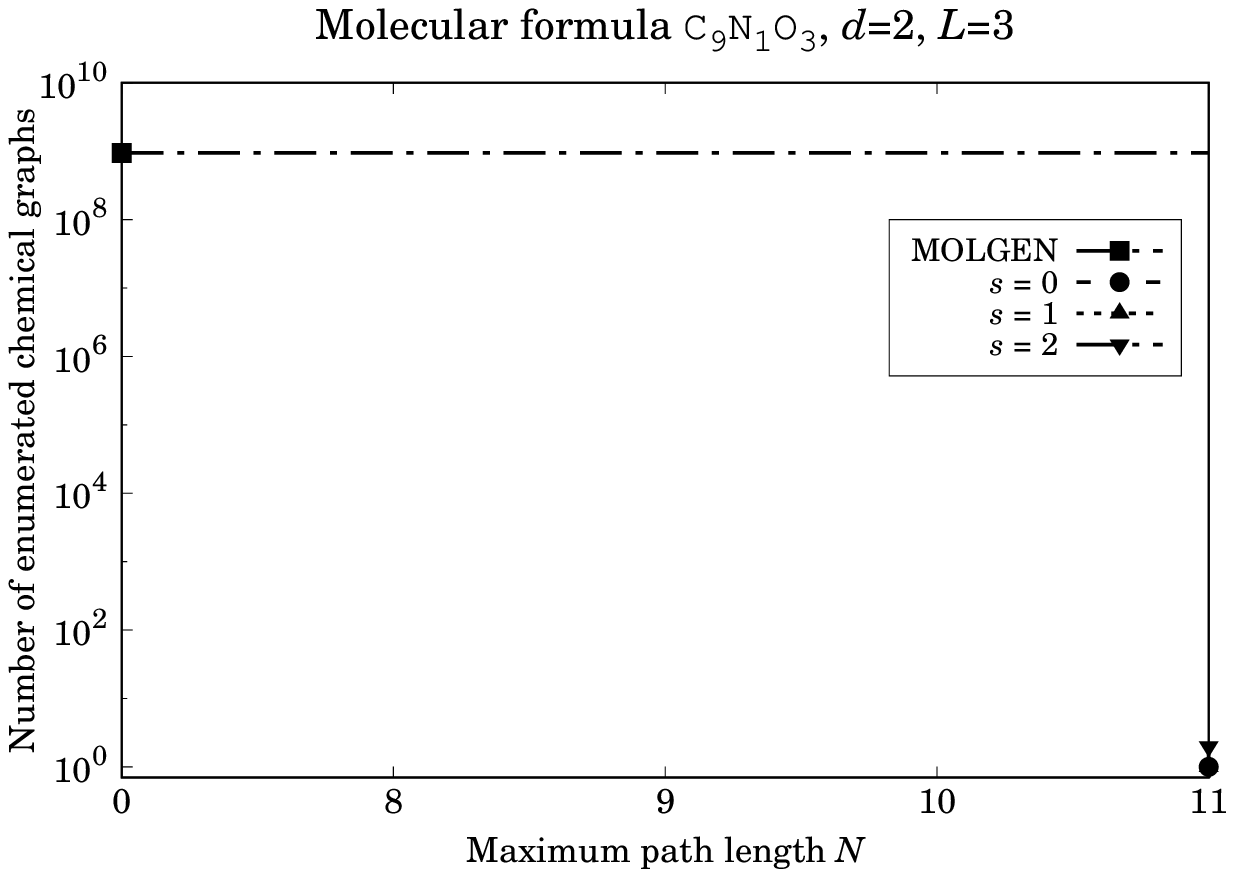}\\
    {\footnotesize (d)}\\
  \end{minipage} 
  \medskip
  
  \vspace{1cm}
  
  \caption{
    Plots showing the computation time 
    and number of chemical graphs enumerated by our algorithm
    for instance type EULF-$L$-P, as compared to MOLGEN.
    The sample structure from PubChem is with CID~130964701, 
    molecular formula {\tt C$_9$N$_1$O$_3$}, 
    and maximum bond multiplicity~$d=2$.
    (a), (b)~Running time;
    (c), (d)~Number of enumerated chemical graphs.
  }
 \label{fig:result_graphs_1.2}
 \end{figure}

  \begin{figure}[!ht]
  \begin{minipage}{0.45\textwidth}
   \centering
   \includegraphics[width=1.1\textwidth]{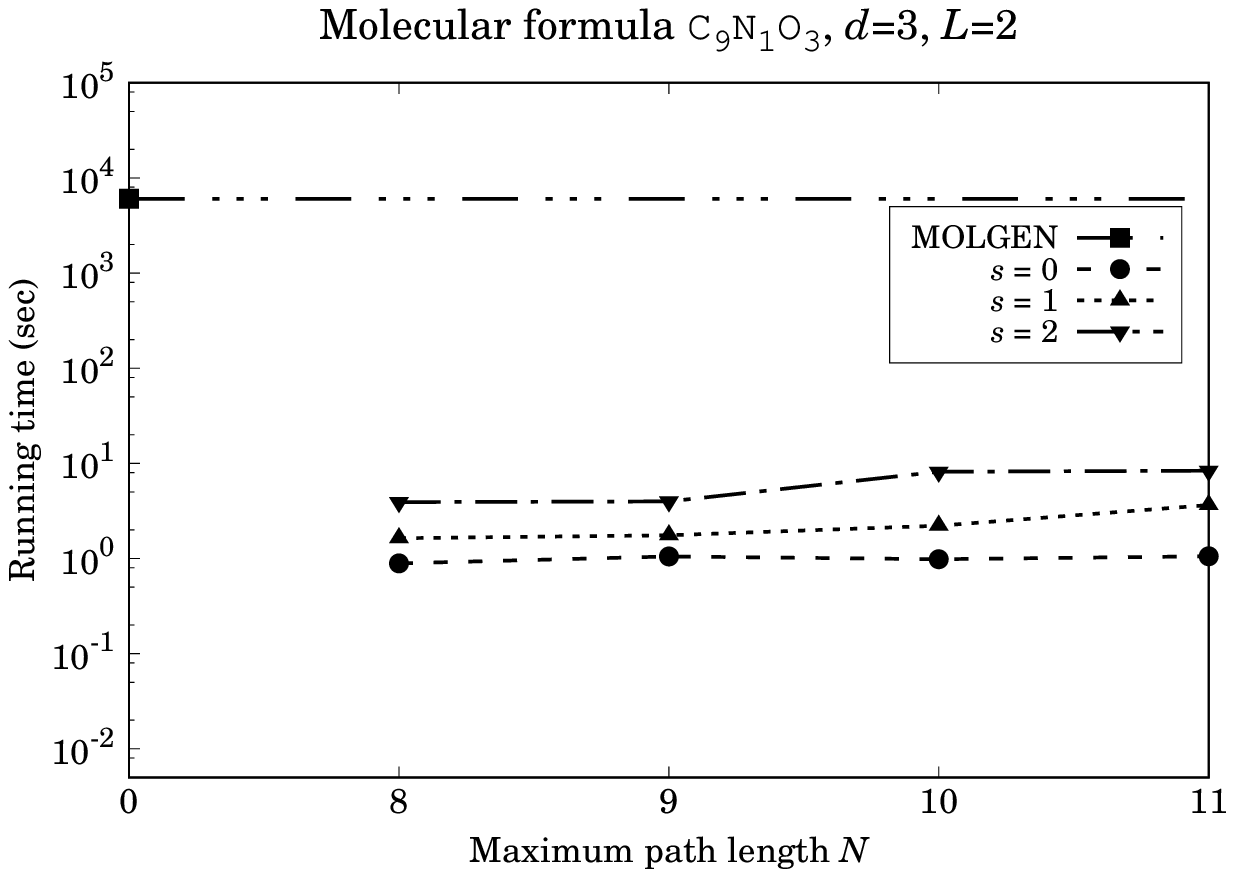}\\
   {\footnotesize (a)}\\
  \end{minipage}
\hfill
  \begin{minipage}{0.45\textwidth}
   \centering
   \includegraphics[width=1.1\textwidth]{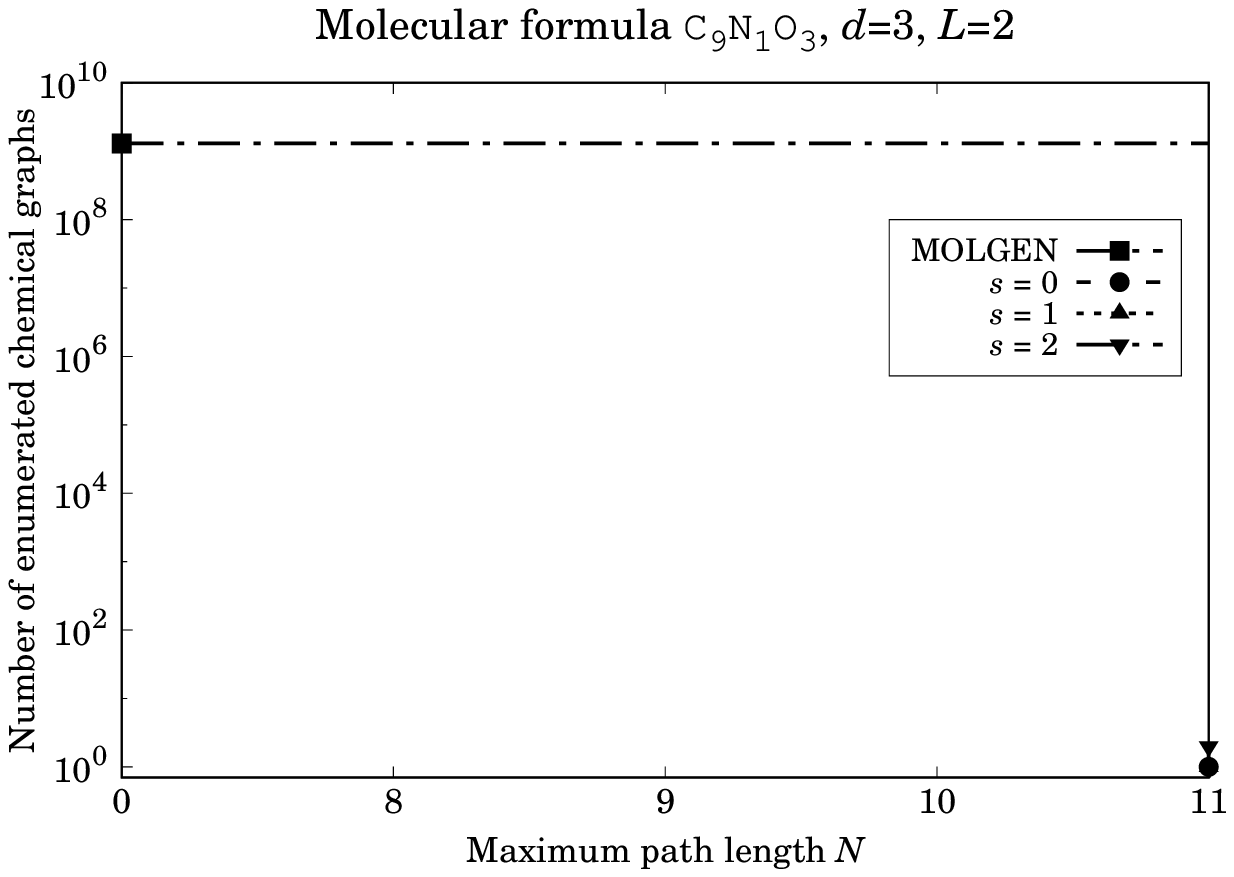}\\
   {\footnotesize (c)}\\
  \end{minipage} 
  \medskip

  \begin{minipage}{0.45\textwidth}
   \centering
      \includegraphics[width=1.1\textwidth]{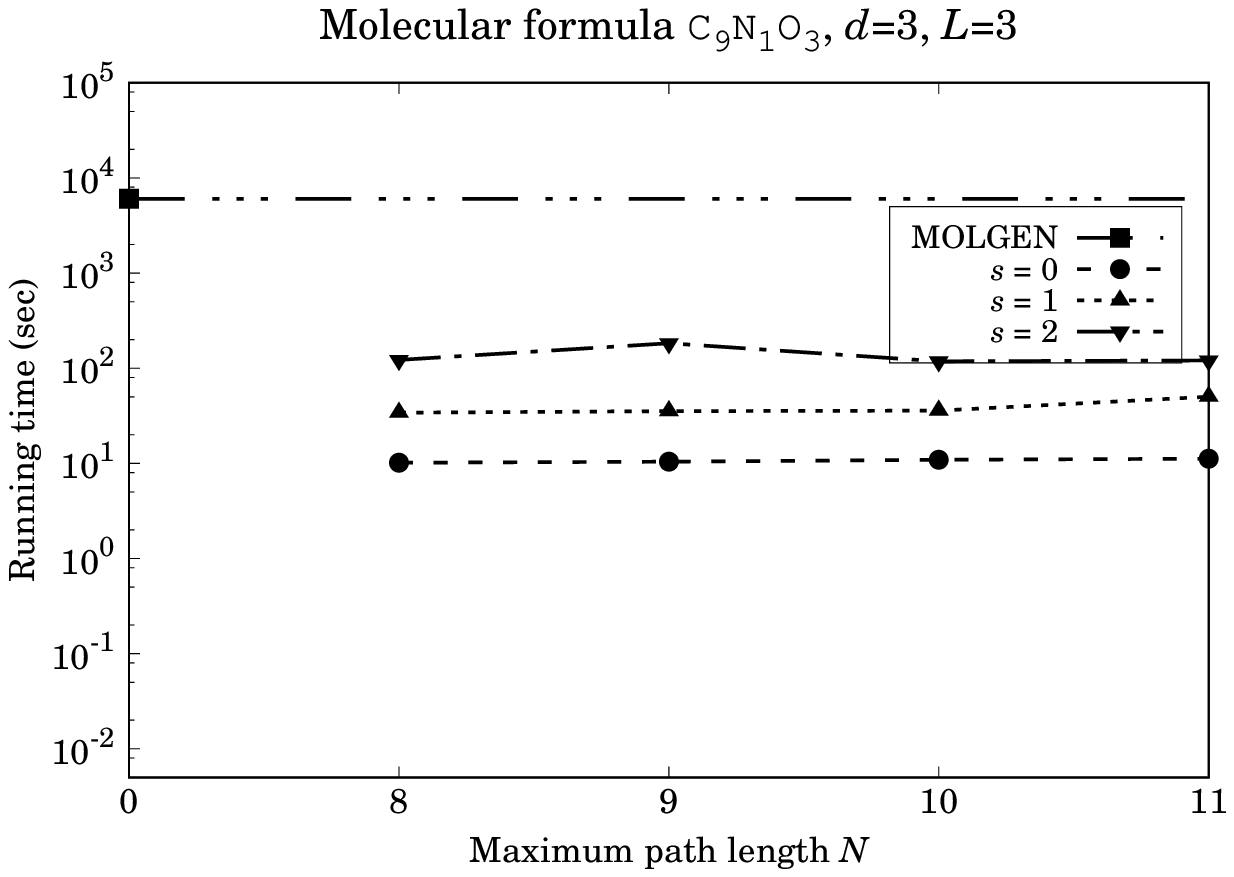}\\
      {\footnotesize (b)}\\
  \end{minipage} 
\hfill
  \begin{minipage}{0.45\textwidth}
   \centering
    \includegraphics[width=1.1\textwidth]{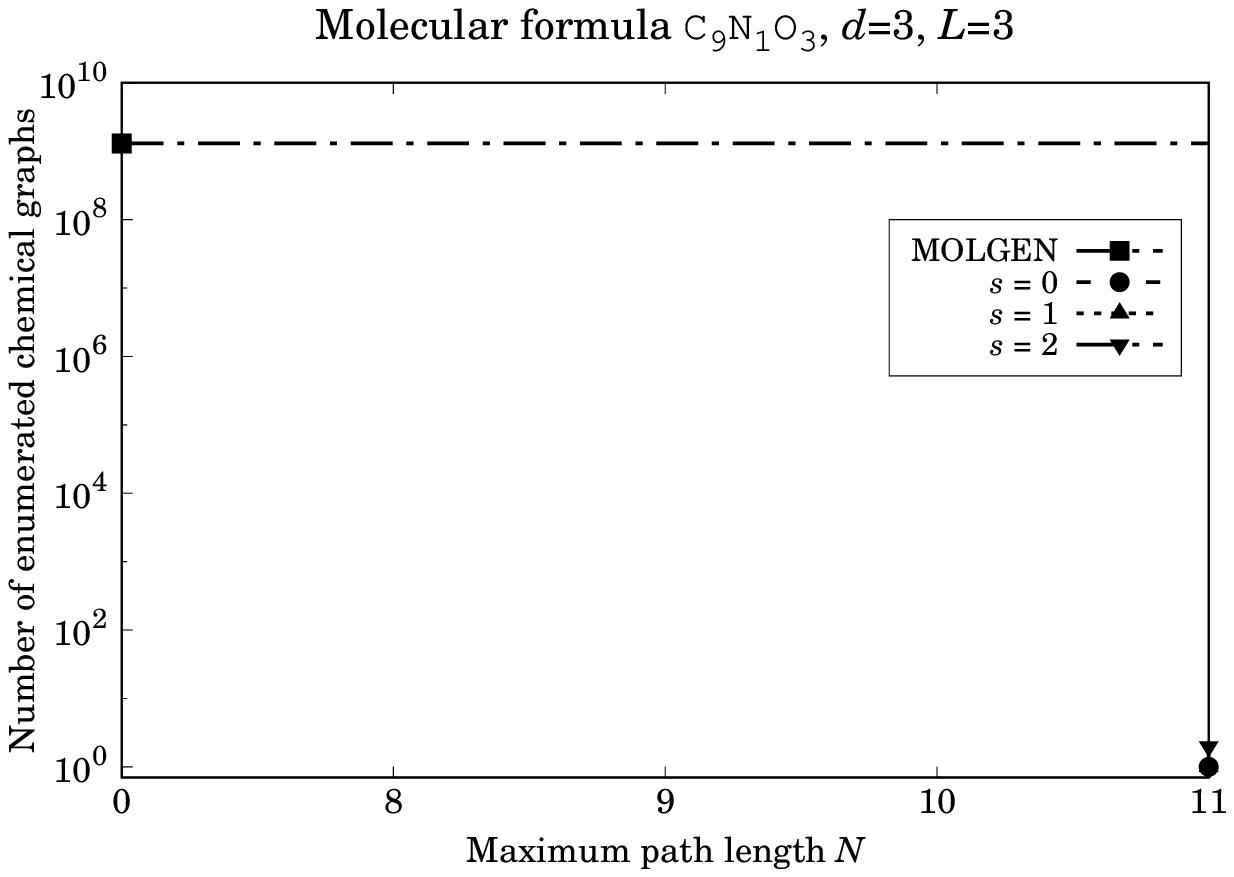}\\
    {\footnotesize (d)}\\
  \end{minipage} 
  \medskip
  
  \vspace{1cm}
  
  \caption{
    Plots showing the computation time 
    and number of chemical graphs enumerated by our algorithm
    for instance type EULF-$L$-P, as compared to MOLGEN.
    The sample structure from PubChem is with CID~131152558, 
    molecular formula {\tt C$_9$N$_1$O$_3$}, 
    and maximum bond multiplicity~$d=3$.
    (a), (b)~Running time;
    (c), (d)~Number of enumerated chemical graphs
    (our algorithm detects that there are no chemical graphs that
    satisfy the given path frequency specification).
  }
 \label{fig:result_graphs_2.2}
 \end{figure}

  \begin{figure}[!ht]
  \begin{minipage}{0.45\textwidth}
   \centering
   \includegraphics[width=1.1\textwidth]{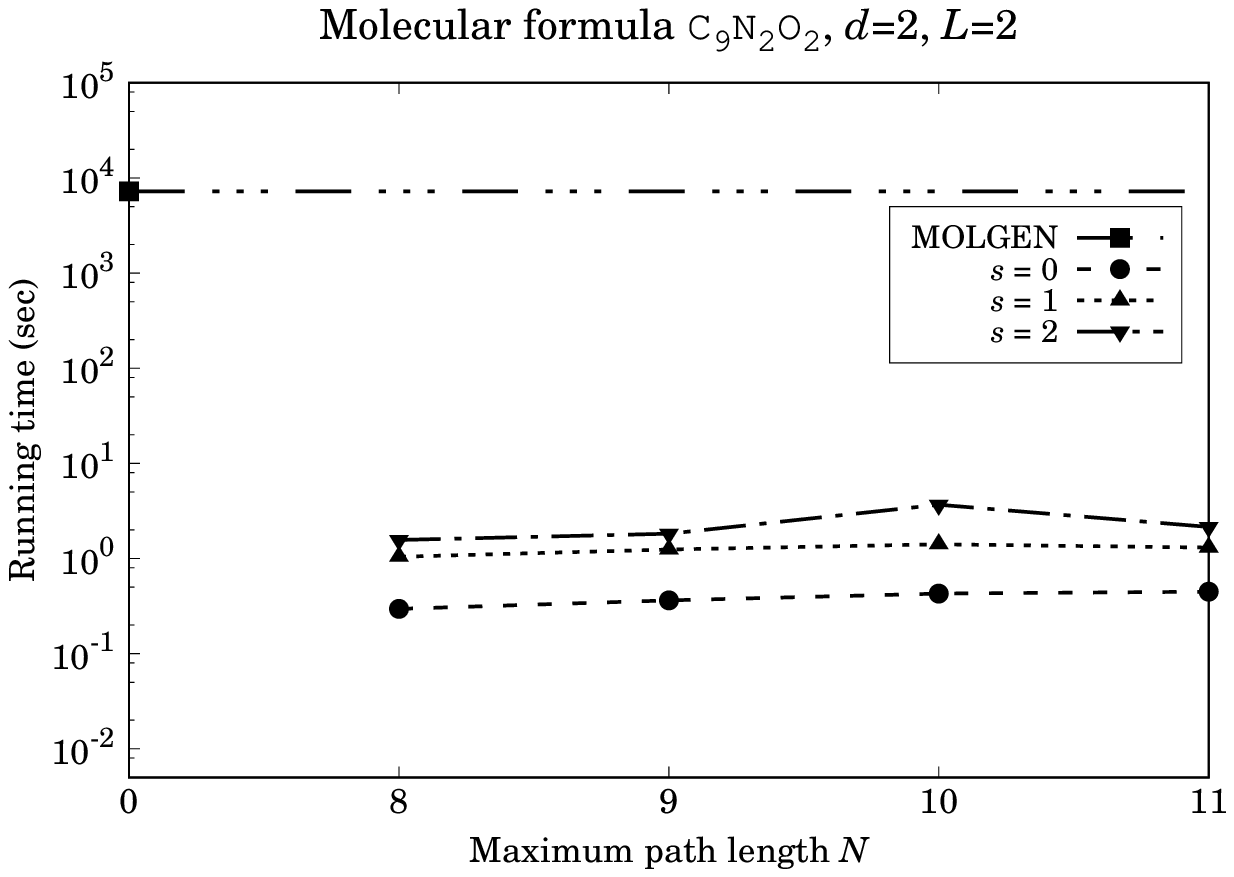}\\
   {\footnotesize (a)}\\
  \end{minipage}
\hfill
  \begin{minipage}{0.45\textwidth}
   \centering
   \includegraphics[width=1.1\textwidth]{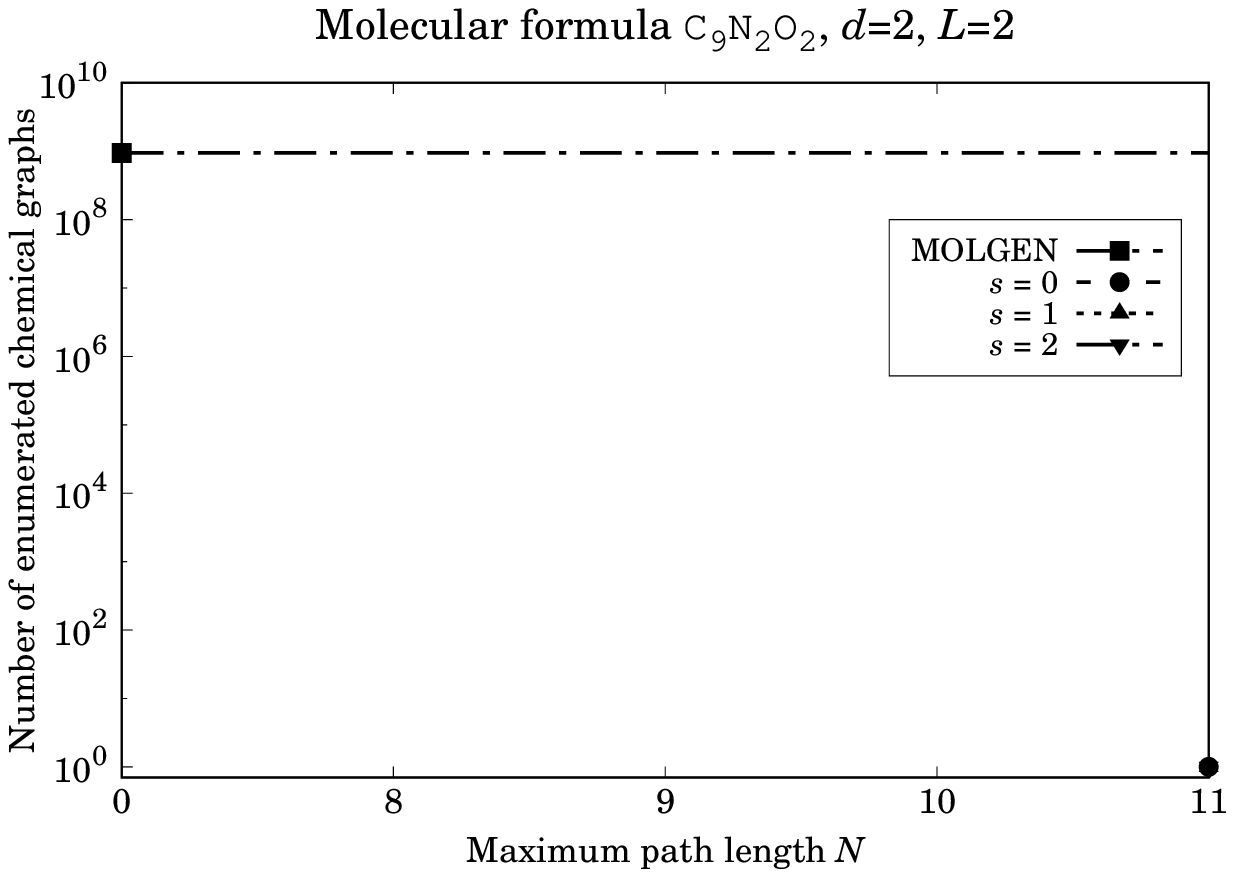}\\
   {\footnotesize (c)}\\
  \end{minipage} 
  \medskip

  \begin{minipage}{0.45\textwidth}
   \centering
      \includegraphics[width=1.1\textwidth]{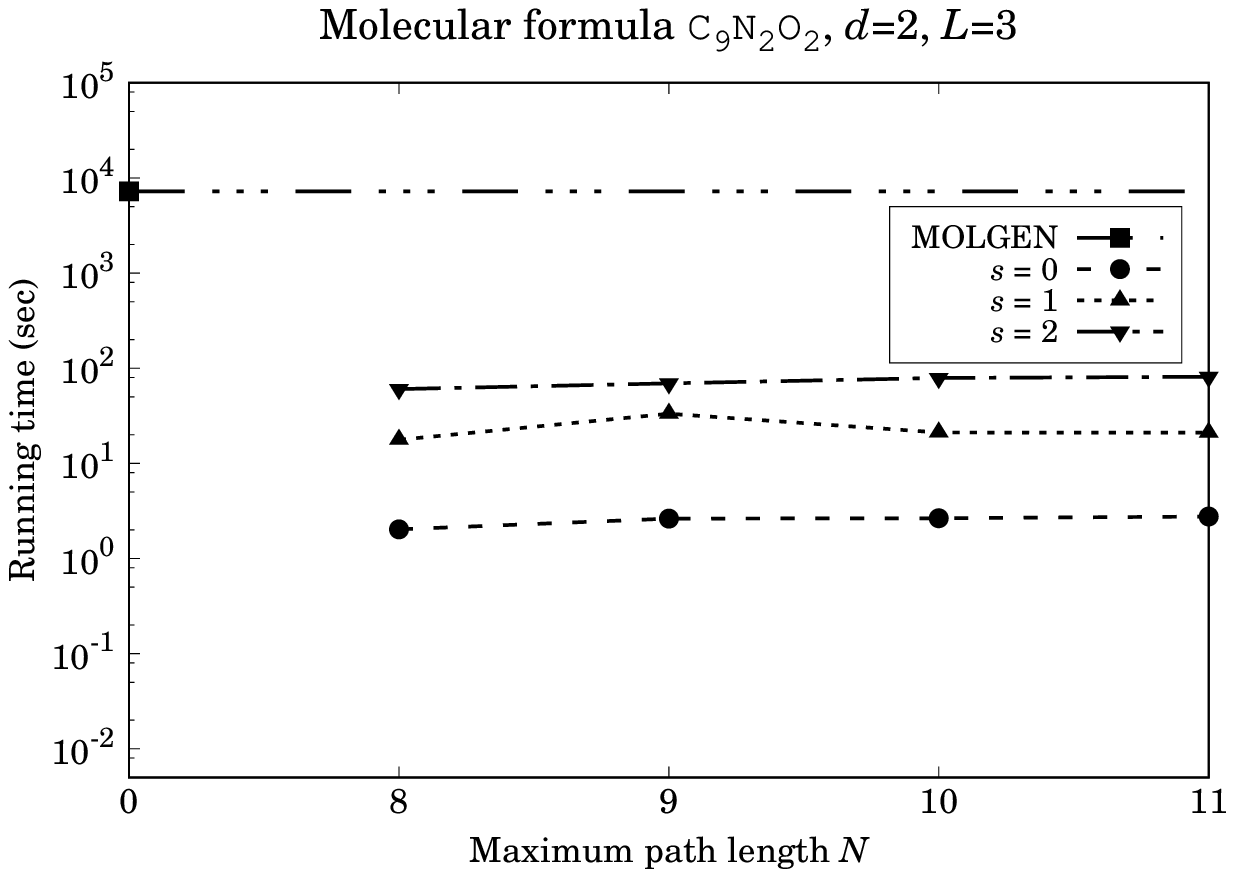}\\
      {\footnotesize (b)}\\
  \end{minipage} 
\hfill
  \begin{minipage}{0.45\textwidth}
   \centering
    \includegraphics[width=1.1\textwidth]{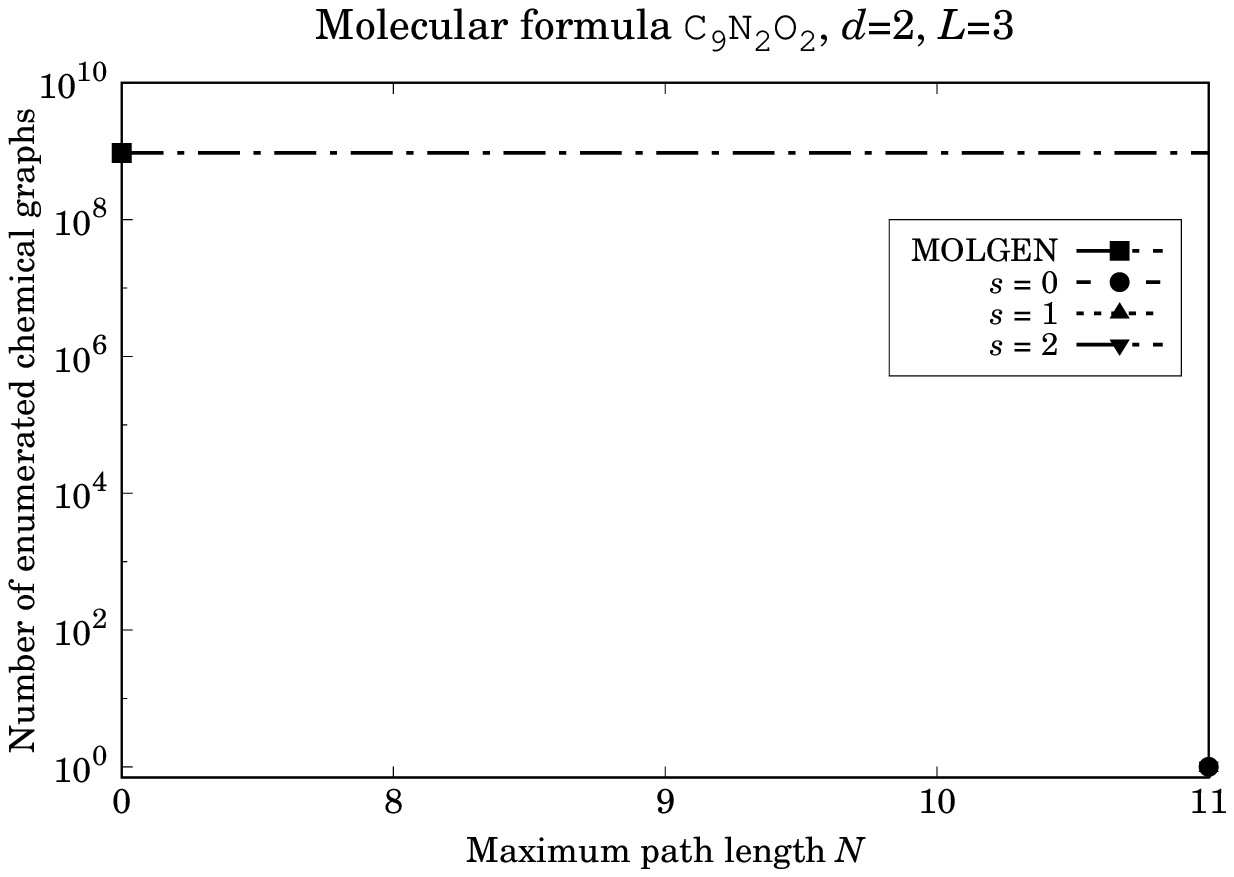}\\
    {\footnotesize (d)}\\
  \end{minipage} 
  \medskip
  
  \vspace{1cm}
  
  \caption{
    Plots showing the computation time 
    and number of chemical graphs enumerated by our algorithm
    for instance type EULF-$L$-P, as compared to MOLGEN.
    The sample structure from PubChem is with CID~742733, 
    molecular formula {\tt C$_9$N$_2$O$_2$}, 
    and maximum bond multiplicity~$d=2$.
    (a), (b)~Running time;
    (c), (d)~Number of enumerated chemical graphs
    (our algorithm detects that there are no chemical graphs that
    satisfy the given path frequency specification).
  }
 \label{fig:result_graphs_3.2}
 \end{figure}

  \begin{figure}[!ht]
  \begin{minipage}{0.45\textwidth}
   \centering
   \includegraphics[width=1.1\textwidth]{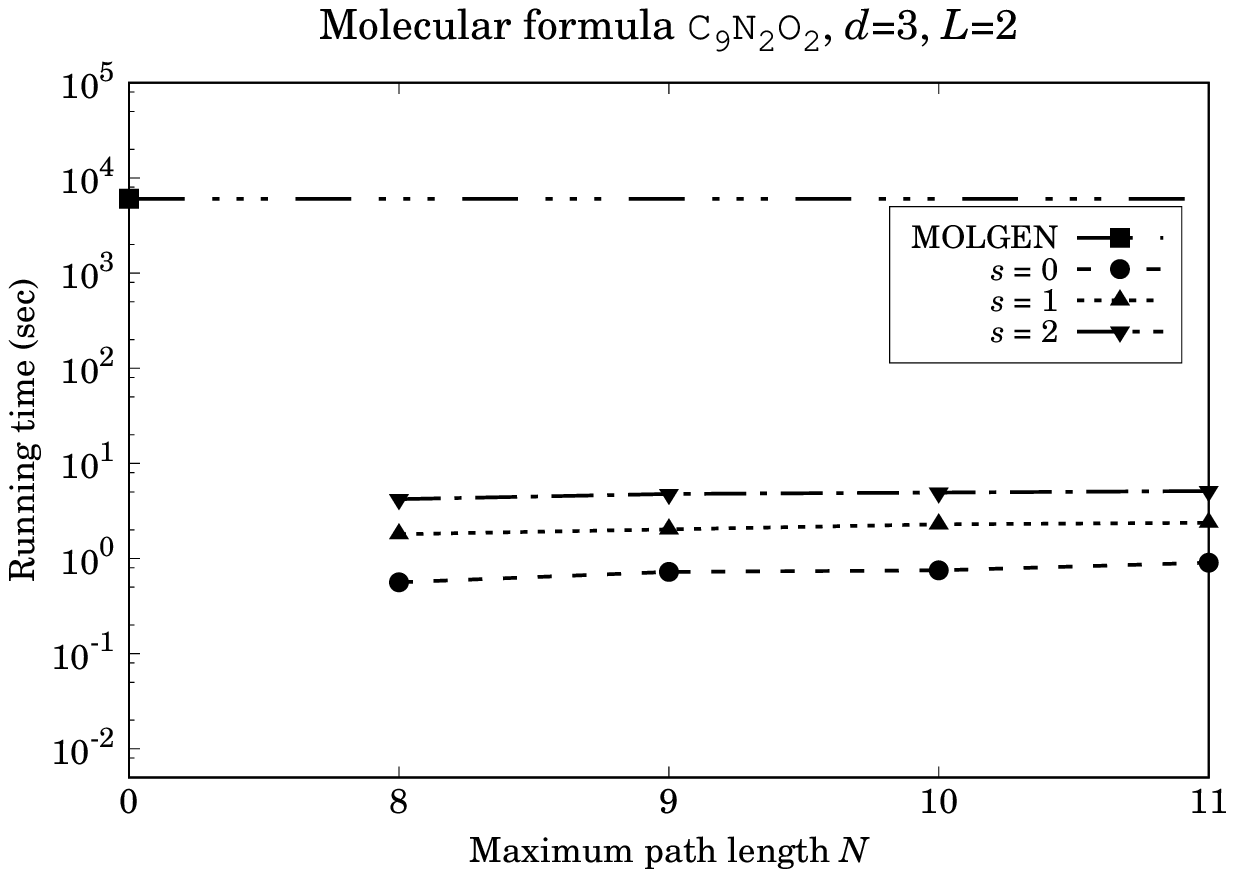}\\
   {\footnotesize (a)}\\
  \end{minipage}
\hfill
  \begin{minipage}{0.45\textwidth}
   \centering
   \includegraphics[width=1.1\textwidth]{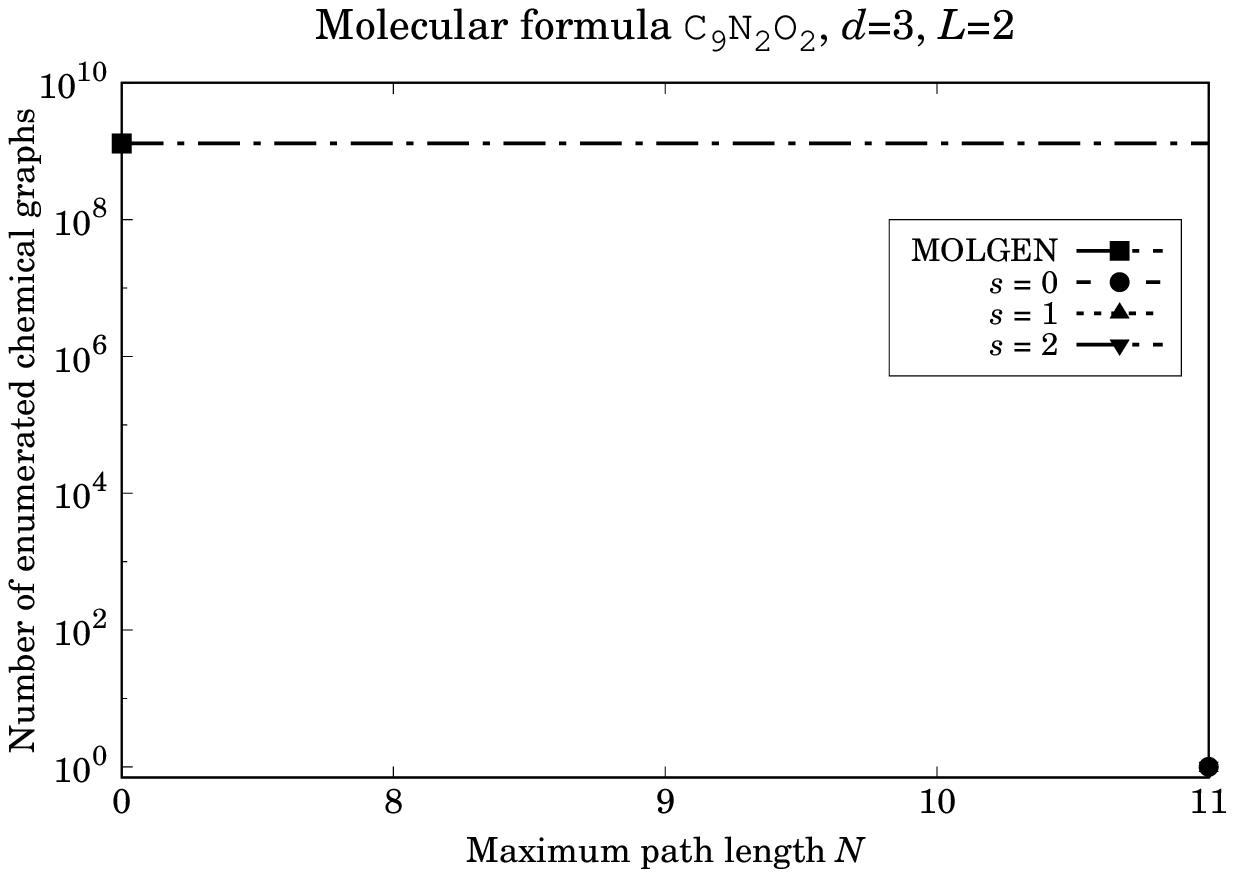}\\
   {\footnotesize (c)}\\
  \end{minipage} 
  \medskip

  \begin{minipage}{0.45\textwidth}
   \centering
      \includegraphics[width=1.1\textwidth]{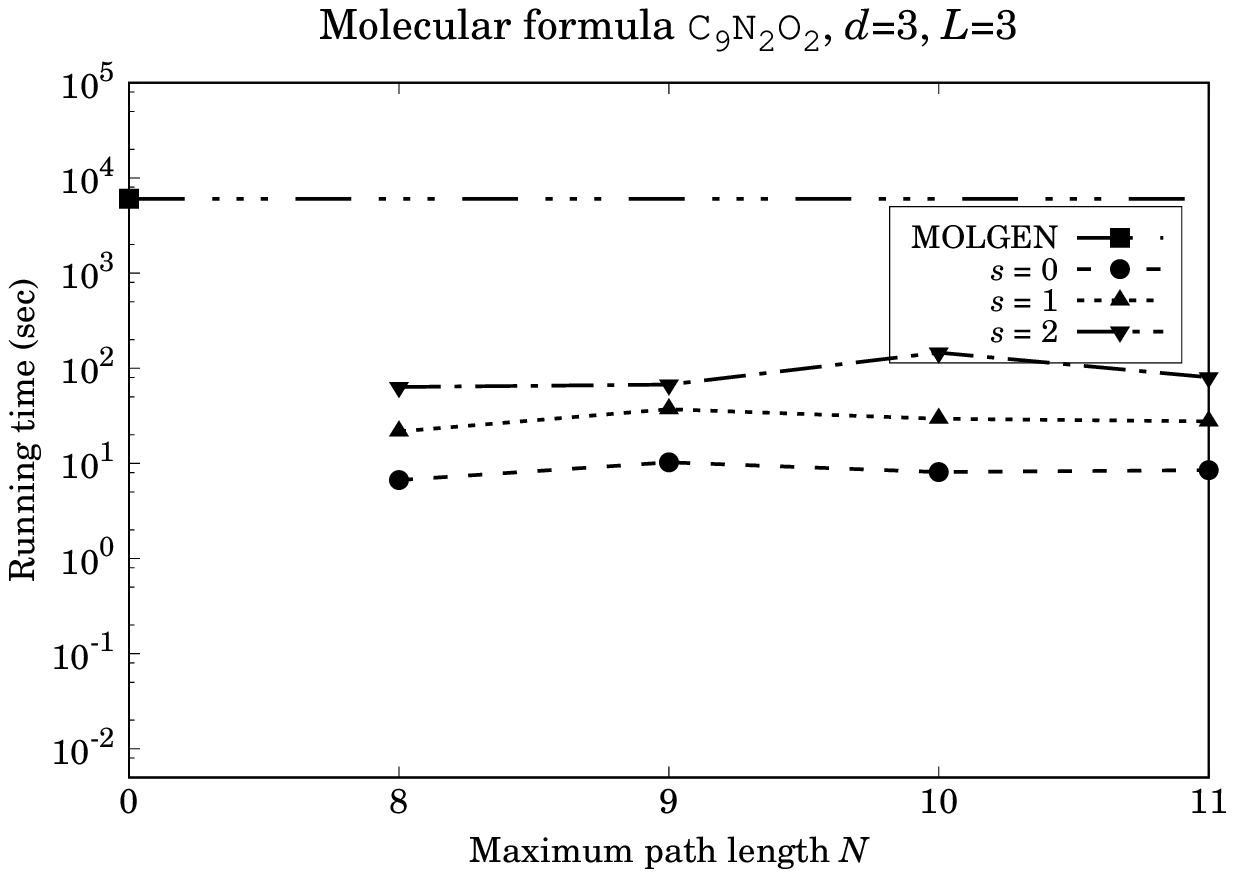}\\
      {\footnotesize (b)}\\
  \end{minipage} 
\hfill
  \begin{minipage}{0.45\textwidth}
   \centering
    \includegraphics[width=1.1\textwidth]{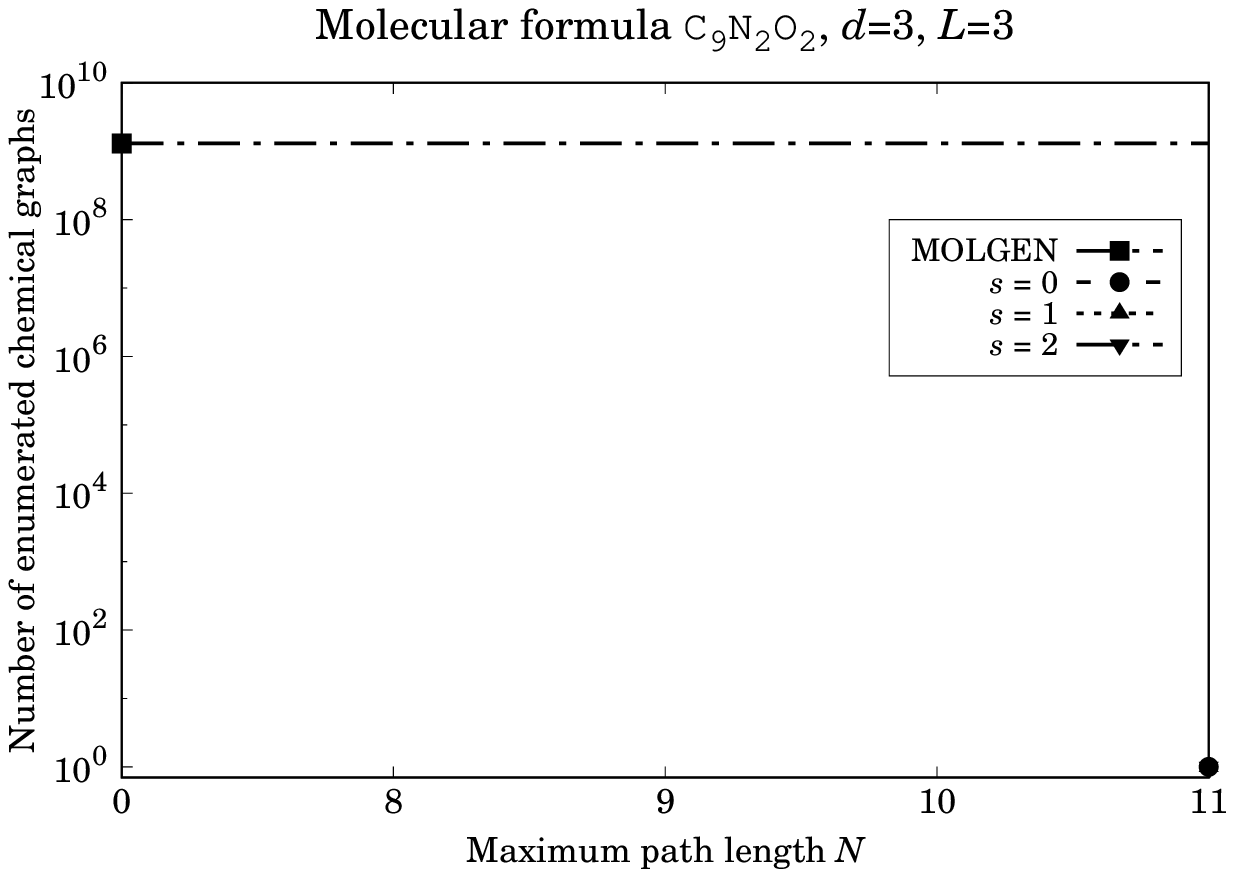}\\
    {\footnotesize (d)}\\
  \end{minipage} 
  \medskip
  
  \vspace{1cm}
  
  \caption{
    Plots showing the computation time 
    and number of chemical graphs enumerated by our algorithm
    for instance type EULF-$L$-P, as compared to MOLGEN.
    The sample structure from PubChem is with CID~23461643, 
    molecular formula {\tt C$_9$N$_2$O$_2$}, 
    and maximum bond multiplicity~$d=3$.
    (a), (b)~Running time;
    (c), (d)~Number of enumerated chemical graphs
    (our algorithm detects that there are no chemical graphs that
    satisfy the given path frequency specification).
  }
 \label{fig:result_graphs_4.2}
 \end{figure}

  \begin{figure}[!ht]
  \begin{minipage}{0.45\textwidth}
   \centering
   \includegraphics[width=1.1\textwidth]{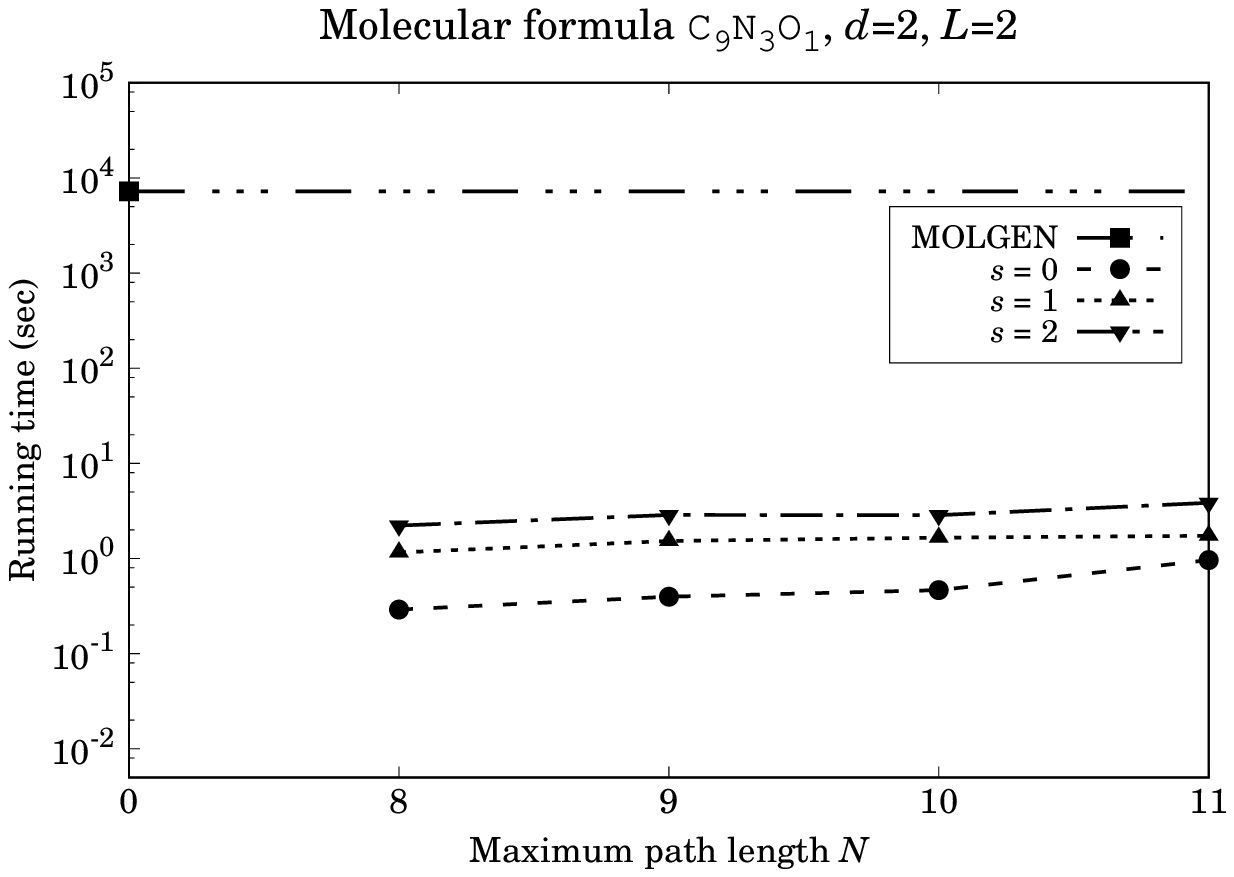}\\
   {\footnotesize (a)}\\
  \end{minipage}
\hfill
  \begin{minipage}{0.45\textwidth}
   \centering
   \includegraphics[width=1.1\textwidth]{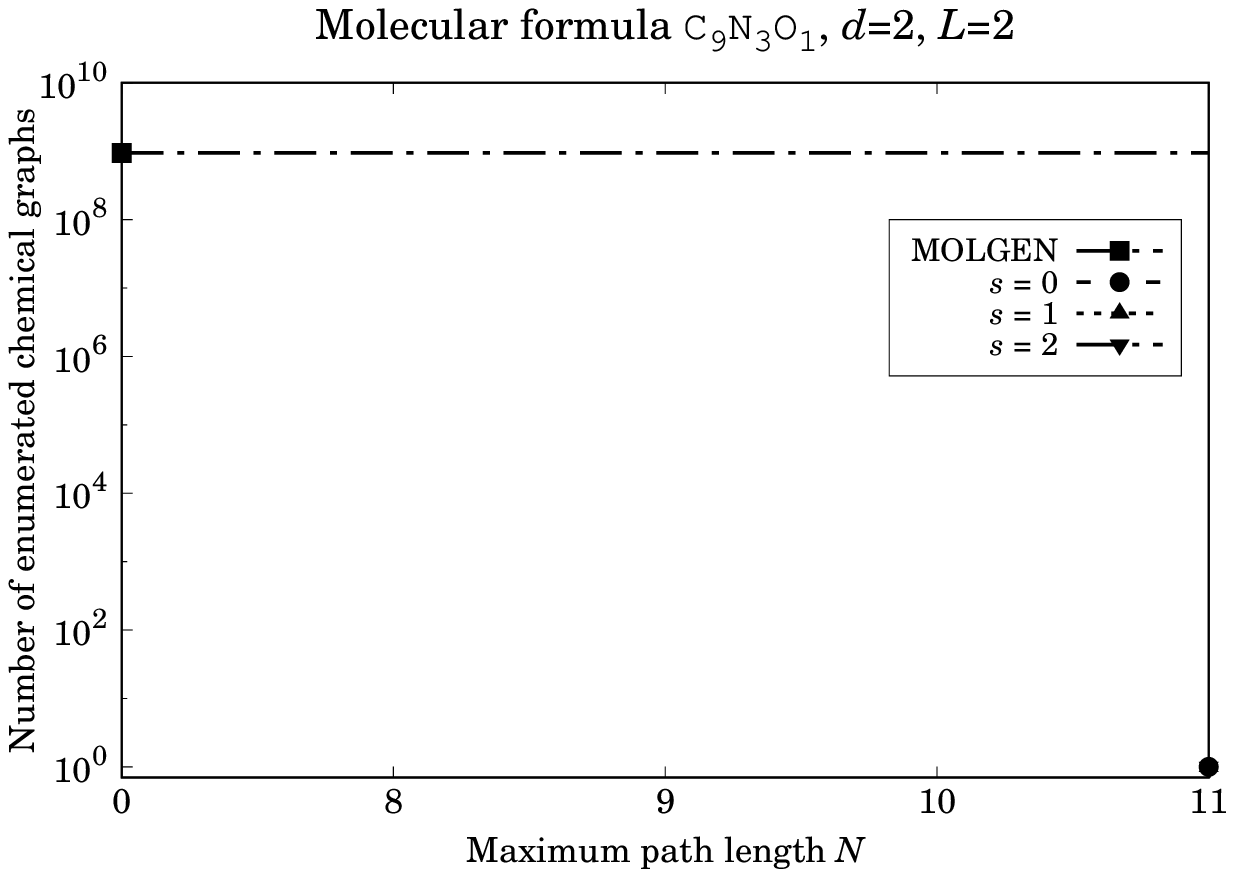}\\
   {\footnotesize (c)}\\
  \end{minipage} 
  \medskip

  \begin{minipage}{0.45\textwidth}
   \centering
      \includegraphics[width=1.1\textwidth]{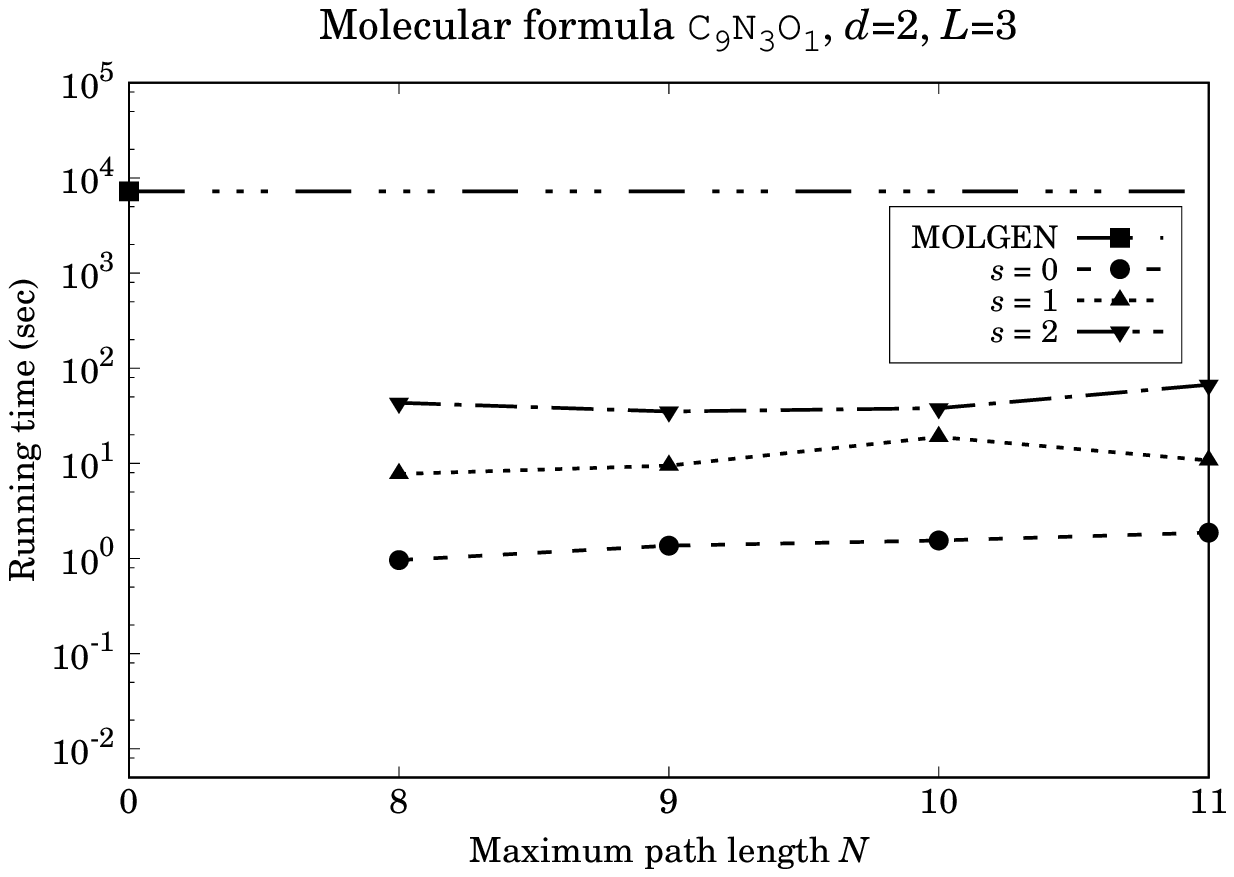}\\
      {\footnotesize (b)}\\
  \end{minipage} 
\hfill
  \begin{minipage}{0.45\textwidth}
   \centering
    \includegraphics[width=1.1\textwidth]{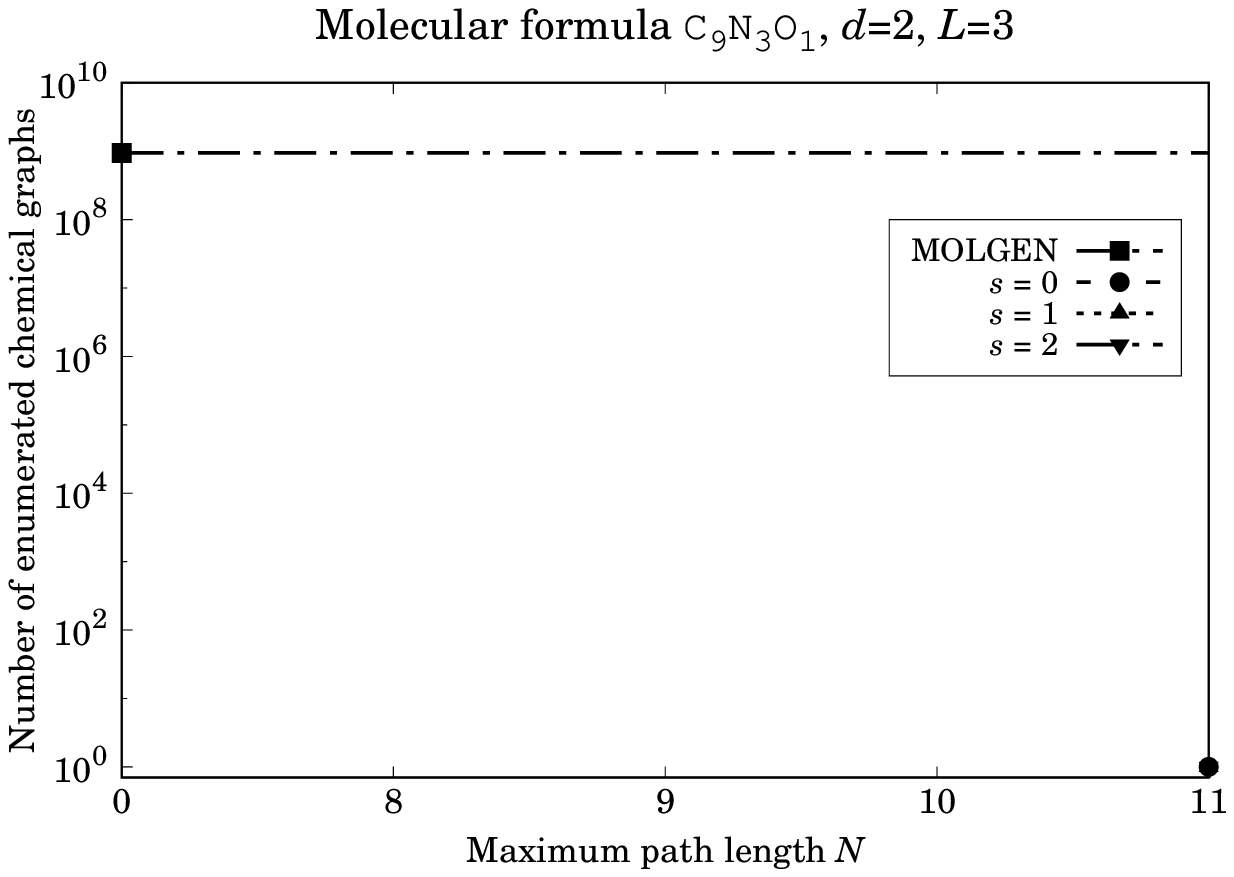}\\
    {\footnotesize (d)}\\
  \end{minipage} 
  \medskip
  
  \vspace{1cm}
  
  \caption{
    Plots showing the computation time 
    and number of chemical graphs enumerated by our algorithm
    for instance type EULF-$L$-P, as compared to MOLGEN.
    The sample structure from PubChem is with CID~10307896, 
    molecular formula {\tt C$_9$N$_3$O$_1$}, 
    and maximum bond multiplicity~$d=2$.
    (a), (b)~Running time;
    (c), (d)~Number of enumerated chemical graphs
    (our algorithm detects that there are no chemical graphs that
    satisfy the given path frequency specification).
  }
 \label{fig:result_graphs_5.2}
 \end{figure}

  \begin{figure}[!ht]
  \begin{minipage}{0.45\textwidth}
   \centering
   \includegraphics[width=1.1\textwidth]{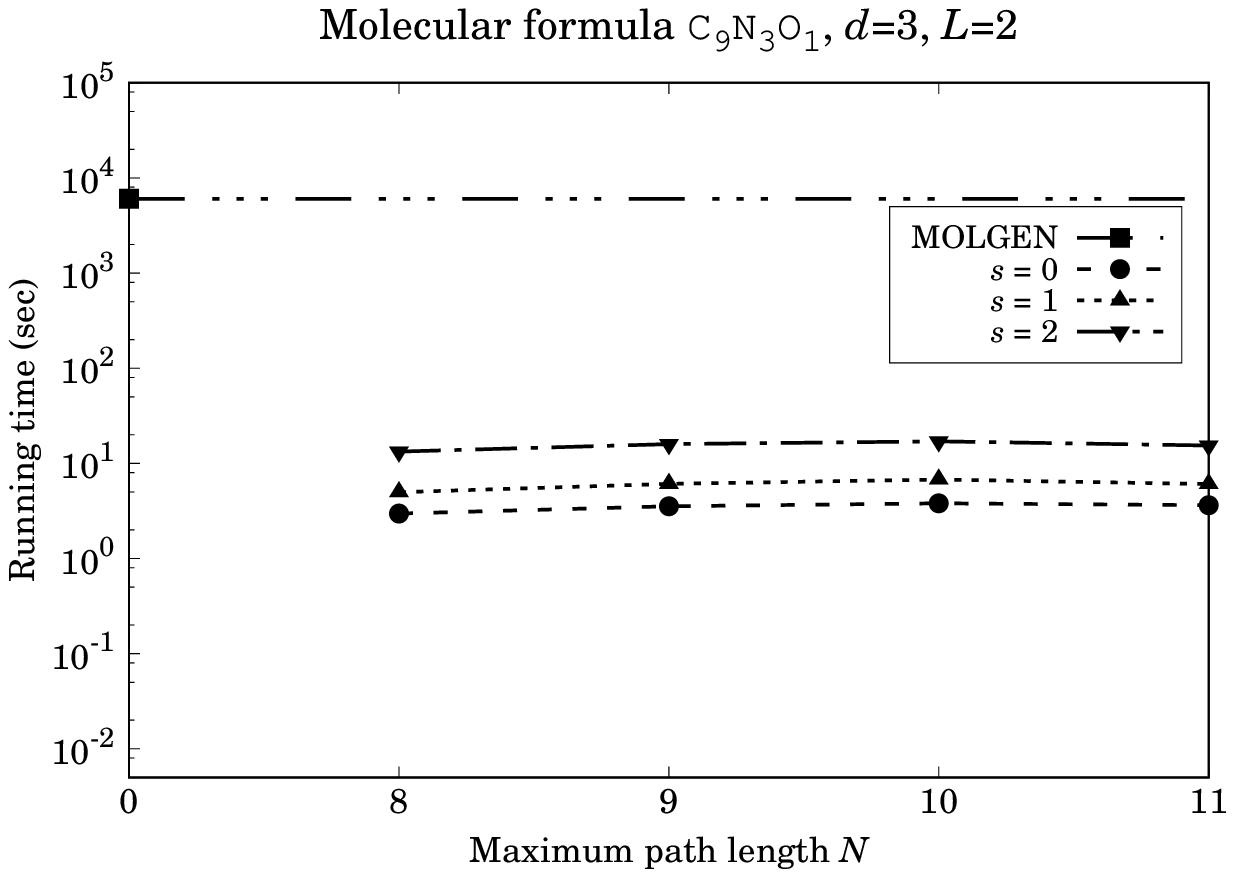}\\
   {\footnotesize (a)}\\
  \end{minipage}
\hfill
  \begin{minipage}{0.45\textwidth}
   \centering
   \includegraphics[width=1.1\textwidth]{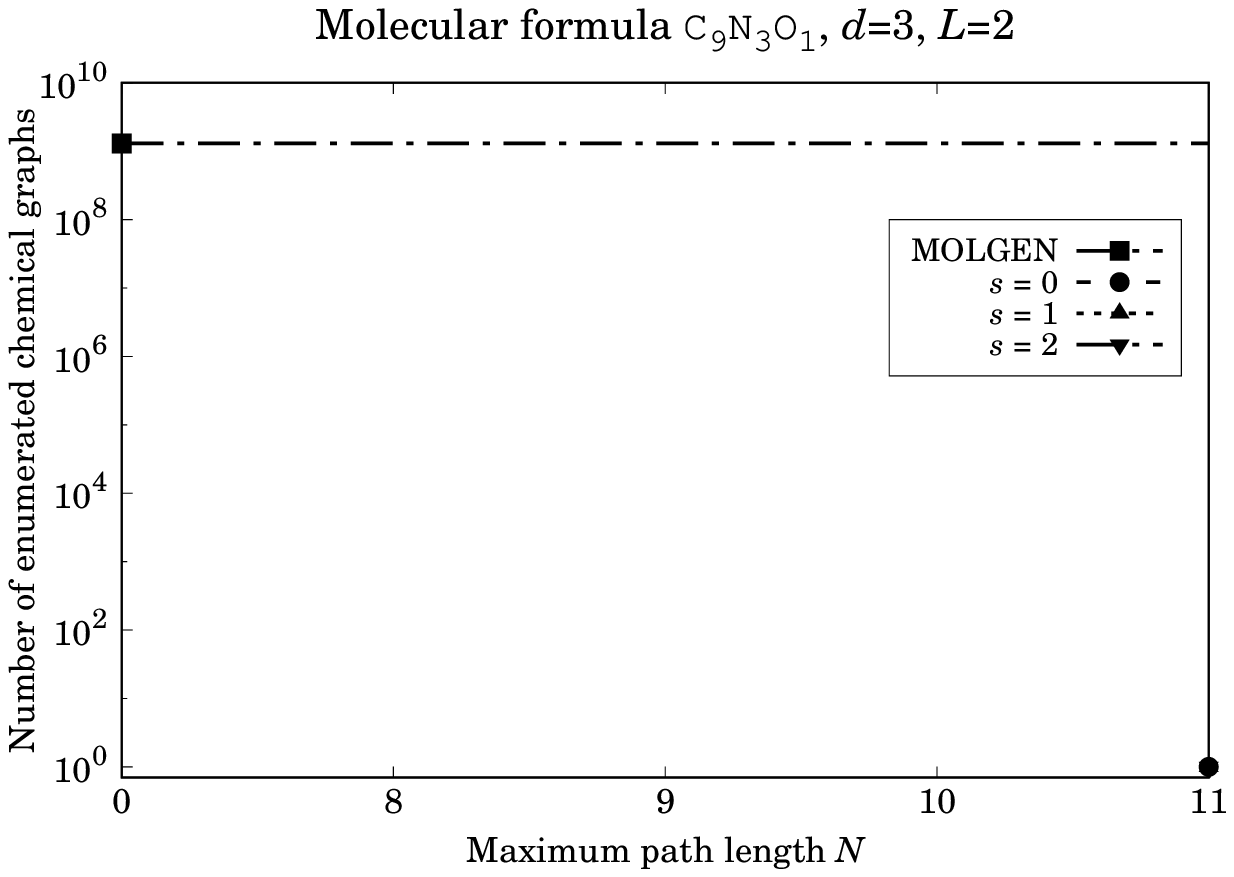}\\
   {\footnotesize (c)}\\
  \end{minipage} 
  \medskip

  \begin{minipage}{0.45\textwidth}
   \centering
      \includegraphics[width=1.1\textwidth]{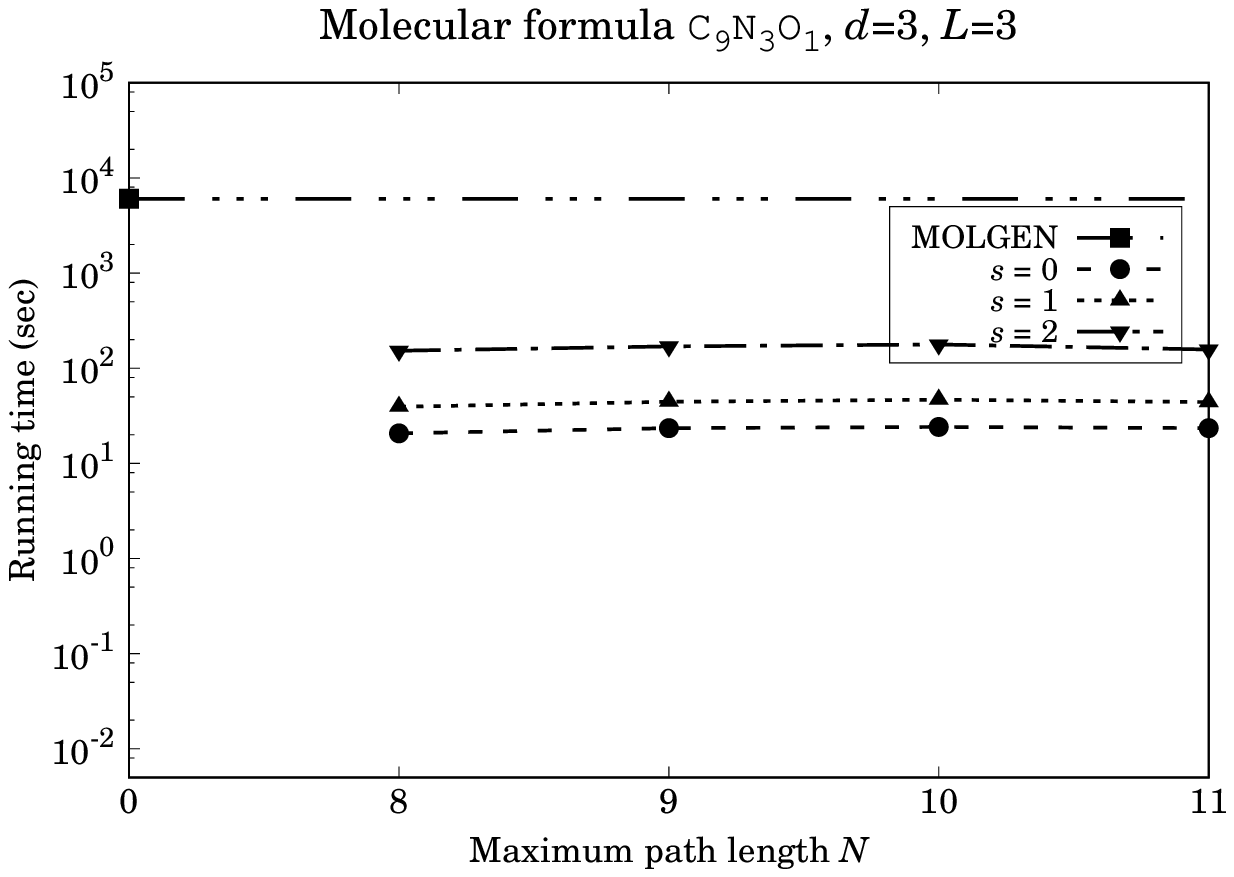}\\
      {\footnotesize (b)}\\
  \end{minipage} 
\hfill
  \begin{minipage}{0.45\textwidth}
   \centering
    \includegraphics[width=1.1\textwidth]{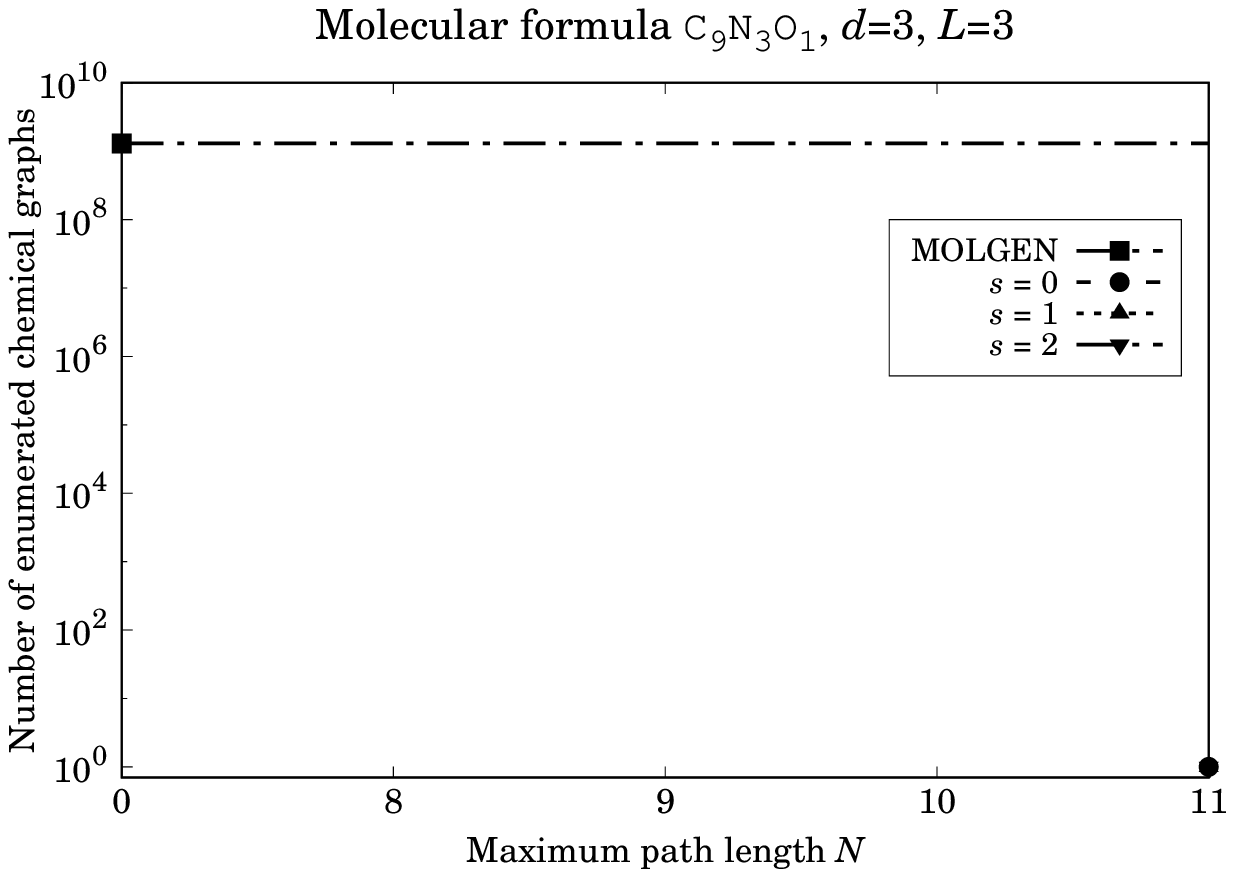}\\
    {\footnotesize (d)}\\
  \end{minipage} 
  \medskip
  
  \vspace{1cm}
  
  \caption{
    Plots showing the computation time 
    and number of chemical graphs enumerated by our algorithm
    for instance type EULF-$L$-P, as compared to MOLGEN.
    The sample structure from PubChem is with CID~10307899, 
    molecular formula {\tt C$_9$N$_3$O$_1$}, 
    and maximum bond multiplicity~$d=3$.
    (a), (b)~Running time;
    (c), (d)~Number of enumerated chemical graphs.
  }
 \label{fig:result_graphs_6.2}
 \end{figure}

\section{Conclusions and Future Work} \label{sec:conclusion}

We formulated the problem of enumerating bi-block 
2-augmented tree structures from a given set of feature vectors.
The set of feature vectors in turn, is specified by a pair of upper and lower bounds on feature vectors.
We designed a branch-and-bound algorithm for the problem 
partially following the ideas from the method of
enumerating monocyclic chemical graphs due to Suzuki~{\em et al.}~\cite{Suzuki14}.

Our experimental results show that our algorithm with a restriction on path
frequencies on paths of length greater than one offers a big advantage over
approaches that enumerate chemical graphs solely based on chemical formula,
in the sense that much fewer graphs of interest are enumerated, in drastically
shorter time.
Nevertheless, for the second of our problem settings, we observe from
computational experiments that the algorithm still takes a considerable
time even when there are no chemical graphs that satisfy the given
path frequency specification.
It would be of interest to develop conditions that detect 
such situation much quicker.

Other possible directions for future work is to extend our
 approach to chemical graphs with structure beyond 2-augmented trees.


\end{document}